\documentclass[11pt,letterpaper]{article}
\usepackage{fullpage}
\usepackage[round]{natbib}
\usepackage{amsmath,amsfonts,amssymb,amsthm,  mathrsfs,mathtools}
\usepackage{tikz}
\usepackage[margin=1in]{geometry}
\usepackage{comment}
\usepackage{multirow} 
\usepackage[hyperindex,breaklinks]{hyperref}
\hypersetup{
  colorlinks,
  citecolor=violet,
  linkcolor=blue,
  urlcolor=blue}
\usepackage{dsfont}
\usepackage[most]{tcolorbox}
\numberwithin{equation}{section}   
\usepackage{booktabs} 
\usepackage{subfigure} 
\usepackage[ruled,noend]{algorithm2e} 

\SetAlFnt{\small}
\SetAlCapFnt{\small}
\SetAlCapNameFnt{\small}
\SetAlCapHSkip{0pt}
\IncMargin{-\parindent}

\usepackage[thinc]{esdiff}

\numberwithin{equation}{section}
\usepackage{packages/macros-common}
\usepackage{packages/macros-project-specific}
\usepackage{packages/mytheorems}

\usepackage[suppress]{color-edits}
\addauthor{ss}{magenta} 
\addauthor{NG}{brown} 

\title{Optimal Contest Beyond Convexity}

  \newcommand{\institution}[1]{#1}
  \newcommand{\email}[1]{Email: \texttt{#1}}
  \newcommand{\affiliation}{\thanks}
\author{
 {Negin Golrezaei
 \affiliation{
   \institution{Massachusetts Institute of Technology, Sloan School of Management. }
 \email{golrezae@mit.edu}
 }}
 \and
 {MohammadTaghi Hajiaghayi
 \affiliation{
   \institution{University of Maryland, Computer Science.}
 \email{hajiagha@umd.edu}
 }}
 \and
 {Suho Shin
 \affiliation{
   \institution{University of Maryland, Computer Science.}
 \email{suhoshin@umd.edu}
 }}
}

\begin{document}
\date{}
\maketitle

\begin{abstract}
    In the \emph{contest design} problem, initiated by Lazear and Rosen (JPE'81), there are $n$ strategic contestants, each of whom decides an effort level. A contest designer with a fixed budget must then design a mechanism that allocates a prize $p_i$ to the $i$-th rank based on the outcome, to incentivize contestants to exert higher costly efforts and induce high-quality outcomes.
    The contest design has received substantial attention in both the economics literature, \eg Moldovanu and Sivan (AER'01) and Fang, Noe, and Strack (JPE'20) and computer science literature, \eg Chawla, Hartline and Sivan (SODA'12) and Ghosh and Kleinberg (EC'14), over the past few decades, with plentiful applications to research contests, crowdsourcing platforms, and modern online content platforms.
    Despite the significant interest, there have been only a few results on optimal mechanisms, even in simplistic settings such as convex (concave) objective functions.

    In this paper, we significantly deepen our understanding of optimal mechanisms under general settings by considering \emph{nonconvex high-dimensional} objective functions in contestants' qualities.
    Notably, our results accommodate the following objective functions: (i) any convex combination of user welfare (motivated by recommender systems) and the average quality of contestants that is neither convex nor concave, (ii) \emph{arbitrary posynomials} over quality. 
    In particular, these subsume \emph{classic measures} in mechanism design such as social welfare, order statistics, and (inverse) S-shaped functions, which have received little or no attention in the contest literature to the best of our knowledge.

    Surprisingly, across all these regimes, we show that the optimal mechanism is \emph{highly structured}:
    it allocates potentially higher prize to the first-ranked contestant, zero to the last-ranked one, and equal prizes to the all intermediate contestants, \ie $p_1 \ge p_2 = \ldots = p_{n-1} \ge p_n = 0$.
    In particular, in some special cases, we observe a stark \emph{phase transition} between two extreme mechanisms: (i) $\HM$ policy ($p_1 = 1, p_2 = \ldots = p_n = 0$) and (ii) $\UNI$ policy ($p_1 = \ldots = p_{n-1}=1/(n-1), p_n = 0$) depending on the objective and cost function, cementing and unifying evidences witnessed in the literature.
    \ssedit{More importantly, thanks to the structural characterization, we obtain a \emph{fully polynomial-time approximation scheme} given a value oracle.}
    
    Our technical results rely on Schur-convexity (or concavity) of Bernstein basis polynomial–weighted functions, total positivity and variation diminishing property.
    \ssedit{En route to our results, we obtain a surprising reduction from a structured high-dimensional nonconvex optimization to a single-dimensional optimization by connecting the shape of the gradient sequences of the objective function to the number of transition points in optimum, which might be of independent interest.}

    \sscomment{Another version for last sentence: Our analysis yields a general dimension-collapse theorem for order-constrained optimization with totally positive gradients: under mild quasiconvexity, any global optimum must lie on a two-level chain, reducing an $n$-dimensional non-convex problem to a single-variable optimization, which might be of independent interest.}

\end{abstract}

\newpage
\tableofcontents
\newpage

\section{Introduction}
The theory of \emph{contest design}, introduced nearly half a century ago by~\cite{lazear1981rank}, has been one of the most central problems in game theory and mechanism design, attracting attention traditionally from economists~\citep{glazer1988optimal,barut1998symmetric,moldovanu2008optimal,fang2020turning} to more recent interest from computer science~\citep{ghosh2011game,ghosh2016optimal,greenwald2018simple,chawla2019optimal}, with numerous applications including research contests, sports tournaments, crowdsourcing, and online content or question-answering platforms.
In the contest design problem, there are $n$ strategic contestants, each of whom decides the amount of costly effort to exert in the competition, which maps to their quality $q$ given a cost function $c(q)$ to maximize their payoff.
The contest designer aims to design a rank-based mechanism that maps each rank $i$ (based on quality) of each contestant to an allocated budget $p_i$ for $i \in [n]$, given a fixed total (normalized) budget $\sum_{i \in [n]} p_i = 1$.


As an illustrative example, consider the problem of online content platforms or social media that design a recommendation algorithm to decide which content to deliver to each user.
The production of content on these platforms is typically handled by individuals known as content producers.
However, producing high-quality content requires significant time and effort. To keep content producers engaged in the ecosystem, platforms provide monetary compensation to encourage the production of high-quality content.
Crucially, the incentives for content producers are closely tied to the recommendation algorithm itself, as greater visibility to users translates to higher potential earnings.
The platform not only needs to maintain high average quality of the content registered on the platform, but also guarantee user welfare, \ie the quality of content that is actually shown to the user.\footnote{Looking forward, a convex combination of these two metrics turns out to be neither convex nor concave, which primarily motivates our study of such generalized objective functions. Still, our results accommodate far more than this specific application, as will be discussed shortly.}



The seminal result by~\cite{glazer1988optimal} reveals that the optimal mechanism uniformly distributes the budget among the contestants except the last-ranked one (henceforth, the $\UNI$ policy) when the cost function is linear and the objective function is the average quality among the contestants.
Since then, despite significant interest, there have been only a few known results even under simplistic cases~\citep{barut1998symmetric,moldovanu2008optimal,greenwald2018simple,fang2020turning}, stating that either of the two extreme mechanisms, $\UNI$ or $\HM$ (which allocates all the budget to the highest-ranked contestant), is optimal in some special cases.
In particular, most of these papers assume a linear objective function to maximize the average quality of the contestants~\citep{glazer1988optimal,moldovanu2008optimal,greenwald2018simple,ghosh2016optimal}, or a convex (concave) function over the quality even in the most general case, with a correspondingly restrictive class of cost functions~\citep{fang2020turning}.
In practice, however, many scenarios, such as recommender system contests, involve objective functions that are much more complicated.


In this context, the following major open questions remain unresolved, with no partial evidence having emerged for decades.
\definecolor{mycolor}{rgb}{0.9, 0.90, 0.9}
\begin{tcolorbox}[colback=mycolor,colframe=gray!75!black,colframe=white]
   \emph{\noindent
   \textbf{Q1:} Would $\HM$ or $\UNI$ still be optimal beyond convex objectives?
   \\
   \textbf{Q2:} If not, what would the optimal mechanism be, and is it even structured?
   }
\end{tcolorbox}
In this paper, we answer these questions both negatively and affirmatively.  
First, we show that neither policy is optimal in general, but that a mechanism interpolating between the two is optimal, \ie one that assigns potentially high prize to the first rank (similar to $\HM$) but uniformly distributes the remaining prizes to intermediate ranks except the last (similar $\UNI$)—a mechanism that has never been studied in the literature.  
On the other hand, we show that either $\HM$ or $\UNI$ is optimal in certain special cases, as evidenced by prior work, and we provide an almost complete characterization of the conditions under which such results hold, thereby filling missing gaps in the literature.



To highlight our main technical contributions:
\begin{enumerate}
    \item We first establish the existence of a symmetric mixed Nash equilibrium (MNE) for nontrivial policies\footnote{A policy is nontrivial if it is not $p_1 = \ldots = p_n$, which induces a trivial pure Nash equilibrium of zero effort.} using topological arguments on discontinuous games from~\cite{reny1999existence}.
    \ssedit{We further prove that the symmetric MNE is unique.}
    \item Towards characterizing the optimal structure, we show that either the $\HM$ policy (with $p_1 = 1$ and $p_2 = \ldots = p_n = 0$) or the $\UNI$ policy (with $p_1 = \ldots = p_{n-1} = 1/(n-1)$ and $p_n = 0$) is optimal if the objective function is either user welfare (motivated by recommender systems) or the average quality of the contestants, exhibiting a stark \emph{phase transition} depending on whether the cost function is convex or concave. In fact, we provide much more general conditions for such optimality, complementing partial evidence~\citep{glazer1988optimal,greenwald2018simple,fang2020turning} in the literature.
    \item We extend our analysis to the case where the objective function is an arbitrary convex combination of user welfare and average quality, leading to an optimization problem involving a sum of convex and concave components, which typically does not admit a closed-form solution.
    Nevertheless, we identify the \emph{structure of the optimal policy} to be $p_1 \ge p_2 = \ldots = p_{n-1} \ge p_n = 0$.
    Surprisingly, this result extends to significantly more general objective functions--arbitrary \emph{posynomial}s over qualities--subsuming classic measures in mechanism design such as social welfare, order statistics, (inverse) S-shaped functions, and exponential utility, which have received little or no attention in the literature, to the extent of our knowledge.
    \item \ssedit{Furthermore, thanks to the structural characterization, we obtain a \emph{fully polynomial-time approximation scheme} to compute the optimal policy, assuming the constant parameter in the cost function.\footnote{This is reasonable in practice, since the only scaling parameter in real-world scenario would be the number of contestants $n$ whereas the cost function is fixed.} Interestingly, our techniques yield a \emph{general reduction} from a structured high-dimensional nonconvex optimization to a single-dimensional optimization problem, which implies potential polynomial-time solvability for a wide class of high-dimensional non-convex optimization.}
\end{enumerate}

In particular, we note that our results subsume prior results in comprehensive manner: linear cost~\citep{glazer1988optimal} and convex cost settings~\citep{greenwald2018simple}, average quality as an objective~\citep{barut1998symmetric}, restricted class of policies~\citep{moldovanu2008optimal}, and convex (concave) utility function with concave (convex) cost~\citep{fang2020turning}.
It may be also worth noting that our existential result uses topological arguments by~\cite{reny1999existence} rather than characterizing several analytical properties of symmetric profiles tailored to each setting~\citep{barut1998symmetric,fang2020turning}, which we believe may be of general interest to the literature, formalizing unstructured observations in the literature. 

Further, en route to our results on the structural characterization of $p_1 \ge p_2 = \ldots = p_{n-1} \ge p_n = 0$, we observe a surprising connection between the shape of the gradient sequence of the objective function and number of transition points in optimum for a constrained symmetric function optimization problem, which may be of independent interest to broader optimization community.

In what follows, we give an overview of our results and techniques.
Related works are discussed in Section~\ref{sec:rel}, and we present formal problem setup in Section~\ref{sec:model}, followed by main technical results.
Main technical results, from problem setup in Section~\ref{sec:model} to main results in Section~\ref{sec:opt-structure-general}, will be presented primarily based on the recommender systems application, and we discuss generalized objective functions and their implications on Section~\ref{sec:general}. All the proofs are deferred to appendices.

\subsection{Results and Techniques}
Before presenting a summary of our results and proof techniques, we briefly introduce our setting. More details can be found in Section~\ref{sec:model}.
Looking ahead, our characterization of the optimal structure is summarized in Table~\ref{tab:summary}.
For ease of exposition, we present our model and results based on the scenario of recommender systems, while noting that the results are broadly applicable to various settings studied in the literature, such as research contests and crowdsourcing contests.



There are $n$ content producers (\ie contestants), each of whom decides on a quality $q_i$ for their content (\ie outcome).
Producing content with quality $q$ requires a cost of $c(q) = q^{\beta}$ for $\beta > 0$.\footnote{All our results readily adapt to the linear transformation $c(q) = \gamma q^{\beta}$.}
The recommender system (RS, \ie a contest designer) designs a rank-based recommendation policy $\bp = (p_1,\ldots,p_n) \in \Delta$ where $\Delta = \{\bp \in \R_{\ge 0}^n: p_1 \ge \ldots \ge p_n \ge 0, \sum_{i=1}^n p_i = 1\}$.
Then, whenever a user joins the platform, the RS sorts the existing content in decreasing order of quality and recommends the $k$-th highest quality content with probability $p_k$,\footnote{From the general contest design problem's perspective, each $p_k$ can be viewed as a prize with normalized total budget. \ssedit{Thus, when we write \emph{probability}, this refers to \emph{prize} in the context of general contest design literature.}} \ssedit{where ties are broken uniformly at random.}

Producer $i$ receives a unit reward of $1$ whenever their content is recommended.
Hence, their expected reward under this policy when their content has rank $k$ is $p_k$.
The eventual payoff of producer $i$ is the expected reward minus the cost of producing content with the corresponding quality, \ie $p_k - c(q_i)$ if $q_i$ is ranked $k$-th given others' strategies.

The RS aims to maximize a convex combination of user welfare ($\cW(\bmu, \bp)$) and platform quality ($\cQ(\bmu, \bp)$) at a symmetric MNE $\bmu$. Here, user welfare denotes the expected quality of content shown to users, while platform quality refers to the average expected quality of all content submitted to the system.\footnote{See Section~\ref{sec:model} for formal definitions.}
Note that $\bmu$ represents the product distribution over producers’ mixed strategies, where each producer may randomize over content quality. 
The goal is to design a policy $\bp$ that maximizes the objective:
$
\alpha \cW(\bmu, \bp) + (1-\alpha) \cQ(\bmu, \bp).
$


We say that the RS is \emph{single-minded} if $\alpha = 0$ or $1$, \ie it is interested solely in maximizing either user welfare or platform quality, and \emph{convex-minded} if $\alpha \in (0,1)$, \ie when the objective function is truly a convex combination of the two. This formulation captures the trade-off between maximizing immediate user satisfaction and ensuring long-term ecosystem quality.

\subsubsection*{Existence of symmetric MNE (Section~\ref{sec:sub-exist})}
First, we show that a pure Nash equilibrium (PNE) does not exist under any policy where $p_i \neq p_j$ for some $i, j \in [n]$, thereby generalizing the non-existence result of~\cite{jagadeesan2024supply} beyond the $\HM$ policy.
This motivates the study of mixed strategy equilibria.
Given the homogeneity of the setting, it is natural to focus on symmetric mixed Nash equilibria (MNE), where all producers adopt the same strategy.\footnote{\ssedit{Symmetric equilibrium is widely adopted, especially in symmetric games, due to its simplicity and anonymity: players do not need to coordinate with others about which strategies to play in order to reach equilibrium outcomes.}}

On the other hand, standard existence theorems do not directly apply due to discontinuities in the producers' payoff functions: a small increase in content quality can cause a discontinuous jump in utility, as the recommendation probability may shift abruptly from $p_{i+1}$ to $p_i$ for some $i$.

To overcome this challenge, we exploit the seminal result of~\cite{reny1999existence}, which establishes the existence of symmetric MNE in quasi-symmetric games under two key conditions: diagonally better-reply security and upper semi-continuity of the payoff function with respect to the mixed strategy.
We verify that our setting satisfies these conditions by metrizing the space of mixed strategies (endowed with the weak$^*$ topology) using the L\'evy-Prokhorov metric under the $\ell_2$ norm.\footnote{We refer to Section~\ref{sec:existence} and Appendix~\ref{sec:prelim} for more details on the weak$^*$ topology and the L\'evy-Prokhorov metric.} This approach allows us to reason about distances in the space of pure strategies rather than directly analyzing the underlying topological properties of the space of distributions.

Our existential result generalizes (i) that of~\cite{jagadeesan2024supply}, who prove the existence of symmetric MNE when the platform runs the $\HM$ policy, to a broader class of policies, 
and (ii) that of~\cite{barut1998symmetric}, who characterize properties of the symmetric MNE for the linear cost function.
In particular, the existential proof by~\cite{barut1998symmetric} relies heavily on the necessary conditions that an equilibrium must satisfy if it exists and manually argues that a system of equations provides a symmetric MNE with bounded support and that any deviation beyond the support only decreases the payoff, all of which are largely tailored to their specific setting.
We provide a more generic proof based on topological arguments, which easily extends beyond the linear cost function case in~\cite{barut1998symmetric}.\footnote{\cite{fang2020turning} informally claim that the existential argument by~\cite{barut1998symmetric} easily extends beyond the linear cost case without formal justification, which we believe may need to be formally justified given their subtle arguments. In any case, we believe our results can serve as a formal framework to argue the existence of symmetric MNE in the contest literature.}


\begin{table}[h]
\centering \footnotesize
\renewcommand{\arraystretch}{1.2}
\begin{tabular}{|c|c|c|c|}
\hline
\textbf{Setting} & \(\beta\) & \textbf{Techniques} & \textbf{Optimal Policy} \\
\hline
\(\alpha = 1\) (User welfare) & Any & Schur-convexity & \(\HM\) \\
\hline
\multirow{3}{*}{\(\alpha = 0\) (Platform quality)} 
& \(< 1\) & Schur-convexity & \(\HM\) \\
& \(= 1\) & Neutral (linear) & Any \\
& \(> 1\) & Schur-concavity & \(\UNI\) \\
\hline
\(\alpha \in (0,1)\) (Mixed) & Any & Variation diminishing & \(p_1 \ge p_2 = \ldots = p_{n-1} \ge p_n = 0\) \\
\hline
Arbitrary posynomial & $-$ & Variation diminishing & \(p_1 \ge p_2 = \ldots = p_{n-1} \ge p_n = 0\) \\
\hline
\end{tabular}
\caption{Optimal policy under different values of \(\alpha\) (objective) and \(\beta\) (cost curvature). ``Any'' in the optimal policy column denotes that any policy with $p_n = 0$ would induce the same objective values.}
\label{tab:summary}
\end{table}

\ssedit{
\subsubsection*{Uniqueness of symmetric MNE (Section~\ref{sec:sub-unique})}
Recall that our existential theorem guarantees that, for any nontrivial policy $\bp$, there exists at least one corresponding symmetric equilibrium.
Looking ahead to Section~\ref{sec:opt}, we prove that the RS's optimization problem at any symmetric MNE, given a policy $\bp$, can be reduced to a simpler optimization problem that involves only $\bp$ in the objective and constraints—without requiring $\bbmu$.
This suggests that the symmetric MNE may be unique. Indeed, in Theorem~\ref{thm:uniq-symm}, we prove that this is the case; that is, the symmetric MNE corresponding to a nontrivial policy is unique.
This follows from some observations on the characteristics of the cumulative distribution function of the symmetric MNE, stemming from the properties of the equilibrium.
}


\subsubsection*{Optimal rank-based policy: optimization over Bernstein basis polynomials (Section~\ref{sec:opt})}
We now turn to the problem of characterizing the optimal rank-based recommendation policy, assuming content producers play according to the symmetric MNE. In this setting, the RS aims to solve the following optimization problem:
\begin{align*}
\begin{aligned}
    (\text{\Opt}_1)\quad  & \underset{\bp \in \Delta}{\text{maximize}}
    & & \alpha \cW(\bbmu,\bp) + (1-\alpha ) \cQ(\bbmu,\bp)\\
    & \text{subject to}
    & & \text{$\bbmu$ is a symmetric MNE}
\end{aligned}
\end{align*}
At first glance, solving this problem appears intractable because it requires an explicit characterization of the equilibrium distribution $\bbmu$. 

Interestingly, we show that solving the RS’s problem does not require explicitly computing $\bbmu$. We first prove that, in any optimal policy, the last-ranked content should receive zero exposure, \ie $p_n = 0$ using correlation inequality, and further that the symmetric MNE does not have a point mass.
Building on this, we derive a reduced optimization problem using characteristics of symmetric MNEs and algebraic properties of Bernstein basis polynomials {along with their intricate relation to the platform's objective function}, stated as follows.
\begin{theorem}\label{thm:prop-user}
    The platform's problem can be written as the following optimization problem:\footnote{As a side product, note that one can directly re-obtain Proposition~\ref{thm:single-hm} without characterizing the equilibrium.}
    \begin{align*}
    \begin{aligned}
        (\text{\Opt})\quad  & \underset{\bp}{\text{maximize}}
        & & \alpha n\int_0^1 h(x,\bp)^{1 + 1/\beta}dx + (1-\alpha) \int_0^1 h(x,\bp)^{1/\beta} dx\\
        & \text{subject to}
        & & p_1 \ge p_2 \ge \ldots \ge p_n = 0
        \\
        & & & \sum_{i=1}^n p_i = 1
    \end{aligned}
    \end{align*}
\end{theorem}
Note  that $h(x,\bp) = \sum_{i=1}^n a_i(x) p_i$, and $a_i(x) = \binom{n-1}{i-1}x^{n-i}(1-x)^{i-1}$ are the $(n{-}1)$-dimensional Bernstein basis polynomials.\footnote{The standard $i$-th $n$-dimensional Bernstein basis polynomial is $b_{i,n}(x) = \binom{n}{i} x^i (1 - x)^{n - i}$.}
This reduction reveals a key insight: \emph{we can optimize over $\bp$ without explicitly solving for the equilibrium $\bmu$}, enabling a tractable analysis of the RS's design problem.
Further, this optimization result immediately implies an interesting corollary (Proposition~\ref{prop:price-hardmax}), stating that $\HM$ can be arbitrarily worse than the optimal policy as the number of producers $n$ increases.

\begin{figure}[t]
    \centering
    \subfigure[Optimal Policy with respect to $\alpha$]{
    \includegraphics[scale=0.52]{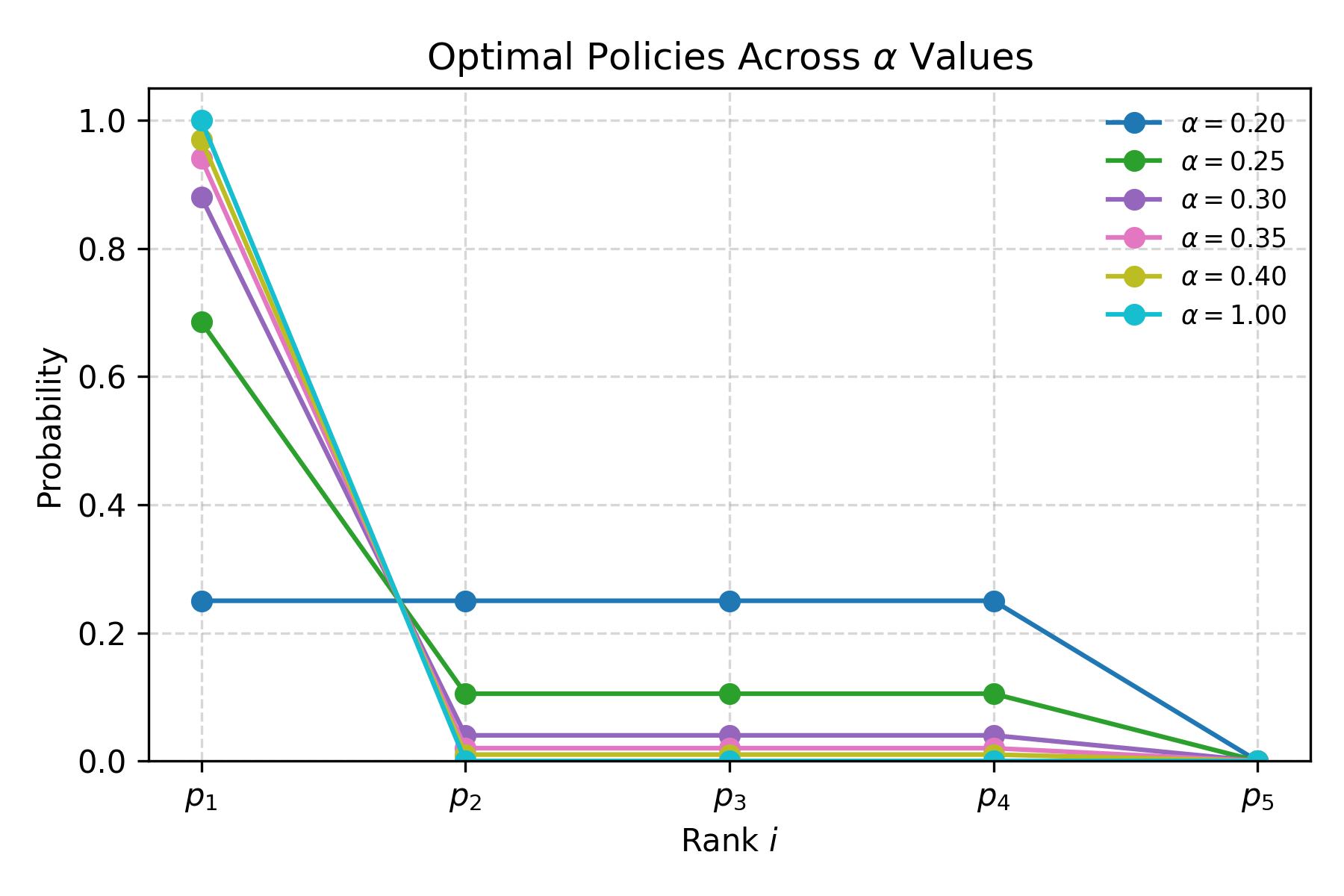}
    }  
    \subfigure[Optimal Policy with respect to $\beta$]{
    \includegraphics[scale=0.52]{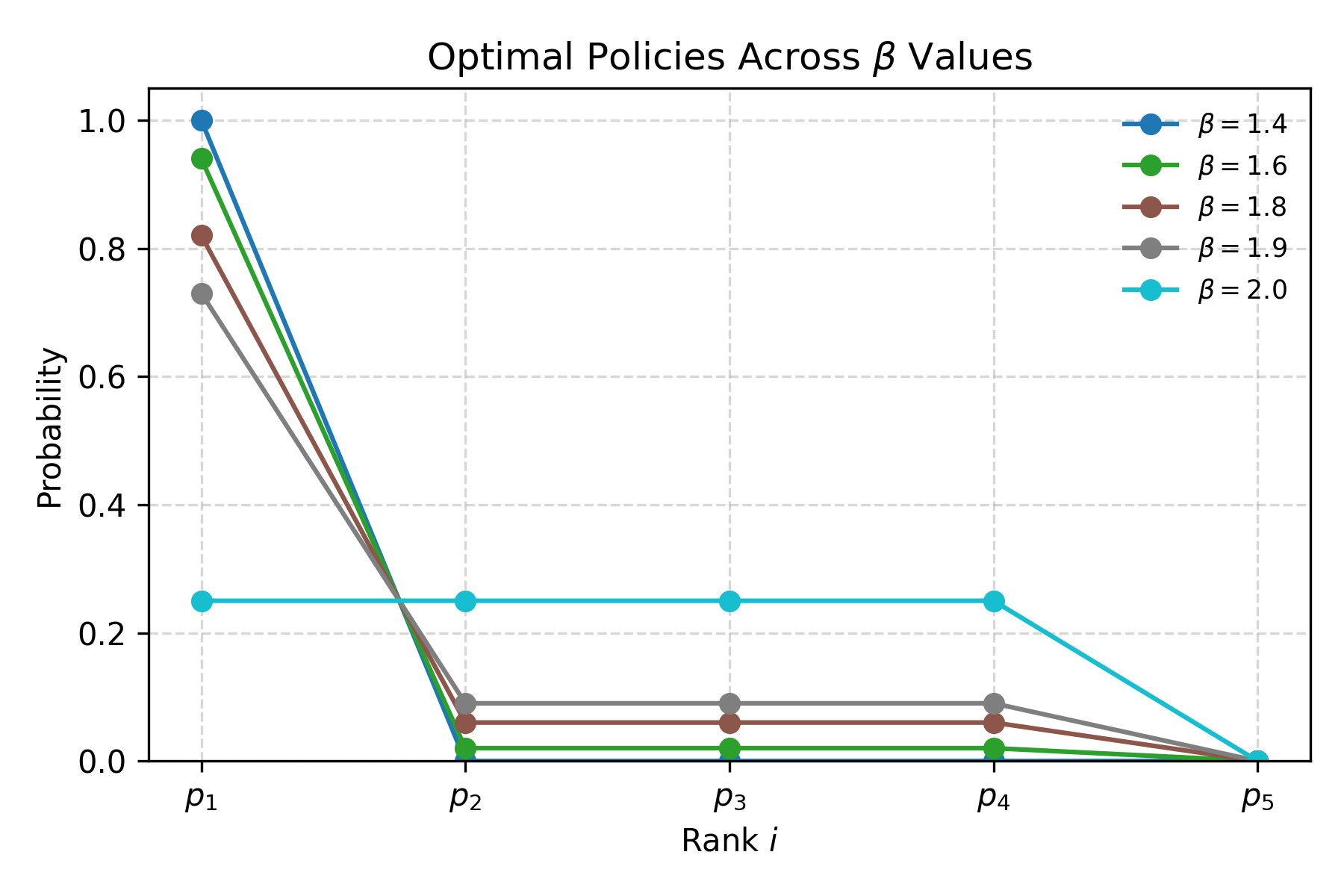}
    }
    \caption{Structures of optimal policy with respect to different $\alpha$ and $\beta$ values given $n = 5$ producers. We compute the optimal policy using a brute-force search over every possible policy after discretizing the policy space with granularity $0.005$, and using trapezoidal rule to approximate integral values over the grid $[0,0.001,0.002,\ldots, 1]$. For the left figure, we set $\beta = 2$ and for the right figure, we set $\alpha = 0.24$. This yields that our characterization is tight in a sense that the structures $p_1 > p_2 = \ldots = p_{n-1} \ge p_n = 0$ or $p_1 = p_2 = \ldots = p_{n-1} \ge p_n = 0$ indeed appear in the optimal policies.}
    \label{fig:opt-pol}
\end{figure} 

\subsubsection*{Single-minded case via Schur-convexity (Section~\ref{sec:opt-structure})}
Despite the simplified formulation, solving the optimization problem remains challenging, as it involves a sum of two integrals—each defined over a nontrivial function $h(x,\bp)$ raised to a power. A key analytical feature is that both integrands share the form $h(x,\bp)^r$, where $h(x,\bp)$ is the inner product between the Bernstein basis polynomials and the policy vector $\bp$.

To make progress, we first analyze the properties of $h(x,\bp)^r$ for $r \ge 0$. Interestingly, if we extend the domain of $h(x,\bp)$ to the full $n$-dimensional probability simplex—by sorting any input vector in decreasing order before applying $h$—the resulting symmetric function becomes \emph{Schur-convex} when $r \ge 1$ or $r \le 0$, and \emph{Schur-concave} otherwise.\footnote{A function $f:\mathbb{R}^n \to \mathbb{R}$ is Schur-convex (resp. Schur-concave) if $f(x) \ge f(y)$ (resp. $f(x) \le f(y)$) whenever $x$ majorizes $y$. See Section~\ref{sec:opt-structure} for details.}

This characterization immediately yields several results. When $\alpha = 1$, the RS's objective becomes Schur-convex, as $r = 1 + 1/\beta > 1$, which—together with the constraint $p_n = 0$—implies the optimality of the $\HM$ policy. When $\alpha = 0$, the objective is Schur-convex if $\beta \le 1$ and Schur-concave if $\beta \ge 1$, implying that $\HM$ is optimal under concave costs, and $\UNI$ under convex costs, \ie a phase transition occurs. Notably, when $\beta < 1$, both user welfare and platform quality are Schur-convex, reinforcing the optimality of $\HM$.
This is formalized in the following theorem.
\begin{theorem}\label{thm:opt-single}
Depending on the weight parameter $\alpha$ and the cost parameter $\beta$, the following holds:
\begin{enumerate}
    \item If $\alpha = 1$ or $\beta \le 1$, or $\alpha = 0$ and $\beta < 1$, then $\HM$ is optimal.
    \item If $\alpha = 0$ and $\beta > 1$, then $\UNI$ is optimal.
    \item If $\alpha = 0$ and $\beta = 1$, then any $\bp \in \Delta$ with $p_n = 0$ is optimal.
\end{enumerate}
\end{theorem}

A similar result is often observed in the literature~\citep{glazer1988optimal,fang2020turning}, \ie either of these two extreme policies is optimal in specific scenarios.  
Our analysis via Schur-convexity provides clearer intuition on why these policies turn out to be optimal in these regimes, with broader generalizability due to its analytical clarity.  
For instance, this approach allows the extension to several generalized objective functions discussed in Section~\ref{sec:general}, combined with our techniques for the convex-minded case.



\subsubsection*{Convex-minded case via variation diminishing property (Section~\ref{sec:opt-structure-general})}
In the general case where $\alpha \in (0,1)$, the objective becomes a combination of a Schur-convex and a Schur-concave function. This poses a significant analytical challenge, as such objective functions typically do not admit a closed-form characterization of the optimal solution; \eg consider optimizing a function that is the sum of a convex and a concave function.

Despite this difficulty, we establish that the optimal policy admits a remarkably structured form: $p_1 \ge p_2 = \ldots = p_{n-1} \ge p_n = 0$. In other words, the optimal $\bp$ is a piecewise constant vector with at most two distinct values (possibly differing at the top and bottom ranks), which is formally stated as follows:
\begin{theorem}\label{thm:opt-general}
    For any $\alpha \in [0,1]$ and $\beta > 0$, any optimal policy has the structure:
    \[
        p_1 \ge p_2 = p_3 = \ldots = p_{n-1} \ge p_n = 0.
    \]
\end{theorem}
To derive this result, we develop several technical tools involving Bernstein basis polynomials and their matrix representations, which may be of independent interest.

\begin{figure}[t]
    \centering
    \subfigure[$p_1$ of optimal policy with respect to $\alpha$ and $\beta$]{
    \includegraphics[scale=0.35]{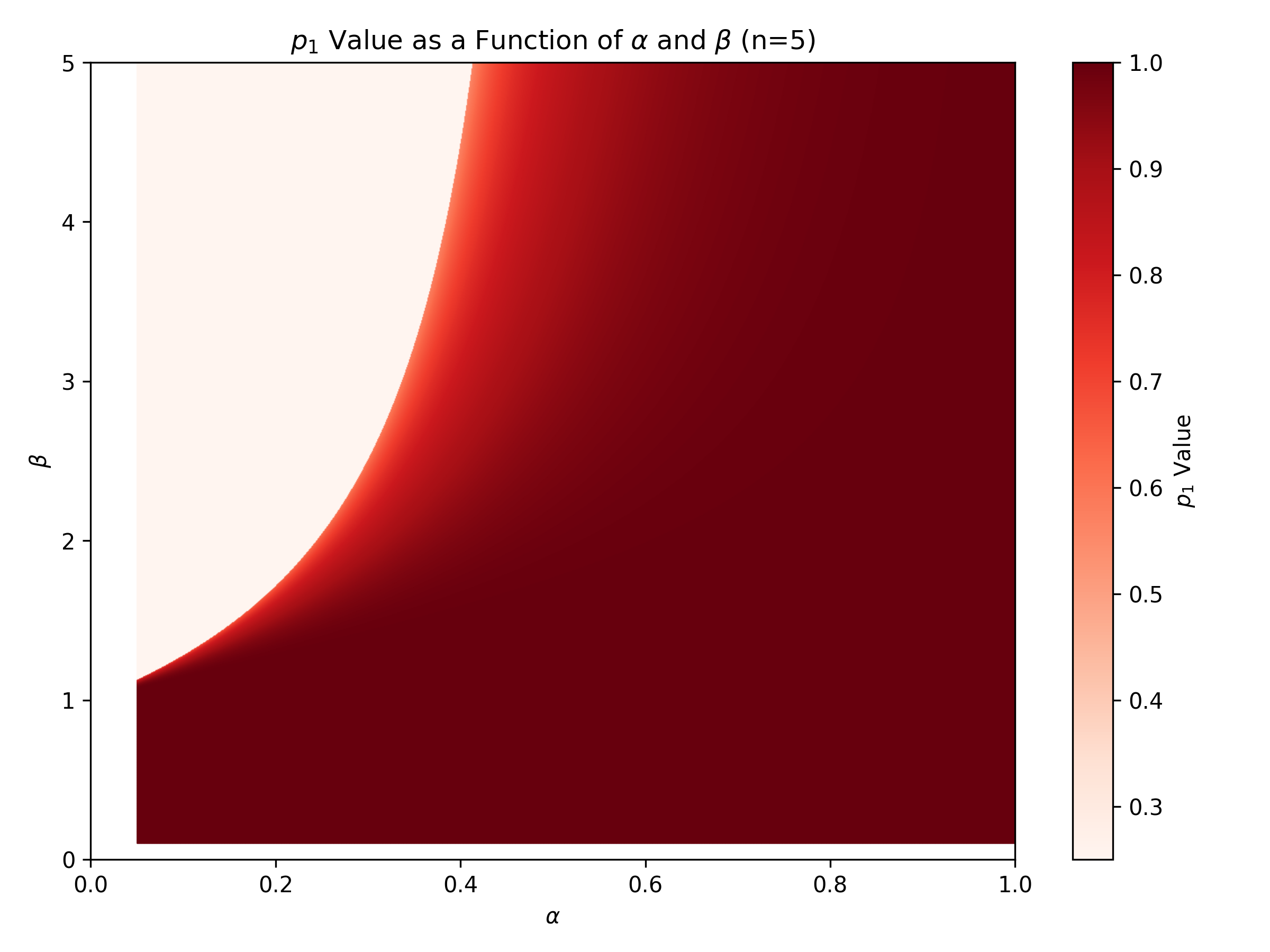}
    }
    \subfigure[$p_2$ of optimal policy with respect to $\alpha$ and $\beta$]{
    \includegraphics[scale=0.35]{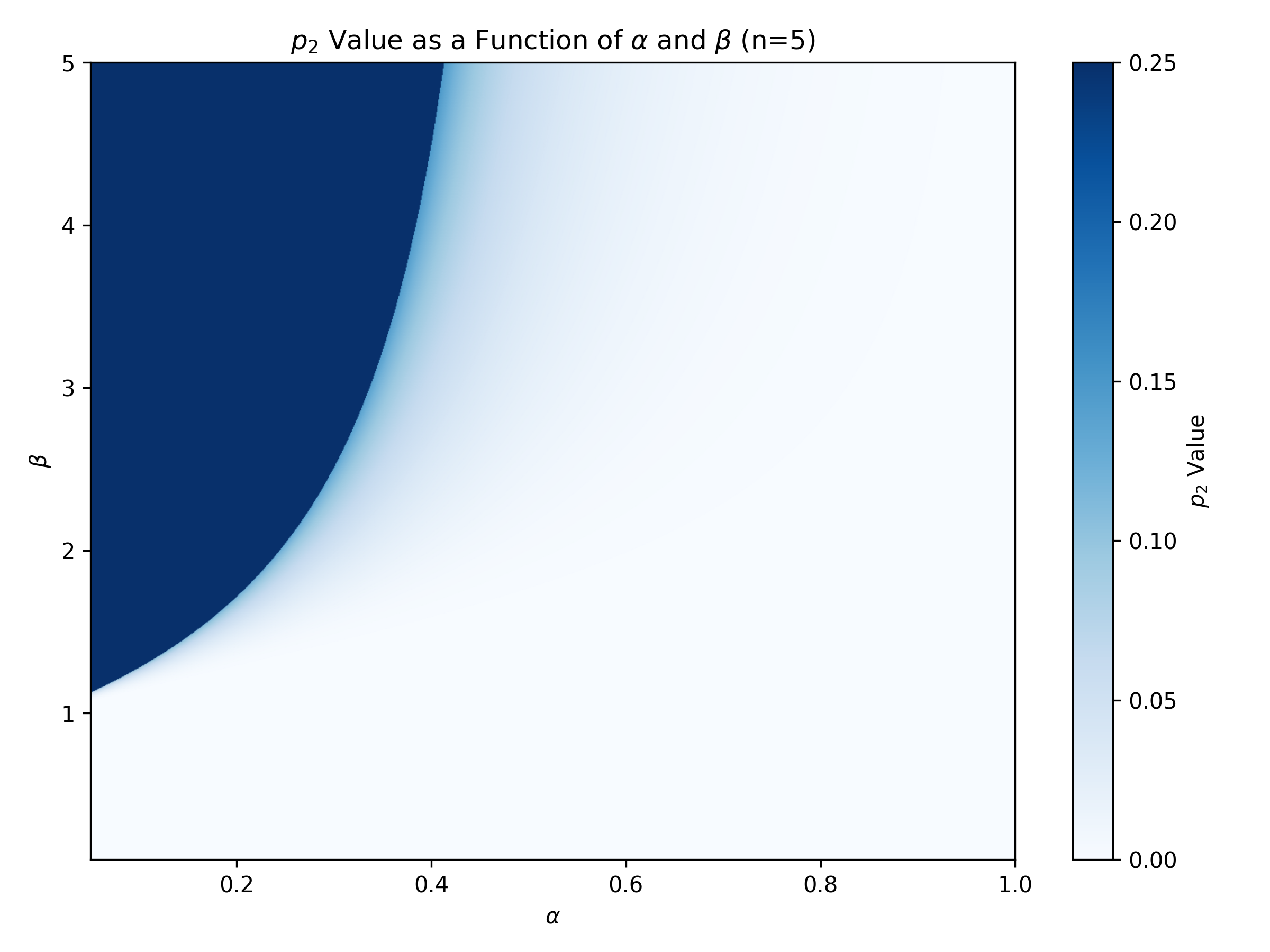}
    }
    \caption{Structures of optimal policy with respect to different $\alpha$ and $\beta$ values given $n=5$ producers. Darker red means larger $p_1$ values, and darker blue means larger $p_2 = p_3 = \ldots = p_{n-1}$ values. Each of the optimal policy is computed in a brute force manner after discretizing policy space with granularity $0.001$ and using trapezoidal rule with number of points $200$. Parameters $\alpha$ and $\beta$ are selected from the ranges $[0.05, 1]$ and $[0.1,5]$, respectively, with $1000$ points calculated for each.}
    \label{fig:opt-pol-two-dim}
\end{figure}

Formally, we show that for any real numbers \( 0 < x_1 < \ldots < x_k < 1 \) and integers \( 0 \le i_1 < \ldots < i_k \le n \), the matrix \( (a_{i_s}(1 - x_l))_{s,l \in [k]} \) exhibits a desirable property known as \emph{total positivity}, as introduced by~\cite{karlin1964total}—meaning that every minor of the matrix has a positive determinant. This is established by expressing each determinant as a generalized \emph{Vandermonde} matrix~\citep{yang2001generalization}.\footnote{See Section~\ref{sec:opt-structure-general} for technical details.}

We then leverage a seminal result from the theory of total positivity and its connection to the \emph{variation diminishing property}\footnote{The variation diminishing property refers to the phenomenon where the number of sign changes in a function is reduced when transformed through a totally positive operator. See Section~\ref{sec:opt-structure-general} for further explanation.}~\citep{karlin1964total,karp2024unimodality} to show that the partial derivatives of the objective function with respect to \( p_i \) for \( i \in [n-1] \) exhibit a \emph{quasiconvex} structure.\footnote{We use ``quasiconvex" to mean that the function first decreases and then increases with respect to the index \( i \). See Definition~\ref{def:quasiconv} for the formal definition.} That is, these partial derivatives decrease up to a certain index and then increase, as a consequence of bounding the number of sign changes in these sequences by those of a related quasiconvex function involving \( h(x, \bp) \).
This is the key ingredient of our main results, formally stated in Lemma~\ref{lm:unimodal}. In particular, it provides a sufficient condition on the objective function under which such properties hold, which will easily allow us to generalize the results, as will be discussed shortly.

Finally, we formulate the Lagrangian of the optimization problem and apply the \emph{Karush--Kuhn--Tucker (KKT)} conditions to derive the necessary optimality conditions, leveraging the geometric insights obtained from the structure of the partial derivatives.

As a final remark, Figure~\ref{fig:opt-pol-two-dim} illustrates how the values of $p_1$ (left) and $p_2$ (right) vary with $\alpha$ and $\beta$. Combined with the fact that $\HM$ can be arbitrarily worse than the optimal policy as $n$ increases, these results suggest that allocating nonzero probabilities to lower-ranked content is crucial for designing optimal recommendation policies in the presence of strategic content producers.

\ssedit{
\paragraph{Implication on high-dimensional nonconvex optimization problem}
In fact, our techniques for such structural characterization yields an interesting implication for a general optimization problem.
Specifically, given a monotone increasing nonnegative (possibly nonconvex) function $f$ and a constant $C > 0$, consider the following generic optimization problem:
\begin{align*}
    \begin{aligned}
        (\text{\Gen})\quad& \underset{\bx \in \R^n}{\text{maximize}}
        & & f(\bx)\\
        & \text{subject to}
        & & x_1\ge x_2 \ge \ldots \ge x_n \ge 0
        \\
        & & & \sum_{i=1}^n x_i = C
        \\
        & & & 
    \end{aligned}
\end{align*}
Then, a general version of our structural result can be written as the following: 

\begin{corollary}\label{thm:opt-generic}
    Consider the constrained optimization problem (\Gen).
    Suppose that the gradient of $f$ with respect to each $x_i$ can be written as the following:
    \begin{align*}
        \frac{\partial f(\bp)}{\partial x_i} = \int_D w(y,\bp)g_i(y) dy,
    \end{align*}
    for a continuous bounded interval $D$, a function $g_i:D \to \R_{\ge 0}$ for $i \in [n]$, and another function $w(y,\bp): D \times \R^n \to \R$.
    Assume further that $g_i(y)$ is $\STP_{n}$,\footnote{We refer to Definition~\ref{def:tp} for more details.}, \ie strongly totally positive of order $n$, and $w(y,\bp)$ is quasiconvex over $y \in D$, \ie it decreases then increases, given any $\bp$.
    Then, it follows that the optimum of $f$ satisfies $x_1 \ge x_2 = \ldots = x_n$.\footnote{Indeed, for our objective function presented in (\text{\Opt}), the proof of Lemma~\ref{lm:unimodal} asserts that this subsumes our result by setting $w(x,\bp) = n\alpha (1+1/\beta)h(x,\bp)^{1/\beta}+(1-\alpha)/\beta h(x,\bp)^{1/\beta - 1}$ and $g_i(x) = a_i(x)$.} \NGcomment{Can we say in a footnote, what is $w$ and $g_i$ in our orginal setting?} \sscomment{Added a footnote.}
\end{corollary}
This structural characterization has a significant implication.
Since once we know that $x_1 \ge x_2 = \ldots = x_n = x$, due to the monotonicity of $f$, it is immediate that $x_1 = C - (n-1)x$.
Thus, the only variable that the algorithm needs to decide is $x \in [0,C/(n-1)]$, \ie it becomes a single-dimensional optimization problem on a bounded interval.
Since single-dimensional optimization usually admits an efficient tractable solution under reasonable assumptions, this implies that the optimum for the original high-dimensional nonconvex optimization can be efficiently solved. \NGcomment{do you say, for our own setting, we present a polytime algorithm as we elaborate next} \sscomment{Added.}
Indeed, in what follows, we discuss that the optimal policy can be efficiently computed for our objective function in (\text{\Opt}) in Theorem~\ref{thm:prop-user}.
}

\sscomment{Can we say any other application?}

\ssedit{

\subsubsection*{Computing the optimal policy (Section~\ref{sec:computation})}
Note that, given the structural results, the optimization problem reduces to a single-variable optimization. However, the objective function is still a combination of convex and concave components, which may not always admit a tractable solution.
In Section~\ref{sec:computation}, we present an efficient algorithm that computes an approximately optimal policy with either additive or multiplicative error $\eps$, running in $\mathrm{poly}(n,1/\eps)$ time given a constant $\beta$. The algorithm assumes access to a value oracle that returns the objective function value for any input.\footnote{We further discuss that the access to the value oracle can be removed by approximating the integrals in the objective function while keeping the order of error rate.} It works by iteratively narrowing an interval using a carefully constructed pair of upper and lower bounds on the objective, whose gap shrinks as the interval length decreases. This is possible because the integrand in the reduced objective function is an affine transformation of the single variable $p_1$, followed by exponentiation parameterized by $\beta$.
Using this structure, a branch-and-bound-style algorithm is employed to explore the search space: it branches on subintervals when necessary and prunes those that violate optimality conditions, efficiently converging to a near-optimal solution.
}

\subsubsection*{Beyond RS: general objective functions (Section~\ref{sec:general})}
Finally, we discuss how our proposed techniques accommodate a broader range of objective functions in Section~\ref{sec:general}.  
First, we show that an arbitrary posynomial\footnote{A posynomial function is a polynomial with nonnegative coefficients whose exponents may be noninteger.} function over the produced contents' qualities at the symmetric MNE still admits the exact same optimal structure:
\begin{theorem}\label{thm:general-obj}
    Consider the following   objective function at symmetric MNE $\bmu$ given $\bp$:
    \begin{align*}
        \cG(\bmu, \bp) = \frac{1}{n}\sum_{i=1}^n\Exu{q_i \sim \mu_i}{e_m q_i^{k_m} + e_{m-1}q_i^{k_{m-1}} + \ldots + e_1 q_i^{k_1}},
    \end{align*}
    where $q_i \sim \mu_i$ for a symmetric strategy profile $\mu_i$ in $\bmu$.
    For any $\beta > 0$, 
    suppose $e_j$'s and $k_j$'s satisfy the following property:
    \begin{enumerate}
        \item There exists $j^* \in [m]$ such that $e_j(k_j - \beta) \le 0$ for $j \le j^*$ and $e_j(k_j - \beta) \ge 0$ for $j > j^*$.
        \item All the elements in the sequence $e_j(k_j - \beta)$ for $j \in [m-1]$ is nonnegative or nonpositive.
    \end{enumerate}
    Then, any optimal policy follows
    \[
        p_1 \ge p_2 = p_3 = \ldots = p_{n-1} \ge p_n = 0.
    \]
    As a special case, if $e_j$'s are all nonnegative, this condition follows.
\end{theorem}
The proof is similar to Theorem~\ref{thm:opt-general}, \ie we show that the objective function satisfies the condition in Lemma~\ref{lm:unimodal}.  
On the other hand, this condition is established using different approaches, by proving how the coefficients of an ordered posynomial relate to the quasiconvexity of the objective function via a generalized version of Descartes' rule of signs~\citep{budan1807nouvelle,wang2004simple}.

Notably, this generalized argument implies several interesting corollaries.  
For instance, our results apply to the following objective functions, thereby admitting the same noted optimal structure:
\begin{enumerate}
    \item Social welfare, which sums up every involved player's utility function combined with a generalized platform utility (posynomial).
    \item Inverse S-shaped functions, \eg $2q^3 - 3q^2 + 2q$, which are concave-then-convex over the induced contents' qualities, capturing a smoothed version of thresholded utility.\footnote{These are the opposite of S-shaped functions from prospect theory~\citep{barberis2013thirty} in behavioral economics, which exhibit convex-then-concave geometry.}
    \item Largest order statistics, corresponding to the highest quality among the contents produced by the producers.
    \item Exponential utility $e^{\lambda q}$ for any $\lambda> 0$ and even arbitrary convex combination among them.
\end{enumerate}
Remarkably, our results allow many positive combinations of these objective functions, as long as the sequence of the objective function's partial derivatives with respect to the $p_i$'s remains quasiconvex, which will become clear from our main analysis.

\section{Related Works}\label{sec:rel}
Our model is closely related to the literature on rank-based contests, where the allocation probability $p_i$ for the $i$-th highest-quality content parallels the assignment of a prize $p_i$ to the $i$-th highest-ranked contestant in a competition with $n$ participants.\footnote{In the contest design literature, the total prize may not sum to 1; however, normalization allows for analogous interpretations.} For a comprehensive survey of classic work on contest design—including rank-based and Tullock contests—we refer to~\cite{sisak2009multiple}. Additional related literature is discussed in Appendix~\ref{apd:rel}.

\paragraph{Rank-based contests}

\paragraph{Rank-based contests}
The design of rank-based contests—where prizes depend on contestants’ ranks—has been extensively studied since \cite{lazear1981rank} and \cite{glazer1988optimal}. \cite{glazer1988optimal} analyze optimal prize allocation under linear costs and common risk preferences, proving that the $\UNI$ policy is optimal in a complete information setting. This aligns with our model when $\beta = 1$ and $\alpha = 0$. \cite{moldovanu2008optimal} extend this to general cost functions under incomplete information, focusing on top-two prize structures and identifying when $\HM$ (winner-take-all) is optimal. \cite{barut1998symmetric} examine equilibrium properties under exogenous policies with linear costs, while \cite{greenwald2018simple} show that uniformly rewarding the top quartile yields a constant-factor approximation for convex costs when maximizing average quality. \cite{ghosh2016optimal} study a strategic participation model with fixed costs and qualities, showing that distributing equal prizes to top contestants is optimal with identical costs and that $\HM$ offers a 3-approximation more generally.
\cite{fang2020turning} observes a general connection  between the competitiveness of the contest and monotonicity of utility function when the utility function is convex (concave) and the cost function is concave (convex), respectively, under the complete information setting.
A recent work by~\cite{elkind2025contest} study a thresholded objective function with at most two steps in the incomplete information setting, and show that the optimal policy is a convex combination over simple policies with structure $(1/t,1/t,\ldots, 1/t,0,\ldots, 0)$ with $t$ numbers of $1/t$ values. Note that this class of polices subsumes our optimal policy.

Our model generalizes prior work by incorporating broader platform objectives and general cost functions, offering new insights into contest design under complete information. 
Conceptually, it is the first to study recommender systems as an application of the classic contest design in the literature.

\paragraph{Strategic content production in recommender system} 
\cite{jagadeesan2024supply} investigate competition among multiple content producers, each deciding an effort level to create content of a certain quality. Their single-dimensional model yields the same expected payoff structure as ours, though they analyze equilibria only under the $\HM$ policy.
\cite{immorlica2024clickbait} examine strategic behavior by content producers who attempt to game the platform to artificially boost user engagement through clickbait. They show that content quality and gaming behavior are positively correlated, and that engagement-based optimization can reduce user welfare by incentivizing such behavior.
\cite{jagadeesan2023competition} study a duopoly market where two platforms running multi-armed bandit algorithms compete for user participation.
\cite{ghosh2011incentivizing} explore strategic content production, where producers decide whether to participate on the platform and choose their effort levels.
\cite{yao2023bad} analyze the inefficiency of top-$k$-style recommendation policies when content creators compete for exposure, showing that these policies can induce suboptimal equilibria.
\cite{yao2024rethinking} challenge the conventional belief that monotone reward structures are always desirable, demonstrating that non-monotone policies can better align creator incentives and improve user welfare.

Our results complement this line of work by adopting a general and abstract model to study strategic interaction, bridging classic contest design literature with the design of recommendation algorithms.

\section{Model}\label{sec:model}

There are $n \ge 2$ content producers (henceforth, \emph{producers}), each of whom creates a single piece of content through costly effort. Each producer $i$ strategically chooses an effort level to generate a competitive content item of quality $q_i$, whose interpretation will be made clear shortly. 

The produced content is registered on an online platform with a recommender system (RS). Whenever a user visits the platform, the RS recommends one of the registered content items according to a predetermined policy.

We consider a class of \emph{rank-based recommendation policies}, where the RS commits to an allocation vector $\bp \in \Delta$, with
$
\Delta = \left\{ (p_i)_{i \in [n]} : p_1 \ge p_2 \ge \cdots \ge p_n \ge 0,\; \sum_{i=1}^n p_i = 1 \right\},
$
i.e., the ordered $n$-dimensional probability simplex. Under policy $\bp$, the RS ranks the contents in decreasing order of quality and displays the item with the $i$-th highest quality with probability $p_i$, breaking ties uniformly at random.  
Formally, given the vector of content qualities $\bq = (q_1, \ldots, q_n)$, the RS first sorts them into order statistics $q_{(1)} \ge q_{(2)} \ge \cdots \ge q_{(n)}$,\footnote{Given values $x_1, \ldots, x_n$, we denote by $x_{(i)}$ the $i$-th largest element (order statistic).} and then displays $q_{(i)}$ with probability $p_i$.

We introduce two specific policies that will appear frequently throughout the paper:
\begin{enumerate}
    \item \textbf{$\HM$ policy}: $p_1 = 1$.
    \item \textbf{$\UNI$ policy}: $p_1 = p_2 = \ldots = p_{n-1} = \frac{1}{n-1}$ and $p_n = 0$.
\end{enumerate}

In words, $\HM$ deterministically recommends the highest-quality (i.e., most preferred) content to the user.  
In contrast, $\UNI$ uniformly randomly selects any content except the one with the lowest quality.  
These two policies represent extreme points in the policy space.

\paragraph{Induced game and equilibrium}
Given a policy $\bp$, we consider a game \emph{induced by} $\bp$ given $n$ producers.
Each producer $i$ strategically decides the quality $q_i$ of the content it creates from the action space $\R_{\ge 0}$.
Each producer $i$ may play a mixed strategy that randomizes over the action space.
In this case $q_i$ is a random variable and we write $\mu_i$ to denote the associated probability measure over $\R_{\ge 0}$ and $F_i: \R_{\ge 0} \to [0,1]$ to denote the corresponding c.d.f.
We say that a producer plays a pure strategy if $\mu_i$ has a singleton support.
We write $\bbmu = (\mu_1,\ldots, \mu_n)$ and $\bbmu_{-i} = (\mu_1,\ldots, \mu_{i-1},\mu_{i+1}, \ldots, \mu_n)$ to denote the strategy profiles except $i$'s strategy.

Producer $i$'s expected utility $\cU_i$ is a function over $\bbmu$ and the policy $\bp$, which will be defined shortly.
Given a policy $\bp$, a strategy profiles $\bbmu$ is called a (Nash) equilibrium if no producer has incentive to deviate from $\bbmu$, \ie for any $i \in [m]$:
\begin{align*}
    \cU_i(\mu_i, \bbmu_{-i}, \bp) \ge \cU_i(\mu'_i, \bbmu_{-i}, \bp),
\end{align*}
for any mixed strategy $\mu'_i$.
If every strategy $\mu_i$ in $\bbmu$ is the pure strategy, the resulting equilibrium is pure Nash equilibrium (PNE), otherwise it is mixed Nash equilibrium (MNE).
An equilibrium is called \emph{symmetric} if every player is playing the same strategy $\mu_i = \mu$ for $i \in [n]$ under the equilibrium.
In this case, we write $F$ to denote the corresponding c.d.f.
In particular, throughout, we focus on symmetric MNE of the induced game, which will formally be shown to  exist in Section~\ref{sec:existence}.\footnote{In fact, a pure Nash equilibrium does not always exist—even in the simple case where the platform uses the $\HM$ policy—as shown in Proposition~\ref{prop:no-pne}.}

\paragraph{Producer's utility}
Producer $i$'s expected utility consists of two components: (i) the \emph{revenue}\footnote{Revenue may originate from advertisements embedded by the producer or inserted by the platform (e.g., pre-roll or mid-roll ads), as well as platform-based compensation mechanisms tied to outcomes such as views, clicks, or user retention.} earned when their content is recommended, and (ii) the \emph{cost} of producing content of a given quality.
Each time producer $i$’s content is recommended, they receive a unit revenue of $1$. The expected revenue $\cR_i$ when producer $i$ follows strategy $\mu_i$ is:
\begin{align*}
\cR_i(\mu_i, \bbmu_{-i}, \bp) = \mathbb{E}_{\bq \sim \bbmu}\left[\Ind{I = i}\right],
\end{align*}
where $I$ is the index of the producer whose content is recommended, based on the realized qualities $\bq$ and the policy $\bp$, and $\Ind{\cdot}$ is an indicator variable.

To generate content, a producer incurs a cost determined by the quality level. Specifically, producing content of quality $q$ incurs cost $c(q) = q^\beta$ for some $\beta > 0$, capturing convex or concave effort costs.
Thus, given others' strategies $\bbmu_{-i}$, the expected utility of producer $i$ when playing mixed strategy $\mu_i$ is:
\begin{align*}
\cU_i(\mu_i, \bbmu_{-i}, \bp)
= \cR_i(\mu_i, \bbmu_{-i}, \bp) - \mathbb{E}_{q_i \sim \mu_i}[c(q_i)].
\end{align*}

It is also worth noting that our model implicitly captures individual rationality (IR) condition of the producers, as playing $q = 0$ gives the deterministic utility of $p_n \ge 0$ (larger utility upon ties), and thus outperforming the outside option of not participating the contest.

\paragraph{User welfare}
The platform’s long-term success critically depends on \emph{user welfare}, which we proxy through user engagement.
We assume that when a user is recommended and consumes content of quality $q$, they derive utility equal to $q$. Given a recommendation policy $\bp$ and the producers' strategy profile $\bbmu$, we define user welfare $\cW(\bbmu, \bp)$ as the expected quality of the recommended content:
\begin{align*}
    \cW(\bbmu, \bp) 
    &= \Exu{\bq \sim \bmu}{\sum_{i=1}^n\Ind{I = i} \cdot q_i}.
\end{align*}
In other words, $\cW(\bbmu, \bp)$ captures the average utility delivered to users based on the realized content qualities and the recommendation outcome.
When $\bp$ is fixed or clear from context, we write $\cW(\bbmu)$ for brevity.

\paragraph{Platform quality}
While the RS primarily aims to maximize user welfare $\cW$ in response to strategic producer behavior, it also benefits from sustaining a high overall quality of content on the platform—a factor that can influence long-term user trust and platform reputation.
We define \emph{platform quality} as the average quality of all registered content at the symmetric equilibrium $\bbmu$ under policy $\bp$:
\begin{align*}
\cQ(\bbmu, \bp) = \frac{1}{n} \sum_{i=1}^n \mathbb{E}_{q_i \sim \mu_i}[q_i].
\end{align*}
When $\bbmu$ is symmetric, this simplifies to the expected quality under any producer’s strategy.
As before, we write $\cQ(\bbmu)$ when $\bp$ is clear from context.

\paragraph{Objective function}
Finally, the RS's objective is to maximize the weighted summation of induced user welfare $\cW(\bbmu)$ and the average quality $\cQ(\bbmu)$:
\begin{align*}
    \OBJ(\bbmu, \bp) = \alpha\cW(\bbmu, \bp) + (1-\alpha) \cQ(\bbmu,\bp),
\end{align*}
where $\alpha \ge 0$ governs the importance between the user welfare and the platform's quality.

Essentially, we are interested in a policy $\bp$ that maximizes the objective function when the producers are playing symmetric equilibrium strategies $\bbmu$.

Thus, we aim to find a policy $\bp$ that solves the following optimization problem:
\begin{align*}
\begin{aligned}
    (\text{\Opt})\quad  & \underset{\bp \in \Delta}{\text{maximize}}
    & & \alpha \cW(\bbmu,\bp) + (1-\alpha ) \cQ(\bbmu,\bp)\\
    & \text{subject to}
    & & \text{$\bbmu$ is a symmetric equilibrium}
\end{aligned}
\end{align*}
We say the RS is \emph{single-minded} if $\alpha = 0$ or $1$, \ie it aims to solely maximize either the user welfare or the platform quality, and \emph{convex-minded} if $\alpha \in (0,1)$.
One might notice that in the optimal policy, it is required to have $p_i\neq p_j$ for some $i, j \in [m]$ since otherwise the dominant strategy is \ssedit{to} play $q = 0$.
Thus, we focus on the class of \emph{nontrivial} policy in $\Delta$ such that $\bp \neq (1/n,1/n,\ldots, 1/n)$.

In Appendix~\ref{sec:comparison}, we detail how our setup generalizes the model by~\cite{glazer1988optimal} and~\cite{jagadeesan2024supply}.



\section{Existence and Uniqueness of Symmetric MNE}\label{sec:existence}
In this section, we provide our results on the existence and uniqueness of symmetric mixed Nash equilibrium.
Since the platform's objective function depends on the equilibrium point, it is crucial to analyze the equilibrium point.
Before presenting our main results, we first prove that there exists no PNE for any nontrivial policy, the proof of which is in Appendix~\ref{apd:no-pne}.
\begin{proposition}\label{prop:no-pne}
    There exists no PNE for any nontrivial policy $\bp$.\footnote{Note that this result strictly generalizes the nonexistence of PNE by~\cite{jagadeesan2024supply} (Proposition 1) under $\HM$ rule.}
\end{proposition}

\subsection{Existence}\label{sec:sub-exist}
Proposition~\ref{prop:no-pne} motivates us to study the MNE instead of PNE.
\ssedit{
Specifically, we will focus on a symmetric MNE.
This is because it admits tractable analysis as well as allows producers to simply commit to the strategy without considering others' potentially unorganized behavior.
For instance, as the game faced by the producers are symmetric, playing a symmetric strategy is their best bet unless they coordinate to play asymmetric strategies.
}
\begin{theorem}\label{thm:exist}
For any nontrivial policy $\bp$, the induced game (defined in Section~\ref{sec:model}) admits a symmetric mixed Nash equilibrium.
\end{theorem} 
The proof is not immediate due to the discontinuous nature on payoff function.
We provide the proof sketch of Theorem \ref{thm:exist} at the end of this section.

Given the existence of the symmetric MNE, the following proposition states that we can effectively wipe out the tie-breaking scenario.

\begin{proposition}\label{cl:atomless}
    Given a nontrivial policy $\bp$ and symmetric MNE $\bbmu$, the corresponding c.d.f. $F$ of $\bmu$ is atomless and strictly increasing in some interval $[0,q_{\max}]$.
\end{proposition}
We refer to Appendix~\ref{apd:atomless} for the proof.
This implies that the c.d.f. $F$ of the symmetric MNE is differentiable almost everywhere (a.e.) by Lesbesgue's theorem, once it exists.


Finally, using the derived properties of the symmetric MNE, we can obtain explicit formulas for the user welfare, platform quality, and producer's payoff as follows.
\begin{lemma}\label{lem:derivation-welfare-payoff}
    Consider a symmetric MNE $\bmu = (\mu,\ldots, \mu)$ given a policy $\bp$.
    Producer $i$'s payoff of playing $q$ in the support of $\mu$ given others playing $\bmu_{-i}$ can be written as:
    \begin{align*}
        \cU_i(q, \bmu_{-i}, \bp) = \parans{\sum_{i=1}^n p_{i} \binom{n-1}{i-1}F(q)^{n-i}(1-F(q))^{i-1}} - c(q),
    \end{align*}
    where $F(\cdot)$ is the c.d.f. of $\mu$.
    Further, the user welfare and the platform's quality can be written as:
    \begin{align*}
        \cW(\bmu, \bp) = \sum_{i=1}^n p_i \Exu{\bq \sim \bmu}{q_{(i)}}, \qquad 
        \cQ(\bmu, \bp) = \Exu{q \sim \mu}{q},
    \end{align*}
    where $q_{(i)}$ denotes the $i$-th largest order statistics among $q_i$'s.
\end{lemma}

\subsubsection{Proof Sketch of Theorem \ref{thm:exist}}
Intuitively, the symmetry of the game---namely, the symmetry in both the providers' payoffs and their action spaces---supports the existence of a symmetric equilibrium. However, the formal proof requires a subtle argument based on the notion of \emph{better-reply security} and topological reasoning. This is due to the fact that, while the action space is continuous, the payoff function is discontinuous over pure strategies, owing to the thresholded nature of the rank-based policy. Specifically, increasing content quality can lead to a discontinuous jump in revenue---from $p_n$ to $p_{n-1}$, and so on, up to $p_1$.

To prove the existence of a symmetric mixed-strategy Nash equilibrium (MNE) for any nontrivial policy $\bp$, we rely on the framework developed by \citet{reny1999existence}, which extends equilibrium existence results to discontinuous games using the notion of \emph{better-reply security}. Our argument generalizes \citet{jagadeesan2024supply}, who focus solely on the $\HM$ policy, to a broader class of rank-based recommendation policies.

\paragraph{Strategy Space and Topology}
Each content producer chooses a quality level $q_i \in \mathbb{R}_{\ge 0}$ and may randomize over this action space. Thus, the strategy space of producer $i$ is the space of probability measures over $\mathbb{R}_{\ge 0}$. We restrict attention to probability measures with compact support in $[0, B]$ for some sufficiently large $B$, which can be justified by utility bounds and standard truncation arguments {as can be seen in the proof}.

We endow this space of probability measures with the \emph{weak$^*$ topology}---that is, the coarsest topology on the space of probability measures on $X$ such that the maps $\mu \mapsto \int_X f d\mu$ are continuous for every bounded continuous function $f: X \mapsto R$.
Formally, a sequence of probability measures $\{\mu_n\}$ converges weakly to $\mu$ if and only if for every bounded continuous function $f$, we have:
\[ 
\int f(x) \, d\mu_n(x) \to \int f(x) \, d\mu(x).
\]
This notion of convergence is crucial because it ensures a player's expected payoff is a continuous function of their opponent's mixed strategy.
Specifically, if an opponent's strategy sequence $\{\mu_n\}$ converges weakly to $\mu$, our expected payoff also converges.

Given that the producer's action space is compact, this continuity of the payoff mapping allows us to establish the existence of a symmetric MNE once we prove that the game satisfies better-reply secureness by~\cite{reny1999existence}.





\paragraph{Discontinuity and Better-Reply Security}
The utility function of each producer $i$ is discontinuous in the joint strategy profile due to tie-breaking in the ranking mechanism. That is, small changes in others' strategies can cause jumps in whether $i$'s content is ranked $j$-th or $(j+1)$-th, which in turn affects the allocation probability $p_j$.

Despite this discontinuity, we prove that the induced game satisfies \emph{diagonal better-reply security}, a sufficient condition for the existence of a mixed-strategy equilibrium established in \citet{reny1999existence}. 
The key idea is to show that for every symmetric strategy profile $(\mu,\ldots, \mu)$ with payoff $u$ that is not an equilibrium, there exists an open neighborhood $B$ of $\mu$ such that a producer can secure a payoff strictly greater than $u$ along the diagonal on $B$, \ie when all other producers play symmetric strategies within $B$.


These properties allow us to invoke Corollary 5.3 of \citet{reny1999existence} (or see Theorem 4.1) to establish the existence of a symmetric mixed-strategy equilibrium. 

Formal definitions and proofs are provided in Appendix~\ref{apd:thm:exist}.

\paragraph{Remark} For completeness, we provide background on the weak$^*$ topology, probability measures on Hausdorff spaces, and technical lemmas used in the proof in Appendix~\ref{sec:prelim}.

\ssedit{

\subsection{Uniqueness}\label{sec:sub-unique}
In fact, we show that the symmetric MNE is even unique.
\begin{theorem}[Uniqueness of the symmetric MNE]\label{thm:uniq-symm}
Fix $n\ge2$, $\beta>0$, and a nontrivial rank-based policy $p=(p_1,\dots,p_n)$ with $p_1\ge\cdots\ge p_n\ge0$
(not all equal).
Then the induced game admits a \emph{unique} symmetric mixed Nash equilibrium.
\end{theorem}
The proof uses the atomlessness and strict monotonicity of the c.d.f. of the symmetric MNE shown in Proposition~\ref{cl:atomless}, along with several observations that each producer's payoff function except the cost term is strictly increasing and continuous on $F(q)$, thereby admitting a continuous inverse function.
Then, the symmetric MNE becomes unique if its support is the same, which can be rather easily verified from the payoffs at the boundary points.

\begin{remark}
    Notably, thanks to our characterization on properties of the symmetric MNE, we further observe that the producers can easily play such equilibrium strategies.
    Formally, they can sample $U \sim \text{Uniform}[0,1]$, and play $q = h(U,\bp)^{1/\beta}$ where the function $h$ is defined in~\eqref{eq:h}.
    This follows from the relation~\eqref{eq:q-to-h} combined with the correspondence between sampling quantiles uniformly at random and sampling $q$ directly from $F$.
\end{remark}
}

\section{Detaching Equilibrium Dependency in  Optimization}\label{sec:opt}
Recall that the RS’s objective, denoted as $\text{\Opt}$, involves optimizing over the policy $\bp$ while accounting for the induced symmetric MNE.
In this section, we show that this problem can be reformulated as a pure optimization over $\bp$, without explicitly involving the equilibrium strategies.

Our key observation is that, under a symmetric MNE, any pure strategy in the support must yield the same expected utility, given that all other producers follow the equilibrium strategy. Leveraging this property, we first establish that any optimal policy must satisfy $p_n = 0$ using correlation inequality. Then, through algebraic manipulation, we express both user welfare $\cW(\bmu, \bp)$ and platform quality $\cQ(\bmu, \bp)$ entirely in terms of $\bp$, eliminating explicit dependence on $\bmu$.
Interestingly, the resulting expressions for $\cW$ and $\cQ$ share the same functional form, differing only in the exponent—highlighting a unified structure in the RS’s objective.

To present our main results, we introduce two notations for ease of exposition.
\begin{align}\label{eq:h}
    a_i(x) =  \binom{n-1}{i-1}x^{n-i}(1-x)^{i-1}, \qquad h(x,\bp) = \sum_{i=1}^n a_i(x) p_i.
\end{align}
Note that $a_i(x)$ is exactly the \emph{Bernstein basis polynomial} of degree $n-1$.  
The standard Bernstein polynomial $b_{\nu,n}$ is defined as:
$
    b_{\nu,n}(x) = \binom{n}{\nu}x^\nu(1-x)^{n-\nu},
$
for $\nu = 0,\ldots,n$.
Due to our construction of $a_i(x)$, observe that for each $i \in [n]$, we have: $a_i(x) = b_{n-i, n-1}(x).$

Now, our main theorem on the reduced optimization formulation can be stated as follows:

\begin{customthm}{\ref{thm:prop-user}}
    The platform's problem can be written as the following optimization problem:\footnote{As a side product, notice that one can directly re-obtain Proposition~\ref{thm:single-hm} without characterizing the equilibrium.}
    \begin{align*}
    \begin{aligned}
        (\text{\Opt})\quad  & \underset{\bp}{\text{maximize}}
        & & \alpha n\int_0^1 h(x,\bp)^{1 + 1/\beta}dx + (1-\alpha) \int_0^1 h(x,\bp)^{1/\beta} dx\\
        & \text{subject to}
        & & p_1 \ge p_2 \ge \ldots \ge p_n = 0
        \\
        & & & \sum_{i=1}^n p_i = 1
    \end{aligned}
    \end{align*}
    We write $G(\bp)$ to denote the objective function in the above optimization.
\end{customthm}
Essentially, we derive algebraic connections between the objective function and the Bernstein basis polynomial as well as its derivatives.
Then, we write a condition for a symmetric strategy profiles to become an equilibrium.
Manipulating the derived equation and carefully combining with the platform's original objective function, we derive the reduced optimization formulation without introducing the equilibrium itself.
The proof is deferred to Appendix~\ref{apd:proof:prop-user} along with other necessary technical lemmas.

\subsection{$\HM$ can be Arbitrarily Bad}
An immediate implication of Theorem~\ref{thm:prop-user} is that in the single-minded case of $\alpha = 0$ with convex cost function $\beta > 1$, $\HM$ can be arbitrarily worse than the optimal policy as $n$ increases.
\begin{proposition}[Price of $\HM$ policy]\label{prop:price-hardmax}
    Consider convex cost function of $\beta > 1$ and the RS maximizes average quality, \ie $\alpha  = 0$.
    Let $G(\bp)$ the objective function obtained in Theorem~\ref{thm:prop-user}, \ie if $\alpha = 0$ then $G(\bp) = \int_0^1 h(x,\bp)^{1/\beta}$.
    Then, for any $\eps > 0$, there exists a sufficiently larger $n$ such that $\HM$ policy becomes arbitrarily bad than the optimal policy, \ie  $\nicefrac{G(\HM)}{G(\OPT)} < \eps$.
\end{proposition}
In particular, the proof follows from comparing $\HM$ with $\UNI$, which is deferred to Appendix~\ref{proof:price-hardmax}.

\section{Optimal Policy Structure for Single-Minded RS}\label{sec:opt-structure}
We begin by presenting our main result, which characterizes the optimal recommendation policy in the single-minded setting—where the RS aims solely to maximize either user welfare or platform quality. Notably, the result shows that fairness, in the form of the uniform exposure policy $\UNI$, can emerge endogenously when content creation costs are convex and the platform prioritizes content quality ($\alpha = 0$).

\begin{customthm}{\ref{thm:opt-single}}[Optimal Policy in Boundary Regimes]
Depending on the weight parameter $\alpha$ and the cost parameter $\beta$, the following holds:
\begin{enumerate}
    \item If $\alpha = 1$ or $\beta \le 1$, then $\HM$ is optimal.
    \item If $\alpha = 0$ and $\beta < 1$, then $\HM$ is optimal.
    \item If $\alpha = 0$ and $\beta > 1$, then $\UNI$ is optimal.
    \item If $\alpha = 0$ and $\beta = 1$, then any $\bp \in \Delta$ with $p_n = 0$ is optimal.
\end{enumerate}
\end{customthm}

To understand this result, we analyze the RS’s objective in the single-minded case as a function of the policy vector $\bp$, and reveal key structural properties using tools from majorization theory.

We first introduce a symmetric extension of the function $h(x, \bp) = \sum_{i=1}^n a_i(x)p_i$, defined over the entire $n$-dimensional probability simplex $\Delta_n$, not necessarily ordered. Let $O(x, \bp) = \sum_{i=1}^n a_i(x)p_{(i)}$, where $p_{(i)}$ is the $i$-th largest component in $\bp$.

This construction allows us to study $O(x, \bp)^r$ as a symmetric function, which we relate to the concepts of majorization and Schur-convexity.

\begin{definition}[Symmetric function]
A function $f: X^n \to \R$ is symmetric if $f(x_1,\ldots, x_n) = f(x_{\sigma(1)}, \ldots, x_{\sigma(n)})$ for any permutation $\sigma$ and any $x \in X^n$.
\end{definition}

\begin{definition}[Majorization]
Given $x, y \in \R^n$, we say $x$ majorizes $y$ ($x \succ y$) if:
\begin{align*}
    \sum_{i=1}^k x_{(i)} &\ge \sum_{i=1}^k y_{(i)} \quad \text{for all } k = 1, \ldots, n-1, \\
    \sum_{i=1}^n x_i &= \sum_{i=1}^n y_i,
\end{align*}
where $x_{(i)}$ and $y_{(i)}$ are sorted in non-increasing order.
\end{definition}

\begin{definition}[Schur-convexity]
A symmetric function $f: \R^n \to \R$ is Schur-convex if $x \succ y$ implies $f(x) \ge f(y)$, and Schur-concave if $f(x) \le f(y)$.
\end{definition}

We use the following criterion to check Schur-convexity:

\begin{lemma}[Schur-Ostrowski Criterion~\citet{peajcariaac1992convex}]\label{thm:schur-ostrowski}
Let $f: \R^n \to \R$ be symmetric and continuously differentiable. Then $f$ is Schur-convex (concave) if and only if:
\begin{align*}
    (x_i - x_j)\left(\frac{\partial f}{\partial x_i} - \frac{\partial f}{\partial x_j}\right) \ge (\le)\; 0,
\end{align*}
for all $i, j$.
\end{lemma}

Our main structural result is stated below.

\begin{lemma}[Schur Properties of Objective] \label{thm:schur}
For any $\bp \in \Delta_n$, the function $\int_0^1 O(x, \bp)^r dx$ is:
\begin{itemize}
    \item Schur-concave if $r \in [0,1]$;
    \item Schur-convex if $r \le 0$ or $r \ge 1$.
\end{itemize}
\end{lemma}

Theorem~\ref{thm:opt-single} follows from applying Lemma~\ref{thm:schur} to the RS objective in the respective cases of $\alpha$ and $\beta$.
A full proof is provided in Appendix~\ref{apd:thm:opt-single}.

\section{Optimal Policy Structure for Convex-minded RS}\label{sec:opt-structure-general}


In this section, we characterize the optimal recommendation policy when the RS has a convex-minded objective, i.e., it balances user welfare and platform quality with a weight parameter $\alpha \in (0,1)$. Let $G(\bp)$ denote the objective function associated with a given $\alpha$ as defined in $\Opt$ (see Theorem~\ref{thm:prop-user}).

We restrict the domain of $G$ to the ordered simplex $\Delta$, as introduced in Section~\ref{sec:model}, to avoid confusion with its symmetric extension $O(x,\bp)$. From Theorem~\ref{thm:prop-user}, we know that for $n=2$, the optimal policy always satisfies $p_n = 0$, implying that the optimal policy in this case is $\HM$. Therefore, we focus on the case where $n \ge 3$.

We now state our main result:
\begin{customthm}{\ref{thm:opt-general}}[Optimal Structure for Convex-Minded RS]
    For any $\alpha \in [0,1]$ and $\beta > 0$, any optimal policy has the structure:
    \[
        p_1 \ge p_2 = p_3 = \ldots = p_{n-1} \ge p_n = 0.
    \]
\end{customthm}

\subsection{Proof Sketch of Theorem~\ref{thm:opt-general}}

We prove this theorem by analyzing the geometry of the partial derivatives $\partial G(\bp)/\partial p_i$ for $i \in [n-1]$ and applying the KKT conditions to characterize the structure of optimal solutions.

\subsubsection{Step 1: Quasiconvexity of the Gradient Sequence} 
Define the partial derivatives as:
\begin{align} \label{eq:d_i}
    d_i := \frac{\partial G(\bp)}{\partial p_i} =  n\alpha \left(1+ \frac{1}{\beta}\right) \int_0^1 h(x,\bp)^{1/\beta} a_i(x) dx + \frac{1-\alpha}{\beta} \int_0^1 h(x,\bp)^{1/\beta-1} a_i(x) dx,
\end{align}
for $i \in [n-1]$.
Factoring out common factors, this can be written as
\begin{align}
    d_i = \int_0^1 \paranl{n\alpha\parans{1 + \frac{1}{\beta} h(x,\bp)^{1/\beta} + \frac{1-\alpha}{\beta} h(x,\bp)^{1/\beta - 1}}} a_i(x)dx.\label{eq:d_i-factored}
\end{align}

In fact, notice that for any general objective function of the form
\begin{align*}
    H(\bp) = e_1 \int_0^1 h(x,\bp)^{r_1}dx + e_2 \int_0^1 h(x,\bp)^{r_2}dx,
\end{align*}
for $e_1, e_2,r_1,r_2 \in \R$,
we can write its partial derivatives with respect to $p_i$'s as:
\begin{align}\label{eq:part-deriv-general}
    \frac{\partial H(\bp)}{\partial p_i} = \int_0^1 q(x) a_i(x)dx,
\end{align}
for some $q(x)$ which involves terms $n, \alpha, \beta$ and $h(x,\bp)$.
Importantly, $q(x)$ in the integrand remains the same across $i \in [n-1]$, which will play a key role in our analysis.

Looking forward, we will show that the sequence $(d_i)_{i \in [n-1]}$ exhibits quasiconveixty (along with some further properties), defined below:

\begin{definition}[Quasiconvexity]\label{def:quasiconv}
    A sequence $x_1,\ldots, x_n$ is \emph{quasiconvex} if there exists $k \in [n]$ such that $x_1 \ge x_2 \ge \ldots \ge x_k \le x_{k+1} \le \ldots \le x_n$. 
    Similarly, a function $f: X \subseteq \R \to \R$ is quasiconvex if there exists $x^* \in \R$ such that $f(x) \ge f(y)$ for $x < y \le x^*$ and $f(x) \le f(y)$ for $x^* \le x < y$.
\end{definition}
In fact, we will prove a rather general statement such that if the partial derivatives of the objective function with respect to $p_i$'s can be written as~\eqref{eq:part-deriv-general}, and if the corresponding weight $q(x)$ behind each $a_i(x)$ is quasiconvex, then the resulting sequence of partial derivatives is highly structured.
This is formally captured by the following key structural lemma:
\begin{figure}[t]
    \centering
    \subfigure[$d_i$ under $\HM$ policy]{
    \includegraphics[scale=0.52]{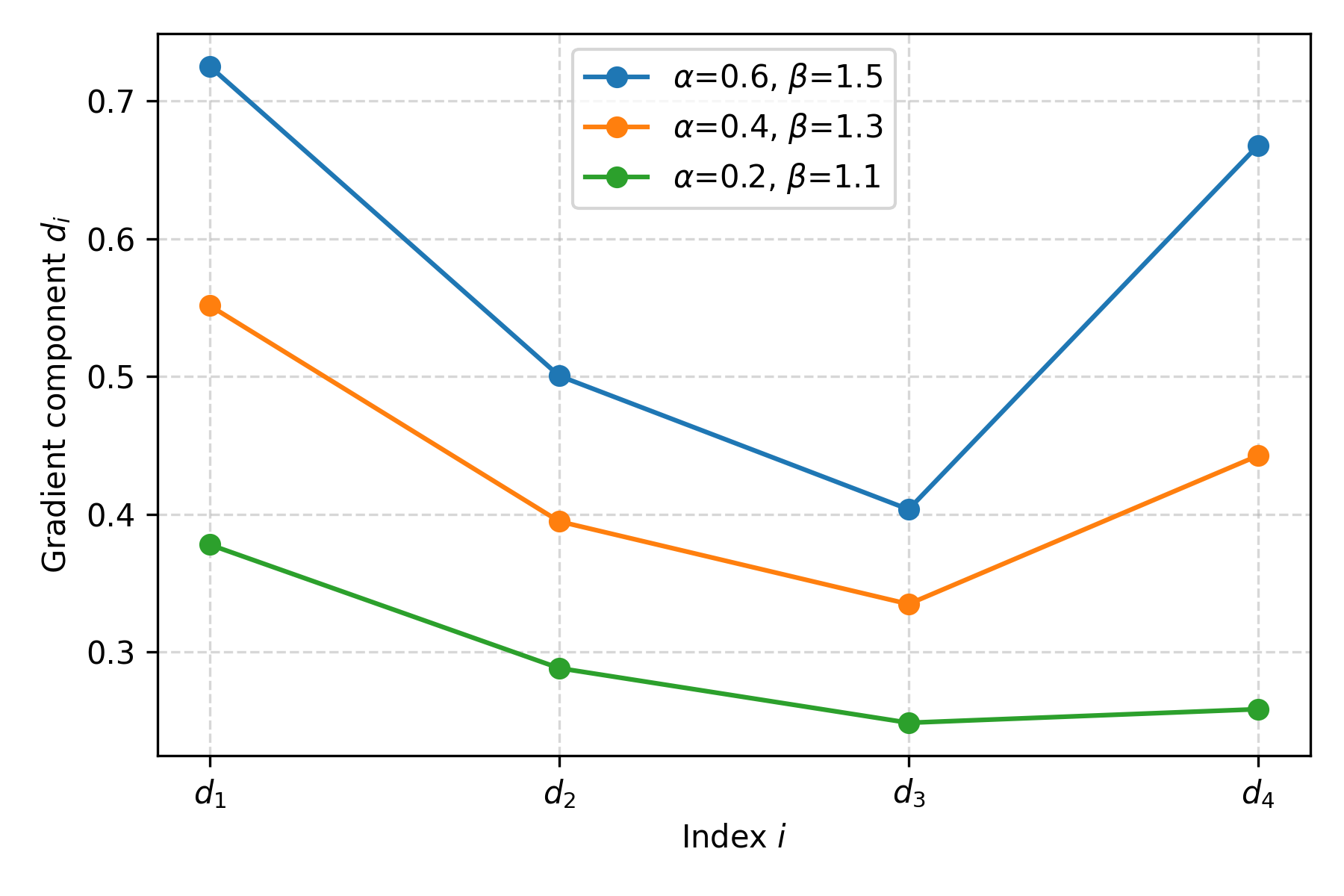}
    }  
    \subfigure[$d_i$ under $\UNI$ policy]{
    \includegraphics[scale=0.52]{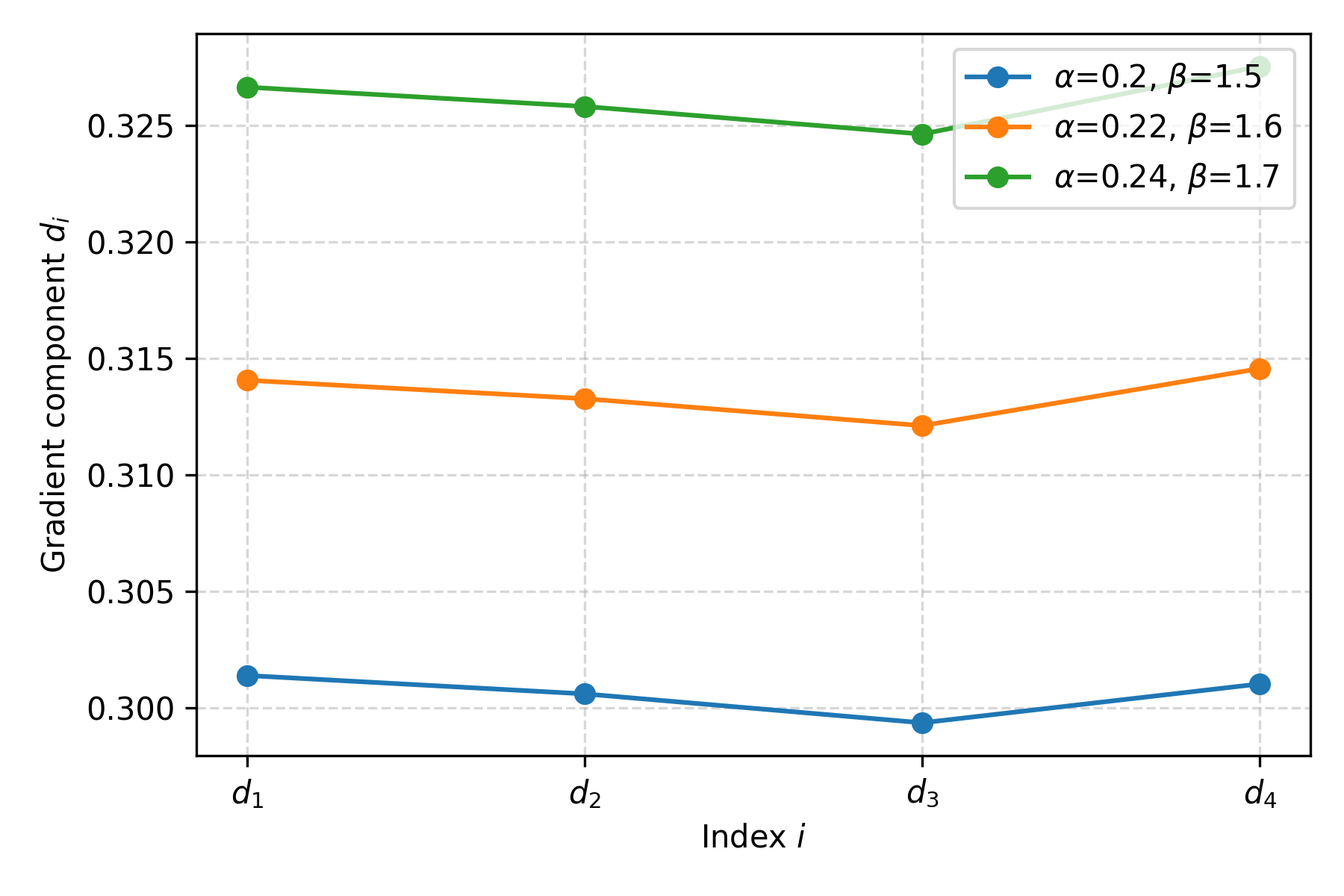}
    }
    \caption{Behavior of $d_i$ for $n = 5$. Since $p_5 = 0$, $d_5$ is omitted. Quasiconvexity of $d_i$ follows from the quasiconvexity of $q(x)$ shown in Figure~\ref{fig:quasi-conv-q}.}
    \label{fig:quasi-conv-d}
\end{figure}

\begin{lemma}[Gradient Quasiconvexity and Structure]\label{lm:unimodal}
Given any general objective function $H(\bp)$, suppose its partial derivatives with respect to $p_i$'s can be written as~\eqref{eq:part-deriv-general}.
If $q(x)$ in the integrand of~\eqref{eq:part-deriv-general} is quasiconvex on $x \in [0,1]$, then the sequence $d_1, d_2, \ldots, d_{n-1}$ satisfies:

\begin{enumerate}
    \item It is quasiconvex.
    \item $d_i = d_{i+1}$ can only happen in one of the following cases: (i) if $k$ is the slope-changing index where $p_{k-1} \ge p_k \le p_{k+1}$, then $d_i = d_{i+1}$ only if $i = k-1$ or $i = k$, (ii) if the entire sequence $d_i$'s for $i \in [n-1]$ is monotone decreasing, then $d_i = d_{i+1}$ only if $i = n-2$, and (iii) if the entire sequence $d_i$'s for $i \in [n-1]$ is monotone increasing, then $d_i = d_{i+1}$ only if $i = 1$ or the entire sequence is monotone decreasing, in which case $d_{n-2} = d_{n-1}$.
    
    
    \item If $d_k = d_{k+1}$, then $d_k \ne d_{k-1}$, and vice versa.
\end{enumerate}
\end{lemma}

We prove Lemma~\ref{lm:unimodal} using tools from total positivity and the variation diminishing property.
Essentially, totally positive operators ($a_i(x)$'s in our case) do not increase the number of sign changes, so the quasiconvexity in the operand $q(x)$ over $x$ translates to the quasiconvexity of $d_i$'s over $i$.

Finally, combined with the following lemma stating the quasiconvexity of $q(x)$ in our objective function $G(\bp)$, we conclude that $d_i$'s has the desired properties.
\begin{lemma}\label{lm:q-quasiconvex}
    The function
    \begin{align*}
        q(x) = \paranl{n\alpha\parans{1 + \frac{1}{\beta} h(x,\bp)^{1/\beta} + \frac{1-\alpha}{\beta} h(x,\bp)^{1/\beta - 1}}}
    \end{align*}
    in~\eqref{eq:d_i-factored} is quasiconvex on $x \in [0,1]$.
\end{lemma}

\subsubsection*{Total Positivity and Variation Diminishing Properties}

We briefly review key definitions:

\begin{definition}[Total Positivity]\label{def:tp}
    Let $K(x,y)$ be a real-valued function over ordered sets $X$ and $Y$. It is totally positive of order $r$ ($\TP_r$) if for all $1 \le k \le r$ and increasing sequences $x_1 < \ldots < x_k$ in $X$, and $y_1 < \ldots < y_k$ in $Y$,
    \[
        \det \left(K(x_i, y_j)\right)_{i,j=1}^k \ge 0.
    \]
    If the inequality is strict, it is strongly totally positive of order $r$ ($\STP_r$).
\end{definition}
The notion of total positivity is closely related to how much sign change occurs in a sequence or a function.
The number of sign change is formally defined as follows:

\begin{definition}[Number of Sign Changes]
    For $x = (x_1, \ldots, x_n) \in \R^n$, let $S^-(x)$ be the number of strict sign changes, ignoring zeroes. $S^+(x)$ allows arbitrary signs for zero entries, maximizing the number of sign changes.
\end{definition}

The following notion of variation diminishing property formally defines the connection between the number of sign changes in a resulting function, which is calculated by a weighted summation over an original function given a two-dimensional weight kernel.
\begin{definition}[Variation Diminishing Property]\label{def:vd}
    Let $K(x,y)$ be an integrable kernel and $q(x)$ an integrable function. Define
    \[
        Q(y) = \int K(x,y)q(x)dx.
    \]
    Then $K$ has the variation diminishing (VD) property if $S^-(Q) \le S^-(q)$ and the strong variation diminishing (SVD) property if $S^+(Q) \le S^-(q)$. Further, if the equation holds, then the sign patterns should remain the same.
\end{definition}

Essentially, the following seminal result asserts that any totally positive kernel cannot change the number of sign changes dramatically, allowing the variation diminishing property:

\begin{theorem}[\cite{karlin1964total}, Chapter 5]\label{thm:karlin}
    If $K(x,y)$ is $\TP_r$, then it has VD for any $q$ with $S^-(q) \le r-1$. If $K$ is $\STP_r$, then it has SVD for any $q$ with $S^-(q) \le r-1$.
\end{theorem}

To relate this to quasiconvexity, we use the following lemma:

\begin{lemma}[\cite{karp2024unimodality}]\label{lm:unimodal-to-signchange}
A function $g: I \to \R$ is quasiconvex iff for every $\lambda \in \R$, the function $f_\lambda(x) := g(x) - \lambda$ has at most two strict sign changes, with a sign pattern $(+,-,+)$ when the bound is tight.
\end{lemma}

We also use the following result on generalized Vandermonde matrices:

\begin{theorem}[\cite{yang2001generalization}]\label{thm:vandermonde}
Let $x_1 < \ldots < x_n$ be positive real numbers and $a_1 > \ldots > a_n$ be integers. Then,
\[
    \det \left(x_i^{a_j}\right)_{i,j=1}^n > 0.
\]
\end{theorem}

We defer the full proof of Lemma~\ref{lm:unimodal} and Lemma~\ref{lm:q-quasiconvex} to Appendix~\ref{apd:lm:unimodal}.
Intuitively, the quasiconvexity of $q(x)$ in Lemma~\ref{lm:q-quasiconvex} translates to the quasiconvexity of $d_i$'s, thanks to the total positivity of our kernel function $a_i(x)$.
This is illustrated in Figure~\ref{fig:quasi-conv-q}.

\begin{figure}[t]
    \centering
    \subfigure[$q(x)$ under $\HM$ policy]{
    \includegraphics[scale=0.52]{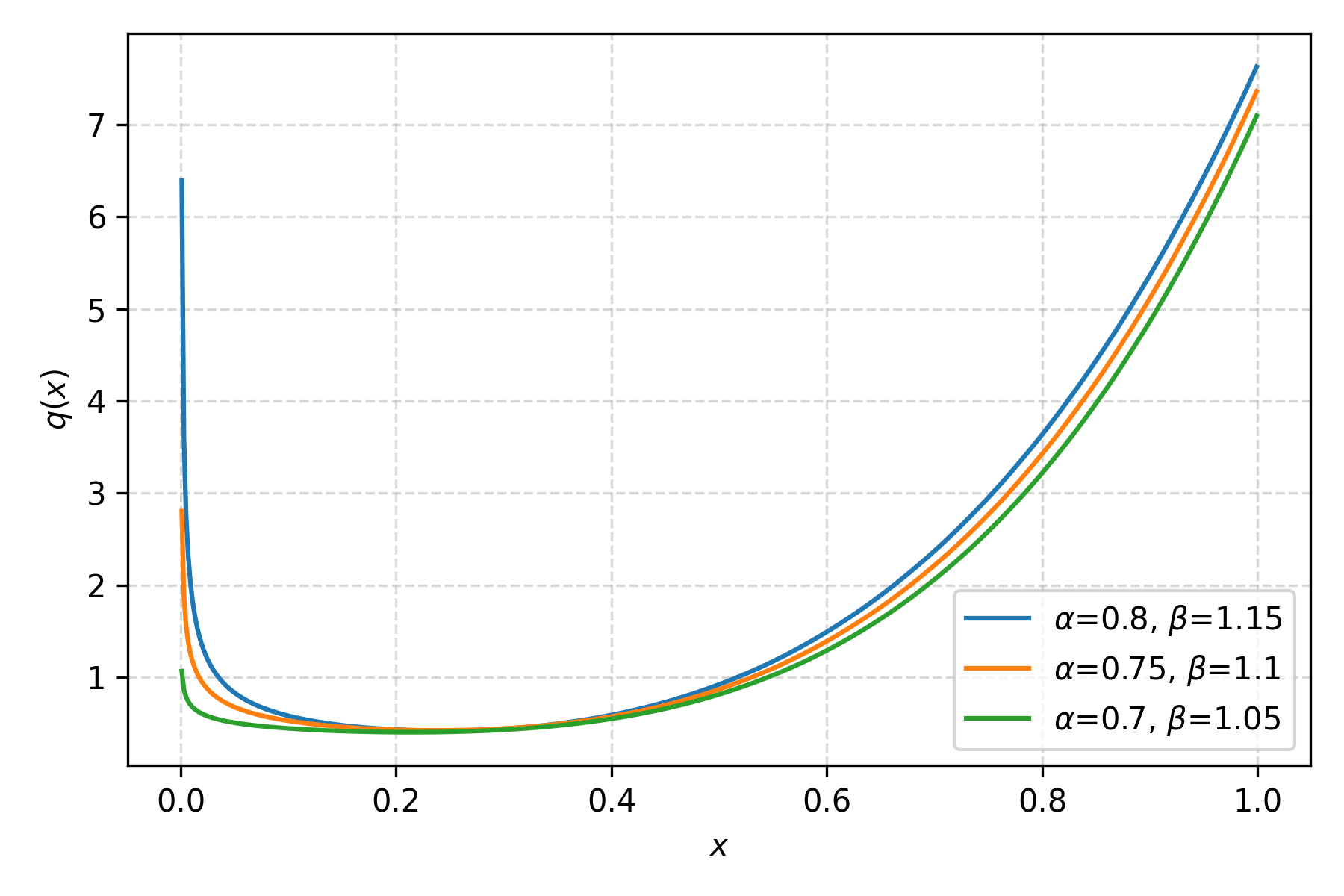}
    }  
    \subfigure[$q(x)$ under $\UNI$ policy]{
    \includegraphics[scale=0.52]{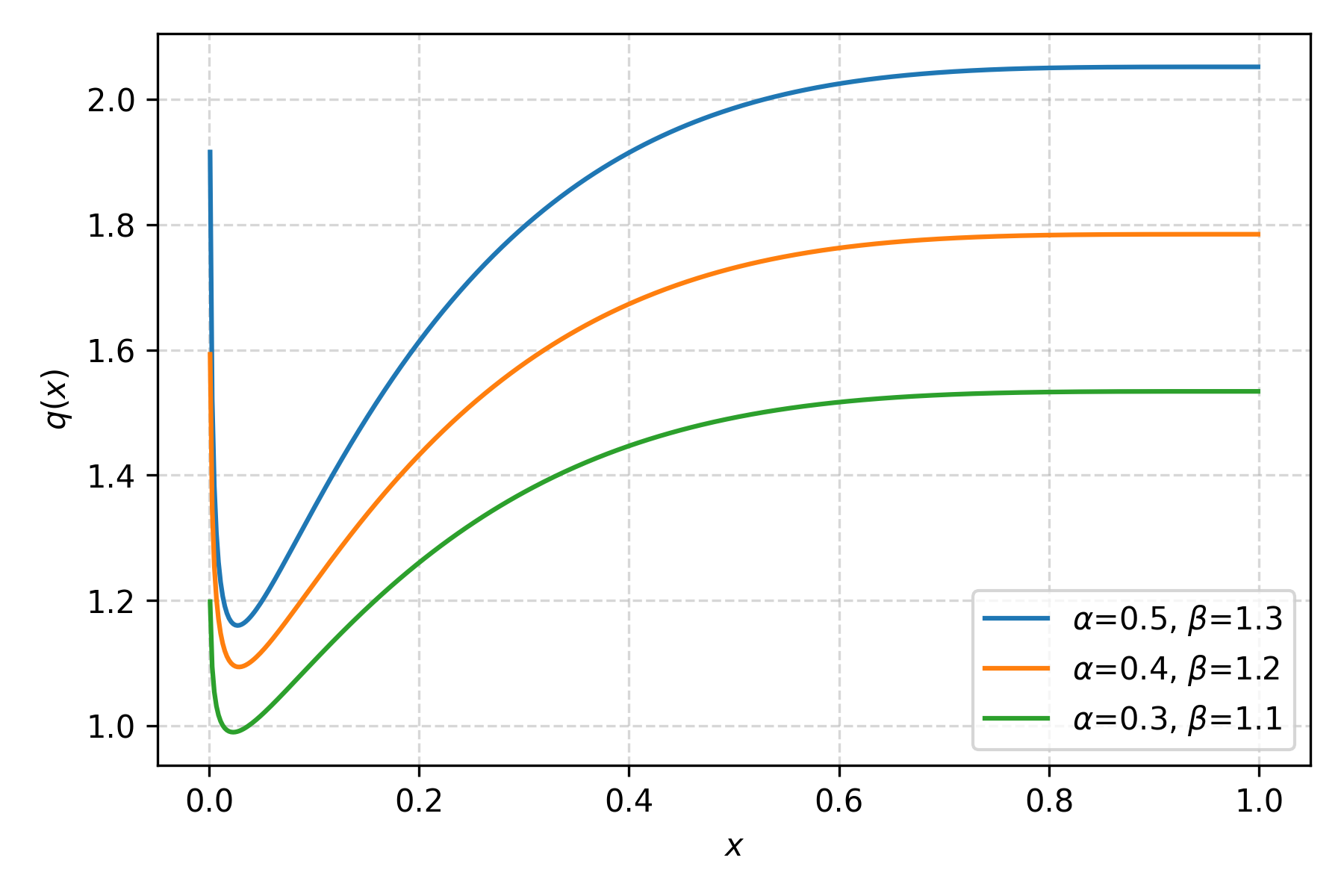}
    }
    \caption{Behavior of $q(x)$ for $n = 5$ (see Appendix~\ref{apd:thm:opt-general} for definition). The quasiconvexity of $q(x)$ transfers to $d_i$ due to variation diminishing properties of the kernel defined by $a_i(x)$.}
    \label{fig:quasi-conv-q}
\end{figure} 

\subsubsection{Step 2: Applying KKT Conditions}

With Lemma~\ref{lm:unimodal} in hand, we analyze the KKT conditions for the optimization problem $\Opt$. These involve the derivatives $d_i$, the probabilities $p_i$, and associated dual variables. Using the structure guaranteed by Lemma~\ref{lm:unimodal}, we show via a case analysis that any deviation from the structure in Theorem~\ref{thm:opt-general} would violate the KKT conditions. In particular, the transition in probability weights must occur only at indices $1$ and $n$.

The full proof is provided in Appendix~\ref{apd:thm:opt-general}.

\subsection{Implication on High-dimensional Nonconvex Optimization Problem}
Recall that Lemma~\ref{lm:unimodal} holds for any objective function whose partial derivatives are of the form~\eqref{eq:part-deriv-general}, and its sequence is quasiconvex.
This immediately implies an interesting corollary on a high-dimensional nonconvex optimization.
Specifically, consider a monotone increasing nonnegative function $f:\R^n \to \R_{\ge 0}$ that may possibly be nonconvex, a constant $C > 0$, and the following constrained optimization problem:
\begin{align*}
    \begin{aligned}
        (\text{\Gen})\quad& \underset{\bx \in \R^n}{\text{maximize}}
        & & f(\bx)\\
        & \text{subject to}
        & & x_1\ge x_2 \ge \ldots \ge x_n \ge 0
        \\
        & & & \sum_{i=1}^n x_i = C
        \\
        & & & 
    \end{aligned}
\end{align*}
Then, by following the proof steps of Theorem~\ref{thm:opt-general}, one can obtain the following general result.
\begin{customcor}{\ref{thm:opt-generic}}
    Consider the constrained optimization problem \Gen.
    Suppose that the gradient of $f$ with respect to each $x_i$ can be written as the following:
    \begin{align*}
        \frac{\partial f(\bp)}{\partial x_i} = \int_D w(y,\bp)g_i(y) dy,
    \end{align*}
    for a continuous bounded interval $D$, a function $g_i:D \to \R_{\ge 0}$ for $i \in [n]$, and another function $w(y,\bp): D \times \R^n \to \R$.
    Assume further that $g_i(y)$ is $\STP_{n}$,\footnote{We refer to Definition~\ref{def:tp} for more details.}, \ie strongly totally positive of order $n$, and $w(y,\bp)$ is quasiconvex over $y \in D$, \ie it decreases then increases, given any $\bp$.
    Then, it follows that the optimum of $f$ satisfies $x_1 \ge x_2 = \ldots = x_n$.
\end{customcor}
We omit the proof as it follows from the exactly same proof structure of Theorem~\ref{thm:opt-general}, equipped with required lemmas.

The implication of this structural characterization is significant, since once we know that $x_1 \ge x_2 = \ldots = x_n = x$, due to the monotonicity of $f$, it is immediate that $x_1 = C - (n-1)x$.
Thus, the optimization can be reduced to the following.
\begin{align*}
    \begin{aligned}
        (\text{\Gen'})\quad& \underset{x \in \R}{\text{maximize}}
        & & f(C-(n-1)x, x, x, \ldots, x)\\
        & \text{subject to}
        & & 0 \le x \le C/(n-1)
        \\
        & & & 
    \end{aligned}
\end{align*}
Notably, this is a single-dimensional optimization problem over $x$ on a bounded interval, which typically admits efficient algorithms when $f$ exhibits desirable properties such as smoothness or Lipschitzness.

\section{Computing Optimal Policy}\label{sec:computation}
Now we know that any optimal policy satisfies $p_1\ge p_2=\cdots=p_{n-1}\ge p_n=0$, we have $p_1=p_2=\cdots=p_{n-1}=\frac{1-p_1}{n-2}$ and $p_n=0$, so the computation of the optimal policy reduces to find $p_1\in\Big[\tfrac{1}{n-1},1\Big]$.
Let us define $\cD = [1/(n-1),1]$, and write $h(x,p_1)$ to denote the corresponding $h(x,\bp)$ given the structural relation.
Recall that $a_i(x)=\binom{n-1}{i-1}x^{n-i}(1-x)^{i-1}$, $h(x,p_1)=\sum_{i=1}^{n-1}a_i(x)p_i$, and then the objective function becomes
\begin{align*}
    G(p_1)= \alpha n\int_0^1 h(x,p_1)^{1+1/\beta}dx+(1-\alpha)\int_0^1 h(x,p_1)^{1/\beta}\,dx,
\end{align*}
where $h(x,p_1)=c_0(x)+c_1(x)p_1$ with
\[
c_0(x)=\frac{1-a_1(x)-a_n(x)}{n-2},\qquad
c_1(x)=a_1(x)-\frac{1-a_1(x)-a_n(x)}{n-2}.
\]

We will show that the resulting optimization over a single variable $p_1$ can be efficiently solved using branch-and-bound algorithms, despite its nonconvexity.
To this end, we first identify a lower bound and upper bound for our objective functions given an interval.

For an interval $I=[\ell,u]\subseteq\mathcal{D}$ define the lower and upper bounds on the
\emph{objective value} over $I$ by
\begin{align*}
L(I)&:=\max\{G(\ell),\,G(u)\},\\
U(I)&:=\alpha n \int_0^1\max\{h(x,\ell)^{1+1/\beta},h(x,u)^{1+1/\beta}\}\,dx
\\&\qquad\qquad+(1-\alpha)\int_0^1\max\{h(x,\ell)^{1/\beta},h(x,u)^{1/\beta}\}\,dx.
\end{align*}
We then first show that these constitute upper and lower bounds on the objective function.
\begin{lemma}\label{lem:computation-validity}
For every $I=[\ell,u]\subseteq\mathcal{D}$,
\[
L(I)\ \le\ \sup_{p_1\in I}G(p_1)\ \le\ U(I).
\]
\end{lemma}

The above bounds will be crucial in our algorithm.
Let us write $|I| := u - \ell$ to denote the length of the interval.
The following lemma derives a Lipschitz-style of bounds on the difference $U(I) - L(I)$.
\begin{lemma}\label{lem:gap}
For every interval $I=[\ell,u]\subseteq\mathcal D$,
\[
U(I)-L(I)\ \le\ C_1\,|I|\ +\ C_2\,|I|^{1/\beta},
\]
where
\[
C_1=\alpha n\Big(1+\tfrac{1}{\beta}\Big)\!\int_0^1 |c_1(x)|\,dx,
\qquad
C_2=(1-\alpha)\!\int_0^1 |c_1(x)|^{1/\beta}\,dx.
\]
\end{lemma}

\begin{remark}
    Our algorithm and analysis will depend on parameters $C_1$ and $C_2$.
    In fact, instead of computing the exact values, one can replace them with rough upper bounds.
    This would only incur a constant factor blowup in the running time.
    Formally, using  $|c_1|\le a_1+\frac{1-a_1-a_n}{n-2}$, we obtain
\[
\int_0^1 |c_1(x)|\,dx\ \le\ \frac{1}{n}+\frac{1-2/n}{\,n-2\,}=\frac{2}{n}.
\]
Note that by $(x+y)^{\rho}\le x^{\rho}+y^{\rho}$ when $\rho\in(0,1]$ and $a_1^{1/\beta}$ integrable,
\[
\int_0^1 |c_1(x)|^{1/\beta}\,dx\ \le\ \frac{\beta}{\beta+n-1}+\frac{1}{(n-2)^{1/\beta}}.
\]
Thus one may take
\[
C_1' =  2\alpha\Big(1+\tfrac{1}{\beta}\Big),\qquad
C_2' = (1-\alpha)\!\left(\frac{\beta}{\beta+n-1}+\frac{1}{(n-2)^{1/\beta}}\right),
\]
and follow the analogous arguments.
\end{remark}

Our main algorithm is presented in Algorithm~\ref{alg-bnb}.
To analyze its correctness and time complexity, we first verify that the optimum is always bounded.
\begin{lemma}\label{prop:invariants}
At every iteration: (i) $L^\star\le \mathrm{OPT}:=\max_{p_1\in\mathcal D}G(p_1)$; (ii) $\max_{I\in\mathcal A}U(I)\ge\mathrm{OPT}$.
\end{lemma}
\begin{algorithm}[h]\label{alg-bnb}
\caption{Active-Set Branch-and-Bound}
\DontPrintSemicolon
\KwIn{Tolerance $\varepsilon > 0$}
\KwOut{$\eps$-optimal policy $\bp^\star$}
Initialize the active set $\mathcal A \leftarrow \{I_0\}$ with $I_0 = [1/(n-1),1]$\;
For each $I \in \mathcal A$, compute $L(I), U(I)$\;
Set $L^\star \leftarrow \max_{I \in \mathcal A} L(I)$ and record endpoint $p^\star$ attaining $L^\star$\;
\While{$\max_{I \in \mathcal A} U(I) > L^\star + \varepsilon$}{
    Pick $I = [\ell, u] \in \mathcal A$ with the largest $U(I)$\;
    \eIf{$U(I) \le L^\star + \varepsilon$}{
        Remove $I$ from $\mathcal A$ (prune)\;
    }{
        Bisect $I$ at midpoint $m = (\ell + u)/2$ into $I_1 = [\ell, m]$ and $I_2 = [m, u]$\;
        Compute $L(\cdot), U(\cdot)$ for $I_1, I_2$ and replace $I$ by $I_1, I_2$ in $\mathcal A$\;
        Update $L^\star$ and $p^\star$ if any child’s $L(\cdot)$ improves the incumbent\;
    }
}
\Return{$\bp^\star = (p^\star, (1-p^\star)/(n-1),\ldots, (1-p^\star)/(n-1), 0)$ }\;
\end{algorithm}
For $\eps$-additive error, we ultimately require the following number of bisecting intervals.
\begin{definition}\label{def:delta}
Given $\varepsilon>0$, define
\[
\delta(\varepsilon):=\min\Big\{\frac{\varepsilon}{C_1}\,,\ \Big(\frac{\varepsilon}{C_2}\Big)^{\!\beta}\Big\}.
\]
\end{definition}
Finally, we can show that the algorithm terminates in logarithmic over $1/\eps$.
\begin{theorem}\label{thm:rate}
Consider a value oracle that exactly computes our objective function given a policy. 
Define
\[
C_1=\alpha n\Big(1+\tfrac{1}{\beta}\Big)\!\int_0^1 |c_1(x)|\,dx,
\qquad
C_2=(1-\alpha)\!\int_0^1 |c_1(x)|^{1/\beta}\,dx,
\qquad 
D = 1-\frac{1}{n-1}.
\]
The algorithm halts and returns $p^\star$ with
$G(p^\star)\ge \mathrm{OPT}-\varepsilon$. Moreover, if an interval has undergone $k$ bisections then its
length is at most $D\,2^{-k}$, and it suffices that
\[
k\ \ge\ \max\!\left\{\log_2\!\frac{C_1D}{\varepsilon}\,,\ \beta\log_2\!\frac{C_2D^{1/\beta}}{\varepsilon}\right\}
\]
to guarantee that every active interval has length $\le\delta(\varepsilon)$.
Further, the total number of nodes explored is $O\!\Big(\frac{C_1}{\varepsilon}\Big)\ +\ O\!\Big(\frac{C_2}{\varepsilon}\Big)^{\!\beta}$, \ie the overall running time is polynomial in $1/\eps$.
\end{theorem}

\sscomment{Defer this to the appendix}
\paragraph{Approximating integrals}
Note that our value oracle requires exactly calculating integrals, which necessarily entails numerical errors.
That being said, if each integral is approximated with absolute error at most $\delta$,
the returned policy $\bp^{\star}$ satisfies 
\[
G(\bp^\star)\ \ge\ \mathrm{OPT}-(\varepsilon+2\delta).
\]
Thus, setting $\delta\le \varepsilon/4$ and running the same algorithm with target $\varepsilon/2$ would keep the overall error less than $\eps$.

Thus, the natural question is whether we can efficiently compute the approximate integrals.

\begin{lemma}[Approximate integrals]\label{lem:riemann}
Fix $\beta>0$ and $p\in\Delta$.
Let $r\in\{1/\beta,\,1+1/\beta\}$ and define $f_r(x):=h(x,p)^r$.
Then $f_r:[0,1]\to[0,1]$ is monotone increasing\footnote{By Lemma~\ref{lm:h-monotone}, $h(\cdot,\bp)$ is strictly increasing for any nontrivial $\bp$, hence $f_r(\cdot)$ is increasing for $r>0$.} and for any integer $m\ge1$,
\[
\Bigg|\int_0^1 f_r(x)\,dx-\frac{1}{m}\sum_{j=1}^{m} f_r\!\Big(\tfrac{j}{m}\Big)\Bigg|
\;\le\;\frac{f_r(1)-f_r(0)}{m}\;\le\;\frac{1}{m}.
\]
In particular, choosing $m=\lceil1/\varepsilon\rceil$ yields a deterministic additive-$\varepsilon$ approximation to $\int_0^1 h(x,\bp)^r dx$ using $m$ evaluations of $h(x,\bp)^r$.\footnote{Note that for multiplicative approximation, one can stop when $\max_I U(I)\le \frac{1}{1-\eta}L^\star$ yields 
$G(\bp^{\star})\ge (1-\eta)\mathrm{OPT}$.}
\end{lemma}
Noting that computation of $h(x,\bp)^r = h(x,p_1)^r = (c_0(x) +c_1(x)p_1)^r$ with $r$ as a function of $\beta$ only requires standard arithmetic operations running in polynomial over $n$ given a fixed $\beta$, the overall running time is polynomial in $n$ and $1/\eps$.

\sscomment{Defer this to the appendix}
\section{Beyond RS: General Objective Functions}\label{sec:general}
In hindsight, we notice that the proof of Lemma~\ref{lm:unimodal} only requires the platform's objective function to have certain property, where our objective function $G(\bp)$ as stated in Theorem~\ref{thm:prop-user} turns out to meet this condition.
This inspires us to derive a significantly more general condition under which our analysis on the characterization of optimal policies carries over.
Interestingly, we observe that our analysis holds for any posynomial objective function with nonnegative exponents.\footnote{Posynomial refers to a function which is a positive combination over monomials. With a single variable, it can be viewed as a polynomial with nonnegative coefficients whose exponent can possibly be noninteger.}
\begin{customthm}{\ref{thm:general-obj}}[Posynomial objective function]
	Given $k_m > k_{m-1} > \ldots > k_1$,  consider the following generalized objective function at symmetric MNE $\bmu$ given $\bp$:
    \begin{align*}
        \cG(\bmu, \bp) = \frac{1}{n}\sum_{i=1}^n\Exu{q_i \sim \mu_i}{e_m q_i^{k_m} + e_{m-1}q_i^{k_{m-1}} + \ldots + e_1 q_i^{k_1}},
    \end{align*}
    where $q_i \sim \mu_i$ for a symmetric strategy profile $\mu_i$ in $\bmu$.
    For any $\beta > 0$, 
    suppose $e_j$'s and $k_j$'s satisfy either of the following property:
    \begin{enumerate}
        \item There exists $j^* \in [n-1]$ such that $e_j(k_j - \beta) \le 0$ for $j \le j^*$ and $e_j(k_j - \beta) \ge 0$ for $j > j^*$.
        \item All the elements in the sequence $e_j(k_j - \beta)$ for $j \in [n-1]$ is nonnegative or nonpositive.
    \end{enumerate}
    Then, the Lemma~\ref{lm:unimodal} holds, \ie  any optimal policy has the same structure as Theorem~\ref{thm:opt-general}, \ie
    \[
        p_1 \ge p_2 = p_3 = \ldots = p_{n-1} \ge p_n = 0.
    \]
    As a special case, if $e_j$'s are all nonnegative, this condition follows.
\end{customthm}
The proof is based on the observation that given the posynomial function, the weights (except $a_i(x)$) in the integrand of its partial derivative with respect to $p_i$'s exhibits quasiconvexity over $x \in [0,1]$, which is precisely the condition for Lemma~\ref{lm:unimodal}.
We defer the proof to Appendix~\ref{apd:general-obj}.


\paragraph{Social welfare}
Our characterization under the general objective function in Theorem~\ref{thm:general-obj} implies several interesting applications.
For instance, let us consider \emph{social welfare} of the contest, which is the summation of the RS's utility and the producers' utilities.
Let $\cV(\bmu, \bp)$ be the platform's utility function at symmetric MNE induced by $\bp$, \eg average quality or a function of it.
Then, at a symmetric MNE $\bmu = (\mu, \ldots, \mu)$, the social welfare can be written as:
\begin{align}
    \cS(\bmu, \bp) &= \cW(\bmu, \bp) + \cV(\bmu, \bp) + \sum_{i\in [n]}\cU_i(\mu, \bmu_{-i},\bp).\label{eq:social-welfare}
\end{align}
Then, Theorem~\ref{thm:general-obj} immediately implies the following corollary:
\begin{corollary}\label{cor:social-welfare}
    Assume the platform's utility function is given by
    \begin{align*}
        \cV(\bmu, \bp) = \Exu{q \in \mu}{e_m q^m + \ldots + e_1 q},
    \end{align*}
    where $e_j \ge 0$ for all $j \in [m]$.
    Suppose the RS aims to maximize the social welfare as written in~\eqref{eq:social-welfare}.
    Then, the optimal policy has the structure $p_1 \ge p_2 = \ldots = p_{n-1} \ge p_n = 0$ at the symmetric MNE.
\end{corollary}

\paragraph{Inverse S-shaped function}
On the other hand, the RS might be interested in a \emph{superstar} outcome of the contest where it realizes a concave utility up to certain quality threshold, and then it realizes a convex utility of increasing marginals.

For instance, assume the convex cost $\beta =2$, and the following utility function:
\begin{align*}
    \Exu{q \sim \mu}{2q^3 -3q^2 + 2q}.
\end{align*}
One can easily verify that this is monotone in $q$, and also satisfies the condition in Theorem~\ref{thm:general-obj}.
Notably, this function is inversely s-shape, \ie  strongly concave-then-convex over $q$.
In words, the RS's utility does not significantly change up to certain threshold quality level, but then it dramatically increases afterwards.
This is in a similar vein to the S-shaped function in prospect theory~\cite{barberis2013thirty} in behavioral economics, but in an inverse manner as standard S-shaped function requires the utility function to be convex-then-concave.
In practice, many research contests that requires innovative outcome would satisfy this property.
Further, one can imagine of a economies of scale, \eg if the outcome of the contest exceeds certain threshold, then the company (contest designer) can monopolize the market to bring increasingly larger marginal revenues.

\paragraph{Order statistics}
Finally, one might be interested in the highest quality contents among the producers, \ie $\cV(\bmu, \bp) = \Ex{q_{(1)}}$ where $q_{(1)}$ is the largest order statistics among $q_1,\ldots, q_n$.
The following corollary suggests that the structure of the optimal policies remains the same for such objective.
\begin{corollary}\label{cor:order-statistics}
    Assume the platform aims to maximize the expected value of the highest quality outcome, \ie $\Exu{q_i \sim \mu_i, i \in [n]}{q_{(1)}}$, where $q_{(1)}$ denotes the largest order statistics.
    Then, the optimal policy has the structure $p_1 \ge p_2 = \ldots = p_{n-1} \ge p_n = 0$ at the symmetric MNE.
\end{corollary}

\paragraph{Exponential utility}
\ssedit{
Even more extreme version of the inverse S-shaped function can be represented as an exponential utility where the utility function is written as $\Exu{q \sim \mu}{e^{\lambda q}}$.
Notice that Theorem~\ref{thm:general-obj} accommodates the exponential utility as shown in the following.
\begin{corollary}\label{cor:exp-utility}
    Assume the platform aims to maximize the exponential utility, \ie $\Exu{q_i \sim \mu_i, i \in [n]}{e^{\lambda q_i}}$, for $\lambda >0$.
    Then, the optimal policy has the structure $p_1 \ge p_2 = \ldots = p_{n-1} \ge p_n = 0$ at the symmetric MNE.
    Further, the same structure holds if the platform aims to maximize a convex combination of exponential utilities $\Exu{q_i \sim \mu_i, i \in [n]}{\sum_{j=1}^m e^{\lambda_j q_i}}$.
\end{corollary}
This follows from uniformly approximating the exponential function with its Taylor series, the coefficients of which satisfy the condition of Theorem~\ref{thm:general-obj}.
}

\section{Conclusion}
We study the optimal contest problem with nonconvex objective functions, motivated by a strategic competition between content producers in online platforms.
We characterize the optimal rank-based recommendation policy to optimally incentivize the strategic content production of producers.
We observe a stark phase transition between $\HM$ that always recommend the highest quality content and $\UNI$ that uniformly randomizes among every content but the lowest quality one when the platform is single-minded to either optimize user welfare or platform quality.
In the convex-minded case when the platform tries to maximize the convex combination over user welfare and platform quality, we reveal that the optimal policy is still highly structured: potentially high probability to the highest quality content, zero probability to the lowest quality content, and equal probabilities to the intermediate quality contents.
We finally discuss that this result holds even more general objective functions, with broader applications beyond recommender systems.

A major open problem is to characterize the structure of the optimal policy when the objective function exhibits more complicated pattern in its gradient sequence beyond quasiconvexity.
Further, it remains also questionable if the structural characterization carries over beyond the complete information setting with homogeneous agents, \eg incomplete information setting with ex-post heterogeneous but ex-interim homogeneous agents~\cite{greenwald2018simple}.



\sscomment{Future direction: (1) more general function with more number of sign changes, and number of value changes in optimal structure. (2) revenue-sharing contest. (3) number of sign changes in coefficients, and approximation theory.}

\section*{Acknowledgements}
This work is partially supported by DARPA expMath, ONR MURI 2024 award on Algorithms, Learning, and Game Theory, Army-Research Laboratory (ARL) grant W911NF2410052, NSF AF:Small grants 2218678, 2114269, 2347322, and MIT Junior faculty research assistance grant.

\bibliographystyle{plainnat}
\bibliography{ref}

\newpage
\appendix

\section{Further Related Works}\label{apd:rel}
\paragraph{Recent works on contest design}
Building on classic rank-based prize design, recent literature has explored a range of extensions, examining how prize structures influence contestant behavior beyond mere effort exertion. For instance, \cite{fang2024winning} show that a convex prize structure can incentivize contestants to select riskier quality distributions. \cite{kirkegaard2023contest} analyze the optimal number of winners to reward with fixed prizes when contestant actions lead to stochastic quality outcomes. 

Taking a different approach, \cite{chawla2019optimal} model crowdsourcing contests with endogenous rewards. \cite{liu2025optimal} study two-stage contests involving initial ability screening followed by rank-based effort decisions. {~\cite{bimpikis2019designing} consider a dynamic setting in which contestants gradually exert effort over time, designing both an efficient reward structure and an information disclosure policy. This work builds on earlier efforts such as~\cite{ridlon2013favoring}, which investigate two-period competitions and prize structures under contestant homogeneity. \cite{acemoglu2014managing} study a matching framework for allocating strategic workers to tasks on a crowdsourcing platform.}

Complementing design, the Tullock contest literature analyzes equilibrium behavior under specific success functions, where each player's probability of winning is proportional to their effort raised to a power—a formulation first introduced by~\cite{tullock2008efficient}. Recent works further analyze its computational complexity~\citep{he2024complexity}.

\paragraph{Stylized models on strategic recommender system}
There are an increasing attention from various communities to study stylzed models to investigate outcomes of strategic behavior of content producers.
\cite{yao2024unveiling} investigate the inherent trade-off between maximizing user satisfaction and maintaining creator productivity, highlighting that aggressive optimization for user metrics can suppress content supply.
\cite{yao2024user} propose a Stackelberg framework to optimize user welfare in the presence of strategic creators, and design efficient algorithms for policy optimization under such endogenous response.
\cite{shin2022multi,shin2025replication} analyze a quantity-competition scenario where content producers may adversarially replicate their own content to abuse the recommendation algorithm, and devise an algorithm that disincentivizes such behavior.
Many other recent works~\citep{ben2018game,ghosh2013learning,zhu2023online,esmaeili2025robust} introduced stylized models, such as learning algorithm with repeated interactions, to study strategic interaction in online platforms, but we will not discuss the details due to the significant difference to our model.
{
In contrast, we consider a simple yet standard problem setup, motivated by the seminal contest design literature, and provide an exact characterization of the optimal policies.
}

\paragraph{Fairness in recommender system} 
{
There are significant interests in studying algorithmic fairness for recommender systems~\citep{chen2024interpolating} as well as for general online platforms~\citep{barocas2016big,abebe2020roles,kasy2021fairness} such as resource allocation~\citep{hooker2012combining}, assortment optimization~\citep{chen2022fair}, online advertising~\citep{deng2024individual} and combinatorial optimization~\citep{golrezaei2024online} in the past few years.
We refer to surveys by~\cite{li2023fairness,wang2023survey,deldjoo2024fairness} for more details in fairness-aware recommender systems.
}
\cite{chen2024interpolating} study the two-sided fairness for both the item and user.
\cite{golrezaei2024online} propose a general framework to incorporate fairness into sequential decision making problem.
\cite{chen2022fair} investigate a fair assortment planning problem to retrieve a set of items with fixed cardinality to maximize expected revenue while satisfying  pairwise fairness constraints.
Our results on the optimality beyond the $\HM$ complement these works by showing that fairness can \emph{emerge endogenously} from optimizing the platform’s own objectives, without needing to be imposed explicitly.

\subsection{Comparison with~\cite{glazer1988optimal}\label{sec:comparison} and~\cite{jagadeesan2024supply}} 
As we have briefly mentioned, our model is closely related to those by~\cite{glazer1988optimal} and~\cite{jagadeesan2024supply}.  
We discuss several technical differences and explain how our model generalizes and subsumes their respective settings.

\paragraph{Comparison with~\cite{glazer1988optimal}}  
Our model generalizes~\cite{glazer1988optimal} by setting $\alpha = 0$.\footnote{More precisely, they consider a utility function over the $p_i$ such that each contestant derives utility $u(p_i)$. This feature is not captured in our model since $p_i$ represents a probability rather than a prize.}  
They consider a rank-based policy to maximize platform quality $\cQ(\bmu, \bp)$, assuming the existence of a symmetric MNE. In particular, they prove the following result for linear costs:

\begin{theorem}[\cite{glazer1988optimal}]
    Consider a linear cost function ($c(q) = q$) and $\alpha = 0$, i.e., the objective is $\cQ(\bmu, \bp)$.  
    Then the optimal policy is $\UNI$.
\end{theorem}

Their proof relies on the following key observations:  
(i) $p_n$ should always be zero to avoid free-riding,  
(ii) in the linear cost case, the objective can be rewritten as a concave function purely in terms of $\bp$, and  
(iii) the $\UNI$ policy satisfies the first-order condition, implying its optimality.  

However, their approach breaks down under nonlinear cost functions, where the objective becomes non-concave and the first-order condition no longer guarantees optimality.

In contrast, we derive the optimal policy for any cost function of the form $c(x) = x^\beta$ for $\beta > 0$, thereby generalizing the linear-cost result of~\cite{glazer1988optimal}.  
Importantly, our analysis does not rely on first-order conditions. Instead, we characterize the structural properties of the objective that lead to optimality. See Section~\ref{sec:opt-structure}.

Finally, note that~\cite{glazer1988optimal} \emph{assume} the existence of a symmetric MNE for any policy, without formally proving it. In contrast, we rigorously prove the existence of a symmetric MNE in Theorem~\ref{thm:exist}, filling this gap in the earlier work.



\paragraph{Comparison with \cite{jagadeesan2024supply}}
Our model, albeit from a different perspective, also subsumes the work of~\cite{jagadeesan2024supply}, who characterize the equilibrium under the specific $\HM$ policy.  
We emphasize that their focus is on analyzing the equilibrium behavior under a given policy, rather than on designing policies to induce desirable outcomes for the platform.

Notably,~\cite{jagadeesan2024supply} begin by quantifying the equilibrium outcome under the $\HM$ policy as follows.\footnote{They also generalize this result to a multi-dimensional setting, which presents an intriguing direction for contest design as well.}

\begin{theorem}[\cite{jagadeesan2024supply}]\label{thm:single-equi-hm}
Under the $\HM$ policy, there exists a unique symmetric equilibrium $\bbmu$ whose cumulative distribution function (c.d.f.) is $F(q) = \parans{q^\beta}^{1/(n-1)}$.
\end{theorem}

The proof of Theorem~\ref{thm:single-equi-hm} relies on specific properties of the $\HM$ policy, which simplify the equilibrium conditions to a first-degree equation in the cumulative distribution function $F(\cdot)$.  
Solving this equation directly yields the desired result.

We further observe that the user welfare induced by the $\HM$ policy is given by the following closed-form expression.

\begin{proposition}\label{thm:single-hm}
    Under the $\HM$ policy, consider the unique symmetric mixed equilibrium $\bbmu$ characterized by Theorem~\ref{thm:single-equi-hm}. Then,
    \begin{align}
        \OBJ(\bbmu; \HM) 
        = \alpha \frac{\beta n}{\beta n + n-1} +  (1-\alpha)\frac{\beta}{\beta + n-1}.\label{eq:sw-hm}
    \end{align}
\end{proposition}

\begin{proof}
    By Theorem~\ref{thm:single-equi-hm}, we have $F(q) = q^{\beta/(n-1)}$.
    Let $F_{\max}$ be the c.d.f. of the largest order statistics over the random contents.
    Then, observe that the induced user welfare of $\HM$ is equivalent to the expectation of the largest order statistics among $q_1,\ldots, q_n$.
    Thus, we can write
    \begin{align*}
        \OBJ(\bbmu; \HM) 
        &= \alpha \Exu{\bq \sim \bmu}{\max_{j \in [n]}q_j} + (1-\alpha) \Exu{q \sim F}{q}
        \\
        &= \alpha\int_{0}^1 (1-F_{\max}(x))dx + (1-\alpha)\int_0^1 (1-F(x))dx.
    \end{align*}
    To compute $\int_{0}^1 (1-F_{\max}(x))dx$, notice that
    \begin{align*}
        F_{\max}(x) = \Pr{q_j \le x, \forall j \in [n]} = \parans{x^\beta}^{ n/(n-1)},
    \end{align*}
    concluding
    \begin{align*}
        \int_{0}^{1} (1-F_{\max}(x))dx 
        &= \int_{0}^{1} \parans{1-\parans{x^\beta}^{n/(n-1)}} dx
        \\
        &=
        1 - \frac{n-1}{\beta n +  n - 1} 
        \\
        &=
         \frac{\beta n}{\beta n + n-1},
    \end{align*}
    For $\int_0^1 1-F(x) dx$, we have
    \begin{align*}
        \int_0^1 (1-F(x))dx 
        &= \int_0^1 1 - x^{\beta/(n-1)}dx
        \\
        &= 1 - \frac{1}{\beta/(n-1) + 1} [x^{\beta/(n-1) + 1}]\Big|^1_0
        \\
        &= \frac{\beta}{n-1 + \beta},
    \end{align*}
    and combining the two terms finishes the proof.
\end{proof}



On the other hand, if one considers an arbitrary policy beyond $\HM$, characterizing the exact equilibrium becomes significantly more challenging. The equilibrium condition turns into an equation of degree $\Theta(n)$ in both $F(q)$ and $q$, making it difficult to analyze precisely, as obtaining exact solutions for high-degree polynomial equations is generally impossible.

Nonetheless, in our analysis, we show that the value of the objective function can be characterized purely as a function of $\bp$, \emph{without explicitly characterizing} the equilibrium.  
This allows us to re-derive Proposition~\ref{thm:single-hm} without relying on the specific form of $F(\cdot)$, as demonstrated in Section~\ref{sec:opt}.

\section{Preliminaries for Section~\ref{sec:existence}}

\subsection{Preliminaries on Topology}\label{sec:topology}

\subsubsection{Basic of Topology}
\paragraph{Topological space}
A topology on $X$ is a collection (topology) $\cT$ of subsets of $X$ satisfying the following:
\begin{enumerate}
    \item $\emptyset \in \cT$ and $X \in \cT$.
    \item $\cT$ is closed under any (possibly infinite) union of members of $\cT$.
    \item $\cT$ is closed under any finite intersection of members of $\cT$.
\end{enumerate}

\paragraph{Open sets} The members of $\cT$, \ie each $U \in \cT$, are called \emph{open sets}.
{A representative example is the standard topology with $X = \R$, where $\cT$ is any union of open intervals in $\R$ with respect to one-dimensional Euclidean metric.}

The pair $(X,\cT)$ is called a \emph{topological space}, where we often write $X$ as a topological space if underlying topology $\cT$ is clear from the context.
A subset $U \subset X$ is closed in $\cT$ if $X \setminus U \in \cT$.

\paragraph{Neighborhood, Hausdorff space, and continuous function}
{
Given a topological space $X$ and $x \in X$, a subset $V \subseteq X$ is a neighborhood of $x$ if there exists an open set $U \in \cT$ such that $x \in U$ and $U \subset V$.
}
We say two points $x,y \in X$  can be separated by neighborhoods if there exists a neighborhood $U$ of $x$ and $V$ of $y$ such that $U \cap V = \emptyset$.
$X$ is \emph{Hausdorff space} if any two distinct points $x,y \in X$ are separated by neighborhoods.

A function $f: X \to Y $ between topological spaces is called \emph{continuous} if for every point $x \in X$, and for every neighborhood $V$ of $f(x)$ in $Y$, there exists a neighborhood  $U$ of $x$ in $X$ such that $f(U) \subseteq V$. 

In our problem setup, the action space $X$ is an Euclidean metric space, which is a well-known to be a {Hausdorff space~\citep{kelley2017general}.
Looking forward, we will use the existential results with a so-called Hausdorff game (to be defined) by~\cite{reny1999existence}, which requires the action space to be the Hausdorff space.}

\paragraph{Basis of topology}{
Given a topological space $(X,\cT)$, a collection $\cB$ of open sets from $\cT$ is called a basis of $\cT$ if every open set $U \in \cT$ can be expressed as a union of elements from $\cB$.
Equivalently, for any open set $U \in \cT$ and any point $x \in U$, there exists $B \in \cB$ such that $x \in B \subseteq U$.}



For a natural number $n$, let $X_i$ be a topological space for $i \in [n]$. Let $X := \prod_{i \in [n]} X_i$. The product topology on $X$ is the topology whose basis consists of all sets of the form $U_1 \times U_2 \times \ldots \times U_n$, where $U_i$ is an open set in $X_i$ for $i \in [n]$.


\subsubsection{Dual Space and Weak$^*$ Topology}
Our analysis relies on the \emph{metrizability} of the space of mixed strategies, viewed as a topological space equipped with the L\'evy-Prokhorov metric (to be defined shortly). Intuitively, it can be technically challenging to reason directly about open neighborhoods of mixed strategies (i.e., probability measures over the pure strategy space) using only the abstract topology. A more tractable approach is to metrize this topology using a suitable metric—such as the L\'evy-Prokhorov metric—which induces a natural notion of distance derived from the topology on the pure strategy space. This metric-based perspective enables a more concrete analysis of convergence and continuity in the space of mixed strategies.

To formalize this, we begin with the classical fact that if \( X \) is a compact Hausdorff space, then by the Riesz Representation Theorem, every positive bounded linear functional on \( L(X) \) can be uniquely represented by a finite Borel measure on \( X \), where 
 the space \( L(X) \) is defined as:
\begin{align}\label{eq:C}
L(X) = \{f: X \to \mathbb{R} \mid f \text{ is continuous and linear} \}.
\end{align}
Then, \( L(X) \) is a normed vector space\footnote{A normed vector space is a vector space on which a norm is defined.} under the supremum norm:
\[
\|f\| = \sup_{x \in X} |f(x)|.
\]
Since \( X \) is compact, this norm is finite, and \( L(X) \) is complete—making it a Banach space (see, e.g.,~\citep{parthasarathy2005probability}).



We define the dual space \( L(X)^* \) as the set of all real-valued continuous linear functionals on \( L(X) \), that is,
{
\[
    L(X)^* = \{ \phi: L(X) \to \R \big| \forall f,g \in L(X), \forall \alpha, \beta \in \R, \phi(\alpha f + \beta g) = \alpha \phi(f) + \beta \phi(G) \text{ and $\phi$ is continuous} \} 
\]
}
The dual norm on \( L(X)^* \) is given by
\[
\|\phi\| = \sup \{ |\phi(f)| : f \in L(X), \|f\| \le 1 \}, \qquad \phi \in  L(X)^*\,.
\]
{
In general, a dual space of of a normed vector space $X$ is the set of all continuous linear functionals from $X$ to $\R$ similarly defined as above.
}

\paragraph{Dual space and the space of probability measure}
By the Riesz representation theorem, any \( \phi \in L(X)^* \) corresponds to integration against a signed regular Borel measure on \( X \), where the Riesz representation theorem can be written as follows:
\begin{theorem}[Riesz representation theorem,~\cite{parthasarathy2005probability}]\label{thm:riesz}
    Let $X$ be a compact Hausdorff space and $L(X)$ be the real vector space of all real-valued continuous bounded linear functionals on $X$, as defined in Equation \eqref{eq:C}. 
    If $\phi \in L(X)^*$ is a positive linear functional such that $\phi(f) \ge 0$ for all $f \in L(X)$ with $f(x) \ge 0$ for every $x \in X$ and $||\phi|| = 1$ with the supremum norm, then there exists a unique Borel probability measure $\mu$ on $X$ such that
    \begin{align*}
    \phi(f) = \int_X f d\mu,\qquad 	\forall f \in L(X).
    \end{align*}
\end{theorem}
Let $M$ denote the set of (Borel) probability measures over a compact Hausdorff space $X$, \ie the space of mixed strategies.
Then, the Riesz representation theorem asserts the one-to-one correspondence\footnote{Note that the reverse direction is immediate as given a Borel probability measure $\mu$ on $X$, once we write $\phi(f) = \int_X f d\mu$ for $f \in L(X)$, this is a positive linear functional on $L(X)^*$.} between the space of positive linear functional $\phi$ on $L(X)$ with $\norm{\phi} = 1$ and $M$, \ie the space of mixed strategy.
Such equivalence will be helpful in endowing the space of mixed strategies with a proper topology:
\begin{definition}[Weak$^*$ topology]\label{def:weak-star-topo}
    {
        Let $X$ be a normed vector space over $\R$ and $X^*$ be its dual space, \ie the space of all linear continuous functionals $f: X \to \R$.
        For each $x \in X$, define a mapping $g_x: X^* \to \R$ as follows: $g_x(f) = f(x)$ for every $f \in X^*$.
        The weak$^*$ topology on $X^*$ is defined as the coarsest topology on $X^*$ (with the fewest number of open sets) such that all mappings $g_x$ for every $x \in X$ are continuous.
    }
    
\end{definition}
A seminal consequence of the weak$^*$ topology is that a sequence of functionals $(f_i)$ in $X^*$ converges to a functional $f \in X^*$ in the weak$^*$ topology if $f_i(x)$ converges to $f(x)$ for every $x \in X$ in a pointwise manner~\citep{kelley2017general}.

Throughout we consider the set of probability measures endowed with the weak$^*$ topology.

\paragraph{Metrizability}
A topological space $(X,\cT)$ is called \emph{metrizable} if its topology $\cT$ is induced by a metric.
Formally, if there exists a metric $d$ on the set such that the open sets defined by $d$\footnote{{Given a metric space $(X,d)$, a subset $U \subseteq X$ is called open if for all $x \in U$, there exists $\eps > 0$ such that $B_d(x,\eps) \subseteq U$ where $B_d(x,\eps) = \{y \in X: d(x,y) < \eps\}$.}} 
If $X$ is a metric space, it is known that $L(X)^*$ is metrizable with a proper metric, stated as follows. 
\begin{theorem}[\cite{narici2010topological}]\label{thm:equivalence-weak-prok}
If \( X \) is a metric space and \( M \) is a bounded subset of its dual space \( L(X)^* \) (the space of bounded linear continuous functionals), then \( M \), when endowed with the weak\(^*\) topology, is a metrizable topological space. That is, it is homeomorphic\footnote{{A function $h:X \to Y$ between two topological spaces (with some topologies endowed) if it is a bijection, continuous, and also its inverse is continuous. If such a function exists, we say two topological spaces are homeomorphic.}} to a metric space under the L\'evy-Prokhorov metric induced by the metric on \( X \).

\end{theorem}
Note that the L\'evy-Prokhorov metric is defined as follows.
\begin{definition}[L\'evy-Prokhorov metric]\label{def:levy}
    Let $(X,d)$ be a metric space {given a distance function (metric) $d$} with its Borel sigma algebra $\cB(X)$.  Let $\cP(X)$ denote the collection of all Borel probability measures on the measurable space $(M, \cB(X))$.
    For a subset $A \subseteq X$, define the $\eps$-neighborhood of $A$ by
    \begin{align*}
        A^\eps = \{p \in X : \exists q \in A, d(p,q) < \eps\} = \cup_{p \in A} B_\eps(p),
    \end{align*}
    where $B_\eps(p)$  is the open ball of radius $\eps$ at $p$, {\ie $B_{\eps}(p) = \{q \in X: d(p,q) < \eps\}$}.
    Then, the L\'evy-Prokhorov metric $\pi: \cP(M)^2 \to [0,\infty]$ is defined by setting the distance between two probability measures $\mu$ and $\nu$ to be

    \[
        \pi(\mu, \nu) := \inf \left\{ \epsilon > 0 :
        \begin{array}{l}
        \mu(A) \le \nu(A^\epsilon) + \epsilon \\
        \nu(A) \le \mu(A^\epsilon) + \epsilon
        \end{array}
        \quad \text{for all } A \in \mathcal{B}(X) \right\}.
    \]

Intuitively, \(\pi(\mu, \nu)\) is the smallest \(\epsilon\) such that for every measurable set \(A\), the measure \(\mu\) of \(A\) is close to the measure \(\nu\) of the \(\epsilon\)-expanded version of \(A\), and vice versa. This captures how much the two probability distributions differ in terms of how they spread mass over points that are close together in the metric \(d\). 
\end{definition}
Thus, Theorem~\ref{thm:equivalence-weak-prok} immediately asserts that L\'evy-Prokhorov metric metrizes the weak$^*$ topology on the space of bounded linear functionals, which in fact is the space of mixed strategies by Theorem~\ref{thm:riesz}, \ie any open ball with respect to the L\'evy-Prokhorov metric is an open set in the weak$^*$ topology.
{
The following proposition summarizes the preliminaries we require to prove the existential result.
}
\begin{proposition}\label{prop:summary-topo}
   The family of mixed strategies $X^*$ over an action space $X$, where $X$ is a metric space, can be endowed with the weak$^*$ topology, which can be metrized by the L\'evy–Prokhorov metric. 
    In particular, if $\phi : X \to \mathbb{R}$ is bounded and continuous, then the map $h : \mu \mapsto \int \phi \, d\mu$ is continuous with respect to the weak$^*$ topology on $X^*$.

\end{proposition}

\subsection{Preliminaries on Games and Existential Results}\label{sec:prelim}

Technically speaking, to argue the existence of a symmetric MNE, we need to analyze the expected payoff of mixed strategies that lie within an open neighborhood of a symmetric strategy in order to apply the argument by~\cite{reny1999existence}.
To this end, we rigorously define the notion of an open neighborhood around a mixed strategy (i.e., a probability measure on a compact set) using the weak$^*$ topology, as formalized in Definition~\ref{def:weak-star-topo}.

Since reasoning directly over the topological space can be challenging, we exploit the fact that the space of probability measures on a compact Hausdorff space is metrizable under the L'evy-Prokhorov metric. This enables us to treat the topological space as a metric space—namely, we analyze the topology of probability measures on $X$ (endowed with the weak$^*$ topology) using a metric induced from the original topological (metric) space $X$.

For readers unfamiliar with basic topological concepts, we refer to Section~\ref{sec:topology} for background material.

\subsubsection{Game and Mixed Extension}
A game $G = (X_i, u_i)_{i \in [n]}$ consists of $n$ players, a set of pure strategy $X_i$ and payoff functions $u_i: X \to \R$ where $X = \times_{i \in [n]} X_i$ for $i \in [n]$.
We assume throughout that $X_i$ is a compact Hausdorff space, and we say the game is \emph{Hausdorff} game in such case.
Throughout, the product of any number of sets is endowed with the product topology.
Let $M_i$ denote the set of (Borel) probability measures\footnote{A Borel measure on a topological space is a measure that is defined on all open sets.} on $X_i$ for $i \in [n]$, \ie the set of mixed strategies over $X_i$.

Given a game $G = (X_i, u_i)_{i \in [n]}$ and the set of mixed strategies $M_i$  defined above, we extend each $u_i$ to have its domain as $M = \times_{i \in [n]} M_i$ by defining
\begin{align*}
    u_i(\mu) = \int_X u_i(x) d\mu,
\end{align*}
for all $\mu \in M$.
We denote by $\bar{G} = (M_i, u_i)_{i \in [n]}$ the mixed extension of $G$.

\paragraph{Quasi-symmetric game}
Consider a \emph{game} $G = (X_i, u_i)_{i \in [n]}$ (which is possibly a mixed extension of another original game $G$').
We say that $G$ is a \emph{quasi-symmetric} game if $X = X_1 = \ldots X_n$ and for any $x, y \in X$, it follows that
\begin{align*}
    u_1(x,y,\ldots, y) = u_2(y,x,y,\ldots, y) = \ldots = u_n(y, \ldots, y, x).
\end{align*}
Define a quasi-symmetric game's \emph{diagonal} payoff function $v:X \to \R$ as \[v(x) = u_1(x,\ldots, x) = u_n(x, \ldots, x)\] for every $x \in X$.

\subsubsection{Better-reply Secureness}
In proving the existence of a symmetric MNE, the notion of better-reply secure is crucial.
To that end, we first introduce a notion of securing a payoff.

\begin{definition}[Secure a payoff]\label{def:secure}
    Given a quasi-symmetric game $G = (X, u_i)_{i \in [n]}$, player $i$ can \emph{secure a payoff of $\alpha \in \R$ along the diagonal} at $(x,\ldots, x)$ if there exists $\bar{x} \in X$ such that \[u_i(x', \ldots, \bar{x}, \ldots, x') \ge \alpha\] for every $x'$ in some open neighborhood of $x \in X$.\footnote{When dealing with $G$ as a mixed extension of another original game, the open neighborhood refers to the open neighborhood with respect to the underlying topology over the space of mixed strategies.}
\end{definition}

The following diagonally better-reply secureness guarantees that any player could secure a strictly larger payoff along the diagonal, given a symmetric strategy profile which is not an equilibrium.

\begin{definition}[Diagonally better-reply secure]\label{def:better-reply}
    Given a quasi-symmetric game $G = (X, u_i)_{i \in [n]}$, let $(x, v(x)) \in X \times \R$ be the pair of an action $x \in X$ and diagonal payoff function $v(x)$.
    Then, $G$ is \emph{diagonally better-reply secure} if for any $(x, v(x))$ with $x \in X$ such that $(x,\ldots, x)$ is not an equilibrium, there exists a player $i$ who can secure a payoff strictly above $v(x)$ along the diagonal at $(x,\ldots, x)$ per  Definition \ref{def:secure}.  
\end{definition}


The seminal result by~\cite{reny1999existence} prove that any quasi-symmetric game possesses a symmetric MNE under the diagonally better-reply secureness and semi-continuity of the payoff function.

\begin{theorem}[Corollary 5.3,~\cite{reny1999existence}]\label{thm:reny}
    Suppose $G = (X_i, u_i)_{i \in [n]}$ is a quasi-symmetric game with compact Hausdorff action space.
    Then, $G$ has a symmetric MNE if its mixed extension $\bar{G}$ is diagonally better-reply secure and each $u_i(\mu, \ldots, \mu)$ is upper semi-continuous over $\mu$. 
\end{theorem}


\section{Omitted Proofs in Section~\ref{sec:existence}}
\subsection{Proof of Theorem~\ref{thm:exist}}\label{apd:thm:exist}
We begin with the following result, establishing that the induced game is both quasi-symmetric and a Hausdorff game (as defined in Section~\ref{sec:prelim}).
\begin{lemma}\label{cl:sym}
    The induced game given a nontrivial policy  is both quasi-symmetric and a Hausdorff game (as defined in Section~\ref{sec:prelim}).
\end{lemma}
\begin{proof}

    The quasi-symmetry immediately follows from the fact that every producer is homogeneous in a sense that they have the same cost function and thereby the same payoff as a function over the quality. This is true even if 
    the mixed strategy is not  atomless as we assume uniform random tie-breaking rule. 

    To prove that the game is Hausdorff, it suffices to show that its action space is compact.  
Given any policy, for any symmetric MNE, each provider must receive nonnegative expected utility for every action \( q \) in the support of the MNE, since playing \( q = 0 \) yields utility zero.  

On the other hand, if the support contains \( q > 1 \), the expected utility becomes strictly negative due to the cost function as $c(q) = q^\beta > 1$.   
Therefore, it is without loss of generality to restrict the action space to the interval \([0,1]\), which is compact.

\end{proof}

Next, we prove that the utility function of each producer, viewed as a function over mixed strategies, is upper semi-continuous at any symmetric strategy profile.
As discussed in Section~\ref{sec:topology}, to handle continuity with respect to the mixed strategy $\mu$—that is, a probability measure—we consider the space of mixed strategies endowed with the weak$^*$ topology. This topology is metrizable via the L'evy-Prokhorov metric, as noted in Proposition~\ref{prop:summary-topo}.
We refer the reader to Section~\ref{sec:topology} for additional background.

\begin{lemma}\label{cl:cont}
    For any nontrivial policy $\bp$, at any symmetric mixed strategy profiles $\bbmu = (\mu,\ldots, \mu)$ 
\end{lemma}
\begin{proof}
    Due to the symmetric of the game, for any symmetric mixed strategy, it is straightforward to see that the overall expected revenue of each producer is simply $1/n$.
    Thus, the expected utility can be written as
    \begin{align*}
        \cU_i(\mu; \bmu_{-i}) = \frac{1}{n} - \int_0^1 c(x)d\mu.
    \end{align*}
    As we endow the space of set of mixed strategies $\mu$ with the weak$^*$ topology, 
     the cost function $c(x)$ is continuous, and thus $\int_0^1 c(x) d\mu$ is continuous in $\mu$ {by Proposition~\ref{prop:summary-topo}}. 
    This concludes that the utility function $\cU_i$ is upper semi-continuous with respect to $\mu$ {as continuity implies upper semi-continuity.}. 

    
\end{proof}

Now, for ease of exposition, we define  
\begin{align}
    R(q; \bbmu_{-i}) = \sum_{i=1}^n p_i \binom{n-1}{i-1}F(q)^{n-i} (1-F(q))^{i-1},\label{eq:r-function}
\end{align}  
which represents the expected revenue from choosing quality level \( q \) when all other producers follow the mixed strategy described by \( F \).  
That is, as noted in Equation~\eqref{eq:rev},  
\begin{align*}
    \cU_i(q; \bbmu_{-i}) = R(q;\bbmu_{-i}) - c(q),
\end{align*}
as presented in Lemma~\ref{lem:derivation-welfare-payoff}, under the assumption that \( F \) has no point masses.  
We often write \( R(q) \) when the strategies of the other producers are clear from context.

The following lemma will be useful in proving diagonally better-reply secureness.

\begin{lemma}\label{cl:r-monotone}
    For any nontrivial policy $\bp$, given the other producers' symmetric mixed strategies $\bbmu_{-i}$ with corresponding c.d.f.\ $F$, the expected revenue $R(q; \bbmu_{-i})$ is monotone increasing in $q$.

\end{lemma}
\begin{proof}
    Fix $q$, and consider a random variable $Y_q \sim \text{Binomial}(n-1, 1-F(q))$.
    Then 
    \begin{align*}
        \Pr{Y_q = i-1} = \binom{n-1}{i-1}F(q)^{n-i} (1-F(q))^{i-1}.
    \end{align*}
    Therefore, by Equation \eqref{eq:r-function}, we have
    \begin{align*}
        R(q; \bbmu_{-i})
        &= \sum_{i=1}^n p_i \Pr{Y_q = i-1} 
        \\
        &= \sum_{i=1}^n p_i \Pr{i = Y_q + 1} 
        \\
        &= \Ex{p_{Y_q+1}}.
    \end{align*}
Now if $q \ge q'$, then $F(q) \ge F(q')$, and thus $Y_{q}$ is stochastically dominated by $Y_{q'}$.
    Hence, we have $\Ex{f(Y_q)} \ge \Ex{f(Y_{q'})}$ for any non-increasing function $f$, \eg see~\cite{mas1995microeconomic}. 
    Since $p_{Y+1}$ is non-increasing in $Y$ as $p_1 \ge p_2 \ge \ldots \ge p_n$, we have $\Ex{p_{Y_q + 1}} \ge \Ex{p_{Y_{q'} + 1}}$, which finishes the proof. 
\end{proof}

To prove the existence of the symmetric MNE, the final step is to prove the diagonally better-reply secureness by~\cite{reny1999existence}; see Theorem \ref{thm:reny}. 
\begin{lemma}\label{cl:diag}
    For any policy $\bp$ with $p_1 \ge p_2 \ge \ldots \ge p_n$, the induced game is diagonally better-reply secure per Definition \ref{def:better-reply}.
\end{lemma}
\begin{proof} 
Let $\mu$ be a mixed strategy that does \emph{not} constitute a symmetric MNE, and let $u = \cU(\mu; \mu, \ldots, \mu)$ denote the corresponding diagonal payoff. Per Definition~\ref{def:better-reply}, to establish diagonal better-reply security, it suffices to find another strategy $\mu'$ such that $\cU(\mu'; \mu'', \ldots, \mu'') > u$ for any $\mu''$ in an open neighborhood $B_\delta(\mu)$, for some $\delta > 0$, with respect to the L\'evy-Prokhorov metric. Here, $B_\delta(\mu) = \{\nu : d(\mu, \nu) < \delta\}$ where $d(\cdot, \cdot)$ denotes the L\'evy-Prokhorov metric. 
Since $\mu$ is not an equilibrium, there exists a pure strategy $q^{\text{dev}} \in \R$ such that
    \[
        \cU(q^{\text{dev}}; \mu, \ldots, \mu) > \cU(\mu; \mu, \ldots, \mu) = u.
    \]
   To see why \( q^{\text{dev}} \) must exist, suppose for the sake of contradiction that no pure strategy strictly improves the payoff. Then, for any mixed strategy \( \nu \), we have  
\[
\cU(\nu; \mu, \ldots, \mu) = \int_{q \in \supp(\nu)} \cU(q; \mu, \ldots, \mu) \, d\nu(q) \le \int_{q \in \supp(\nu)} u \, d\nu(q) = u,
\]
which implies that no mixed strategy can strictly increase the payoff---a contradiction. Therefore, such a pure strategy \( q^{\text{dev}} \) must exist.

To construct a secure deviation, we perturb   $q^{\text{dev}}$ to a nearby pure strategy $q^{\text{sec}}$ such that $\mu$ does not have an atom at $q^{\text{sec}}$ and
    \[
        \cU(q^{\text{sec}}; \mu, \ldots, \mu) > u.
    \]
    This is  possible since any probability distribution has at most countably many atoms.
    Let $F$ be the c.d.f. corresponding to $\mu$. 
    Since $q^{sec}$ does not have an atom, by~\eqref{eq:r-function}
    \[
        R(q^{\text{sec}}; \mu, \ldots, \mu) = \sum_{i=1}^n p_i \binom{n-1}{i-1} F(q^{\text{sec}})^{n-i} (1 - F(q^{\text{sec}}))^{i-1},
    \]
    and thus
    \[
        \cU(q^{\text{sec}}; \mu, \ldots, \mu) = R(q^{\text{sec}}; \mu, \ldots, \mu) - c(q^{\text{sec}}) > u.
    \]
    Due to the continuity of $F$, we can find a sufficiently small $\varepsilon > 0$ such that  $\mu$ does not have an atom at $q^{\text{sec}} - \varepsilon$ {and satisfies}
    \begin{align}
        \sum_{i=1}^n p_i \binom{n-1}{i-1} F(q^{\text{sec}} - \varepsilon)^{n-i} (1 - F(q^{\text{sec}} - \varepsilon))^{i-1} - c(q^{\text{sec}}) > u. \label{eq:eps-slack}
    \end{align}
    This holds because as \( \varepsilon \to 0 \), the left-hand side converges to \( \cU(q^{\text{sec}}; \mu, \ldots, \mu) > u \). Moreover, since any distribution has at most countably many atoms, there must exist such a sufficiently small \( \varepsilon \) for which \( \mu \) has no atom in the interval \( [q^{\text{sec}} - \varepsilon, q^{\text{sec}}] \).
    Note that such $\eps$ can be found arbitrarily small since if we consider any small neighborhood $[q^{sec}-\eps', q^{sec}]$ of $q^{sec}$, there exists at most countably many atoms on it, and thus we can find sufficiently small $\eps$ with no atom on it.
    We will prove that $q^{sec}$ is the desired strategy that secures payoff larger than $u$ along the diagonal at $\bmu$.

       Let $A := \{q': q' < q^{\text{sec}}\}$ and $A^\varepsilon := \{q': q' < q^{\text{sec}} - \varepsilon\}$. Then $F(q^{\text{sec}}) = \mu(A)$ and $F(q^{\text{sec}} - \varepsilon) = \mu(A^\varepsilon)$. Rewriting~\eqref{eq:eps-slack}, we get
    \[
        \Gamma := \sum_{i=1}^n p_i \binom{n-1}{i-1} \mu(A^\varepsilon)^{n-i} (1 - \mu(A^\varepsilon))^{i-1} - c(q^{\text{sec}}) > u.
    \]

    Let $A' := \{q': \exists q \in A^\varepsilon \text{ s.t. } |q' - q| < \varepsilon\}$. By the definition of the L\'evy-Prokhorov metric (see Definition \ref{def:levy}), for any $\mu' \in B_\varepsilon(\mu)$,
    \[
        \mu(A^\varepsilon) \le \mu'(A') + \varepsilon.
    \]
    Since any $q' \in A'$ satisfies $q' < q^{\text{sec}}$, we also have
    \[
        \mu'(A) \ge \mu'(A') \ge \mu(A^\varepsilon) - \varepsilon.
    \]

    Therefore, for any $\mu' \in B_\varepsilon(\mu)$, the payoff satisfies:
    \begin{align*}
        \cU(q^{\text{sec}}; \mu', \ldots, \mu') 
        &\ge \sum_{i=1}^n p_i \binom{n-1}{i-1} \mu'(A)^{n-i}(1 - \mu'(A))^{i-1} - c(q^{\text{sec}}) \\
        &\ge \sum_{i=1}^n p_i \binom{n-1}{i-1} (\mu(A^\varepsilon) - \varepsilon)^{n-i}(1 - \mu(A^\varepsilon) + \varepsilon)^{i-1} - c(q^{\text{sec}}) \\
        &= \Gamma + \Delta,
    \end{align*}
    where the second inequality follows from the monotonicity in Lemma~\ref{cl:r-monotone} and
    \[
        \Delta := \sum_{i=1}^n p_i \binom{n-1}{i-1} \left[ (\mu(A^\varepsilon) - \varepsilon)^{n-i}(1 - \mu(A^\varepsilon) + \varepsilon)^{i-1} - \mu(A^\varepsilon)^{n-i}(1 - \mu(A^\varepsilon))^{i-1} \right].
    \]
    Since $\Delta \to 0$ as $\varepsilon \to 0$, we can choose $\varepsilon$ small enough so that $\Delta > -(\Gamma - u)$, implying
    \[
        \cU(q^{\text{sec}}; \mu', \ldots, \mu') > u.
    \]
    This concludes the proof.

\end{proof}

\begin{proof}[Proof of Theorem~\ref{thm:exist}]
    Finally, the proof immediately follows from Theorem~\ref{thm:reny} combined with lemmas~\ref{cl:sym},~\ref{cl:r-monotone},~\ref{cl:cont}, and~\ref{cl:diag}.    
\end{proof}


\subsection{Proof of Proposition~\ref{prop:no-pne}}\label{apd:no-pne}
\begin{proof}
    Given a policy $\bp$, suppose there exists a PNE $q_1 \ge q_2\ge \ldots \ge q_n$.
    If $q_i \neq 0$ but $p_i = 0$, producer $i$ can effectively deviates to $q_i = 0$ while decreasing the cost, which is a contradiction.
    Thus, if $q_i \neq 0$ then we always have $p_i \neq 0$.
    Let $i$ be the largest index where $q_i \neq 0$.
    Suppose there exists no tie among producers $j \in [i-1]$, \ie $q_1>q_2>\ldots > q_i$.
    Any producer $j \in [i]$ can deviate to $q_j(1-\eps)$ for sufficiently small $\eps > 0$ while respecting $q_j(1-\eps) > q_{j+1}$, which only decreases the cost.
    Thus, this cannot be an equilibrium.
    Suppose now that there exists a tie at $j$ such that $q_j = q_{j+1}$.
    Suppose there exists $k \ge 2$ such agents where $j$ be the smallest index, \ie $q_j = q_{j+1} = \ldots = q_{j+k-1}$.
    Due to the uniformly random tie-breaking rule, all the producers from $j$ to $j+k-1$ obtain revenue of $(p_j + p_{j+1} + \ldots p_{j+k-1})/k$.
    If $p_j = p_{j+k-1}$, then producer $j$ can slightly decrease the quality while still receiving the revenue of $p_{j+k-1} = (p_j + p_{j+1} + \ldots p_{j+k-1})/k$ and decreasing the cost.
    Thus, $p_j \neq p_{j+k-1}$.
    In this case, producer $j$ can instead increase the quality from $q_j$ to $q_j + \eps$ and solely obtain $p_j > (p_j + p_{j+1} + \ldots p_{j+k-1})/k$ while only increasing cost function sufficiently small so that the overall utility increases due to the continuity of the payoff function.
    This finishes the proof.
\end{proof}

\subsection{Proof of Proposition~\ref{cl:atomless}}\label{apd:atomless}
\begin{proof}


   To prove that the strategy is atomless, fix an agent \( i \).  
Assume, for the sake of contradiction, that there exists a point \( p \) such that agent \( i \) plays \( p \) with positive probability \( \alpha > 0 \).  
Let \( p' = p + \varepsilon \) for some \( \varepsilon > 0 \).  
Let \( E \) be the event that all other producers \( j \in [n] \setminus \{i\} \) also play \( p \), \ie, \( \Pr{E} = \alpha^{n-1} \).

Conditioned on \( E \), playing \( p \) yields expected revenue \( \nicefrac{1}{n} \).  
On the other hand, playing \( p' \) yields expected revenue at least \( p_1 > \nicefrac{1}{n} \) for any \( \varepsilon > 0 \), since \( \bp \) is nontrivial.  
Thus, by the continuity of the cost function, agent \( i \) can deviate to \( p + \varepsilon \) for sufficiently small \( \varepsilon \) and strictly increase their utility under the event \( E \).

Since this deviation does not decrease the revenue under \( E^c \), one can find sufficiently small \( \varepsilon \) to strictly improve agent \( i \)'s utility overall—a contradiction.

To prove that the distribution function \( F \) is strictly increasing, suppose that \( (a,b) \subseteq [q_{\min}, q_{\max}] \) is the subinterval of maximal length such that \( F \) remains constant on \( (a,b) \), where
\[
q_{\max} = \sup\{y : F(y) < 1\}.
\]
Note that \( q_{\min} \geq 0 \) and \( q_{\max} \leq 1 \) because playing quality \( 1 \) yields nonpositive utility.

Clearly, any agent placing mass on \( b \) can move that mass to \( a \), thereby strictly increasing their utility—a contradiction.
Hence, \( F \) must be strictly increasing on its support.
Similarly, if $a > 0$, any agent can move its mass from $a$ to $0$ to strictly increase their utility, which is a contradiction.
So, $F$ is strictly increasing on $[0,q_{\max}]$.
\end{proof}

\subsection{Proof of Lemma~\ref{lem:derivation-welfare-payoff}}
\begin{proof}
    From Proposition~\ref{cl:atomless}, we know that $\mu$ is atomless, \ie tie-breaking happens with probability $0$.
    Thus, we can write the producer $i$'s payoff as follows:
    \begin{align}\notag
        \cU_i(q, \bbmu_{-i},\bp)
        &= \cR_i(q, \bbmu_{-i}, \bp)  - c(q)
        \\
        &= \Ex{\parans{\sum_{i \in [n]}\Ind{I = i}} - c(q)}\notag
        \\
        &= \parans{\sum_{i \in [n]}\Pr{\text{Provider $i$'s content is recommended to user $i$}}} - c(q) \notag
        \\
        &= \parans{\sum_{i=1}^n p_{i} \binom{n-1}{i-1}F(q)^{n-i}(1-F(q))^{i-1}} - c(q). \label{eq:rev}
    \end{align}
    For the user welfare, let $E_i$ be the random event such that the chosen content has the $i$-th highest quality among $[n]$ and $I$ be the  variable that indicates the index of the content chosen by the RS.
    Expanding using the definition of $\bp$, we obtain: 
    \begin{align*}
        \cW(\bbmu, \bp)  
        &= \Exu{\bq \sim \bmu}{\sum_{i=1}^n\Ind{I= i} \cdot q_i}
        \\
        &= \Exu{\bq \sim \bmu}{\sum_{i=1}^n\Ind{E_i} \cdot q_{(i)}}
        \\
        &= \sum_{i=1}^n p_{i} \Exu{\bq \sim \bmu}{q_{(i)}}.
    \end{align*}
    Finally, the equation for the platform quality immediately follows from the fact that $\bmu$ is symmetric strategy profiles.
    
\end{proof}

\ssedit{
\subsection{Proof of Theorem~\ref{thm:uniq-symm}}
\begin{proof}
Recall first that
\begin{align*}
    a_i(x) &= \binom{n-1}{i-1}x^{\,n-i}(1-x)^{i-1}
    \\
    h(x,p) &= \sum_{i=1}^n p_i\,a_i(x)\quad(x\in[0,1]).    
\end{align*}
By our existence result for symmetric MNE, a symmetric equilibrium $\mu$ exists for every nontrivial policy $p$.
By~\eqref{eq:06121508}, we have that
\begin{equation}\label{eq:10221350}
u_0 + q^\beta\ =\ \sum_{i=1}^n p_i \binom{n-1}{i-1}F(q)^{\,n-i}\bigl(1-F(q)\bigr)^{i-1}
\ =\ h\big(F(q),p\big)\quad,
\end{equation}
for all $q\in[q_{\min},q_{\max}]$ where $q_{\min}$ and $q_{\max}$ denote minimal and maximal value in the support of $F$.
Note that Proposition~\ref{cl:atomless} yields $q_{\min} = 0$.
Thus, plugging $q=0$ into \eqref{eq:10221350} and using $F(0)=0$ gives
$u_0=p_n$.

Recall that the c.d.f. $F$ is atomless and strictly increasing on some interval
$[0,q_{\max}]$ by Proposition~\ref{cl:atomless}.
Hence $h(\cdot,p)$ has a unique continuous inverse on $[h(0,p),h(1,p)]=[p_n,p_1]$. 

Thus, we have
\begin{align}
    F(q)\ =\ h^{-1}\big(p_n+q^\beta\,;p\big),\label{eq:10221356}
\end{align}
for all $q\in[0,q_{\max}]$.
Notice that this determines the c.d.f. $F$ uniquely once its support is given. 
Since $F(0)=0$ we must have $h(0,p)=p_n$,
and since $F(q_{\max})=1$ we must have $p_n+q_{\max}^\beta=h(1,p)=p_1$, i.e.,
$q_{\max}=(p_1-p_n)^{1/\beta}$.

Let us confirm again that the support is always $[0,q_{\max}]$ for $q_{\max} = (p_1 - p_n)^{1/\beta}$.
For any $q>q_{\max}$, all other players’ draws lie in
$[0,q_{\max}]$, so a deviation to $q$ surely earns the top rank and revenue $p_1$ due to atomlessness, while incurring cost $q^\beta$.
Since $q^\beta\ge q_{\max}^\beta=p_1-p_n$, the deviator’s payoff $p_1-q^\beta\le p_n=u_0$, hence there is no
profitable deviation above the support.
If the minimum value in the support of $F$ is $q_{\min} > 0$, then any play at $q_{\min}$ can deviate to $0$ by decreasing its cost and still ensure revenue $p_n$.
Thus, any symmetric equilibrium has the distribution $F$ on the support $[0, q_{\max}]$ for $q_{\max} = (p_1-p_n)^{1/\beta}$.

Finally, to conclude, suppose $\tilde{F}$ be a c.d.f. of another symmetric MNE.
Notice that it needs to satisfy~\eqref{eq:10221356} on the support $[0,q_{\max}]$.
Since $h(\cdot, p)$ is strictly increasing by Lemma~\ref{lm:h-monotone},~\eqref{eq:10221356} implies that $\tilde{F}(q)=h^{-1}(p_n+q^\beta;p)=F(q)$ for all $q$ in the support.
Thus the symmetric MNE is unique.
\end{proof}

}

\section{Omitted Proofs in Section~\ref{sec:opt}}
\subsection{Proof of Theorem~\ref{thm:prop-user}}\label{apd:proof:prop-user}
Before proving the results, we introduce further notations:
\begin{align}
    \psi(x, \bp)
    &= \sum_{i=2}^{n-1} p_i (n-1) \parans{\binom{n-2}{i-1}x^{n-i-1}(1-x)^{i-1} - \binom{n-2}{i-2}x^{n-i}(1-x)^{i-2}} \nonumber
    \\&\quad\quad+ p_1(n-1)x^{n-2} - p_n(n-1)(1-x)^{n-2},\label{eq:psi}
    \\
    \phi(x, \bp) 
    &=  \sum_{i=1}^{n}\parans{1 - \sum_{j=1}^i \binom{n}{j-1}(1-x)^{j-1}(x)^{n-j+1})}p_i \label{eq:phi}
    \\
    &= 1 - \sum_{i=1}^m \binom{n}{i-1}(1-x)^{i-1}x^{n-i+1}(\sum_{j=i}^n p_j), \label{eq:phi2}
\end{align}


Furthermore, we require several technical lemmas.  
We begin by establishing connections between the notations defined above and those introduced in Section~\ref{sec:opt}.

The following is a well-known characterization of the c.d.f. of order statistics for identically and independently distributed (i.i.d.) random variables.

\begin{lemma}[\cite{casella2024statistical}]\label{lem:cdf-orderstat}
    Given $n$ i.i.d. random variables $X_1,..,X_n$, let $F$ be its c.d.f. and $F_{(i)}$ be the c.d.f. of the $i$-th largest order statistic.
    Then, for any $i \in [n]$, 
    \begin{align*}
        F_{(i)}(x) = \sum_{j=1}^i \binom{n}{j-1}(1-F(x))^{j-1}F(x)^{n-j+1}.
    \end{align*}
\end{lemma}
We omit the proof, which is routine and covered in standard texts.

\begin{lemma}\label{lm:psi-to-h}
The function $\psi(x, \bp)$, defined in Equation~\eqref{eq:psi}, represents the derivative of $h(x, \bp)$ with respect to $x$, where $h(x, \bp) = \sum_{i=1}^n p_i \binom{n-1}{i-1}x^{n-i}(1-x)^{i-1}$ is given in Equation~\eqref{eq:h}.

\end{lemma}
\begin{proof}
    Note that
    \begin{align}\label{eq:05031722}
        \frac{dh(x,\bp)}{dx} = \frac{d \sum_{i=1}^n p_i a_i(x)}{dx}.
    \end{align}
    Recall that $ a_i(x) =  \binom{n-1}{i-1}x^{n-i}(1-x)^{i-1}$.  If $i=2,\ldots, n-1$, we have
    \begin{align*}
        \frac{d  a_i(x)}{dx} 
        &= \binom{n-1}{i-1}(n-i) x^{n-i-1}(1-x)^{i-1} - \binom{n-1}{i-1}(i-1) x^{n-i}(1-x)^{i-2}
        \\
        &= \frac{(n-1)!}{(i-1)!(n-i-1)!}x^{n-i-1}(1-x)^{i-1} - \frac{(n-1)!}{(i-2)!(n-i)!}x^{n-i}(1-x)^{i-2}
        \\
        &= (n-1) \parans{\binom{n-2}{i-1}x^{n-i-1}(1-x)^{i-1} - \binom{n-2}{i-2}x^{n-i}(1-x)^{i-2}}.
    \end{align*}
    For $a_1(x)$ and $a_n(x)$, we have
    \begin{align*}
        \frac{d  a_1(x)}{dx} = (n-1)x^{n-2}\,,
    \qquad 
        \frac{d  a_n(x)}{dx} = -(n-1)(1-x)^{n-2}.
    \end{align*}
    Plugging into~\eqref{eq:05031722} finishes the proof.
\end{proof}

The following is a well-known property of Bernstein polynomials, whose proof is provided for completeness.
\begin{lemma}[\cite{lorentz2012bernstein}]\label{lorentz2012bernstein}
    $\int_0^1 a_i(x) dx = 1/n$ for any $i \in [n]$, where  $a_i(x)$ is defined in Equation \eqref{eq:h}.
\end{lemma}
\begin{proof}
Substitute the definition of $a_i(x)$:
\begin{align*}
\int_0^1 a_i(x) dx &= \int_0^1 \binom{n-1}{i-1} x^{n-i}(1-x)^{i-1} dx \\
&= \binom{n-1}{i-1} \int_0^1 x^{n-i}(1-x)^{i-1} dx
\end{align*}
The integral part is a Beta function, which is defined as:
\begin{align*}
    B(p, q) = \int_0^1 t^{p-1}(1-t)^{q-1} dt = \frac{\Gamma(p)\Gamma(q)}{\Gamma(p+q)},
\end{align*}
where $\Gamma(k) = \int_0 t^{k-1} e^{-t}dt$ which is $(k-1)!$ for positive integer $k$.
Thus for positive integers $p$ and $q$, this can be written as $B(p, q) = \frac{(p-1)!(q-1)!}{(p+q-1)!}$.

Comparing the integral $\int_0^1 x^{n-i}(1-x)^{i-1} dx$ with the Beta function definition, notice that the integral part can be written as $B(n-i+1, i)$, which can be written as:
$$B(n-i+1, i) = \frac{\Gamma(n-i+1)\Gamma(i)}{\Gamma((n-i+1)+i)} = \frac{(n-i)!(i-1)!}{(n)!}$$

Now, we also expand the binomial coefficient $\binom{n-1}{i-1}$:
$$\binom{n-1}{i-1} = \frac{(n-1)!}{(i-1)!(n-1-(i-1))!} = \frac{(n-1)!}{(i-1)!(n-i)!}$$

Substitute these expressions back into the integral:
\begin{align*}
\int_0^1 a_i(x) dx &= \frac{(n-1)!}{(i-1)!(n-i)!} \cdot \frac{(n-i)!(i-1)!}{n!} \\
&= \frac{(n-1)!}{n!} \\
&= \frac{(n-1)!}{n \cdot (n-1)!} \\
&= \frac{1}{n},
\end{align*}
which finishes the proof.
\end{proof}

Finally, we require the following relation between $\phi(x,\bp)$ and $h(x,\bp)$.
\begin{lemma}\label{lem:diff-phi}
It follows that
\[
\frac{d\phi(x,\bp)}{dx} = -n h(x,\bp),
\]
where $\phi(x,\bp)$ is defined in Equation~\eqref{eq:phi2} and $h(x,\bp)$ is defined in Equation~\eqref{eq:h}.
\end{lemma}
\begin{proof}
    We have
    \begin{align*}
        \frac{d\phi(x,\bp)}{dx} 
        = -nx^{n-1} 
        &+ \underbrace{\sum_{i=2}^n\binom{n}{i-1}(i-1)(1-x)^{i-2}x^{n-i+1}\sum_{j=i}^n p_j}_{(A)} 
        \\
        &- \underbrace{\sum_{i=2}^n \binom{n}{i-1}(n-i+1) (1-x)^{i-1}x^{n-i}\sum_{j=i}^n p_j}_{(B)}. 
    \end{align*}
    Note that
    \begin{align*}
        (A)
        &= \sum_{i=2}^n \frac{n!}{(i-1)!(n-i+1)!}(i-1)(1-x)^{i-2}x^{n-i+1}\sum_{j=i}^n p_j
        \\
        &= \sum_{i=2}^n \frac{n!}{(i-2)!(n-i+1)!}(1-x)^{i-2}x^{n-i+1}\sum_{j=i}^n p_j
        \\
        &= n\sum_{i=2}^n \frac{(n-1)!}{(i-2)!(n-i+1)!}(1-x)^{i-2}x^{n-i+1}\sum_{j=i}^n p_j
        \\
        &= n\sum_{i=1}^{n-1} \frac{(n-1)!}{(i-1)!(n-i)!}(1-x)^{i-1}x^{n-i}\sum_{j=i+1}^n p_j
    \end{align*}
    On the other hand,
    \begin{align*}
        (B)
        &= \sum_{i=2}^n \frac{n!}{(i-1)!(n-i+1)!}(n-i+1) (1-x)^{i-1}x^{n-i}\sum_{j=i}^n p_j
        \\
        &= \sum_{i=2}^n \frac{n!}{(i-1)!(n-i)!} (1-x)^{i-1}x^{n-i}\sum_{j=i}^n p_j
        \\
        &= n\sum_{i=2}^n \frac{(n-1)!}{(i-1)!(n-i)!} (1-x)^{i-1}x^{n-i}\sum_{j=i}^n p_j.
    \end{align*}
    Hence, we have 
    \begin{align*}
        (A) - (B) 
        &= nx^{n-1}\sum_{j=2}^np_j - n(1-x)^{n-1}p_n + n\sum_{i=2}^{n-1} \frac{(n-1)!}{(i-1)!(n-i)!}(1-x)^{i-1}x^{n-i} \parans{\sum_{j=i+1}^n p_j - \sum_{j=i}^n p_j}
        \\
        &=
        nx^{n-1}\sum_{j=2}^np_j - n(1-x)^{n-1}p_n - n\sum_{i=2}^{n-1} \frac{(n-1)!}{(i-1)!(n-i)!}(1-x)^{i-1}x^{n-i}p_i
        \\
        &=
        nx^{n-1}\sum_{j=1}^n p_j - nx^{n-1} p_1 - n(1-x)^{n-1}p_n - n\sum_{i=2}^{n-1} \frac{(n-1)!}{(i-1)!(n-i)!}(1-x)^{i-1}x^{n-i}p_i
        \\
        &=
        nx^{n-1}\sum_{i=1}^n p_i - n\sum_{i=1}^{n} \frac{(n-1)!}{(i-1)!(n-i)!}(1-x)^{i-1}x^{n-i}p_i
    \end{align*}
    Thus, we have
    \begin{align*}
        (A) - (B) = nx^{n-1} - n\sum_{i=1}^n a_i(x)p_i,
    \end{align*}
    concluding that
    \begin{align}
        \frac{d\phi(x,\bp)}{dx} = -n\sum_{i=1}^n a_i(x)p_i = -nh(x,\bp)\,.\label{eq:01150641}
    \end{align}
\end{proof}

We finally need the following correlation type inequality.
\begin{lemma}[\cite{alon2016probabilistic}]\label{lm:FKG}
    Let $\nu_1,\ldots, \nu_n$ be a probability measure on $\R$.
    If functions $f,g:\R^n \mapsto [0,1]$ are monotone increasing functions, then $\Exu{\nu}{f\cdot g} \ge \Exu{\nu}{f} \cdot \Exu{\nu}{g}$, where $\nu$ is a product measure over $\nu_1,\ldots, \nu_n$.
    The same inequality holds when both $f,g$ are monotone decreasing functions.
\end{lemma}

Now we can proceed to the proof of the main theorem.
\begin{proof}[Proof of Theorem~\ref{thm:prop-user}]
    Let $\bbmu$ be the symmetric MNE and $F$ be corresponding c.d.f. of the induced game given a nontrivial policy $\bp$.
    For $\bbmu$ to be the equilibrium, the following condition must hold:
    \begin{align*}
        u_0 + c(q) = \sum_{i=1}^n p_i\binom{n-1}{i-1}F(q)^{n-i}(1-F(q))^{i-1},
    \end{align*}
    for some $u_0 \ge 0$ and for any $q \in \supp(F)$, {as there exists no atom in $\bmu$ by Proposition~\ref{cl:atomless}}. 
    Note that $u_0$ is the expected payoff at the equilibrium.

    Let {$q_{(i)}$} 
    denote the $i$-th largest quality among $(q_i)_{i \in [n]}$ where $q_i$ is random quality realized for each agent $i \in [n]$ with respect to the mixed strategy $\mu$.
    Let $F_{(i)}$ be the c.d.f. of $q_{(i)}$.
    Using Lemma~\ref{lem:derivation-welfare-payoff}, the RS's problem to maximize the objective function can be written as the following optimization problem:
    \begin{align*}
    \begin{aligned}
        (\Opt_1)\quad & \underset{\bp \in \Delta_m}{\text{maximize}}
        & & \alpha\sum_{i =1}^n p_i \Ex{q_{(i)}} + (1-\alpha)\Ex{q_i} \\
        & \text{subject to}
        & & u_0+c(q) = \sum_{i=1}^n p_i\binom{n-1}{i-1}F(q)^{n-i}(1-F(q))^{i-1} \quad \text{ $\forall q \in \supp(F)$}
        \\
        & & & p_1 \ge p_2 \ge \ldots \ge p_n, \; \exists i,j: p_i \neq p_j
        \\
        & & & \sum_{i=1}^n p_i = 1
        \\
        & & & \text{$F$ is c.d.f. of symmetric MNE given $\bp$}
    \end{aligned}
    \end{align*}

    Recall that $0 \in \supp(F)$ by Proposition~\ref{cl:atomless}.
    Thus, $F(0) = 0$, implying $u_0 = p_n$ by plugging in $q = 0$.
    Plugging $c(q) = q^\beta$  back in the constraint, we obtain
    \begin{align}
        p_n + q^\beta = \sum_{i=1}^n p_i\binom{n-1}{i-1}F(q)^{n-i}(1-F(q))^{i-1},\label{eq:06121508}
    \end{align}
    for every $q \in \supp(F)$.
    Thus, we have 
    \begin{align}
        q = \parans{\sum_{i=1}^n p_i\binom{n-1}{i-1}F(q)^{n-i}(1-F(q))^{i-1} - p_n }^{1/\beta}.\label{eq:03092207}
    \end{align}
    Define
    \begin{align}
        g(F(q), \bp) = \sum_{i=1}^n p_i\binom{n-1}{i-1}F(q)^{n-i}(1-F(q))^{i-1} - p_n.\label{eq:g}
    \end{align}
    Notice that $g(x,\bp) = h(x,\bp) - p_n$.
    Then we have $q = g(F(q),\bp)^{1/\beta}$.
    From Lemma~\ref{lem:derivation-welfare-payoff}, we have
    \begin{align*}
        \cW(\bbmu, \bp) 
        &= \sum_{i=1}^n p_i \Ex{q_{(i)}}
        = \sum_{i=1}^n \int_{q \in \supp(F)} (1-F_{(i)}(x)) dq.
    \end{align*}
    Let $f(q)$ be the derivative\footnote{Remark that since $F$ is not guaranteed to be absolute continuous, $f$ does not necessarily coincide with the density function. However, our analysis does not require $f$ to be the density function.} of $F(q)$ with respect to $q$.\footnote{Since $F$ is strictly increasing and atomless, it is differentiable almost everywhere by Lebesgue's differentiation theorem, implying that the derivative of $F$ exists almost everywhere.  
Thus, since sets of measure zero have no effect on the value of the integral—as the objective function is always finite—it is without loss of generality to assume the differentiability of $F$.  
Alternatively, one can restrict the integrals and the corresponding constraints to the set of differentiable points, which would yield the same objective value as the original formulation.
}

    Thus using Lemma~\ref{lem:cdf-orderstat}, we can rewrite the user welfare function as:\footnote{Note that $f(q) \neq 0 $ with probability $1$ by Proposition~\ref{cl:atomless}.}
    \begin{align*}
        \cW(\bbmu, \bp) 
        &= \int_{q \in \supp(F)} \sum_{i=1}^{n}\parans{1 - \sum_{j=1}^i \binom{n}{j-1}(1-F(q))^{j-1}F(q)^{n-j+1}}p_i dq
        \\
        &=
        \int_{q \in \supp(F)} \sum_{i=1}^{n}\parans{1 - \sum_{j=1}^i \binom{n}{j-1}(1-F(q))^{j-1}F(q)^{n-j+1}}p_i \frac{dq}{d F(q)} dF(q)
        \\
        &=  \int_{q \in \supp(F)} \phi(F(q), \bp) \frac{1}{f(q)}  dF(q) \tag{Equation~\eqref{eq:phi}}
    \end{align*}
    From the constraints on the condition for being equilibrium in Equation~\eqref{eq:06121508}, we have $ p_n + q^\beta = \sum_{i=1}^n p_i\binom{n-1}{i-1}F(q)^{n-i}(1-F(q))^{i-1}$. By differentiating over $q$, we have
    \begin{align}
        \beta q^{\beta- 1} = f(q)\psi(F(q), \bp),\label{eq:02060134}
    \end{align}
    since $\psi(x, \bp)$ is a derivative of $\parans{\sum_{i=1}^n p_i\binom{n-1}{i-1}x^{n-i}(1-x)^{i-1}}$ with respect to $x$ due to Lemma~\ref{lm:psi-to-h}.

    Finally, plugging $f(q)$ yields
    \begin{align*}
        \cW(\bbmu, \bp) 
        &= \int_{q \in \supp(F)} \phi(F(q),\bp)\psi(F(q), \bp) \frac{1}{\beta q^{\beta-1}} dF(q)
        \\
        &=
        \frac{1}{\beta}\int_{q \in \supp(F)} \phi(F(q),\bp)\psi(F(q), \bp) q^{-\beta + 1} dF(q)
        \\
        &=
        \frac{1}{\beta}\int_{q \in \supp(F)} \phi(F(q),\bp)\psi(F(q), \bp) g(F(q),\bp)^{- 1+1/\beta} dF(q)
        \\
        &=
        \frac{1}{\beta}\int_{0}^1 \phi(x,\bp)\psi(x, \bp) g(x,\bp)^{1/\beta -1} dx
    \end{align*}
    Thus, this is entirely a function of $\bp$ whose integrand does not explicitly depend on the cumulative distribution function $F(q)$.
    
    Note further that
    \begin{align}
        \frac{dg(x,\bp)^{1/\beta}}{dx} 
        &= \frac{dg(x,\bp)^{1/\beta}}{dg(x,\bp)} \frac{dg(x,\bp)}{dx} \nonumber
        \\
        &= \frac{1}{\beta}g(x,\bp)^{1/\beta - 1} \frac{dh(x,\bp)}{dx} \nonumber\\
        &= \frac{1}{\beta}g(x,\bp)^{1/\beta - 1} \psi(x,\bp)\label{eq:der-h}
    \end{align}

    Hence,  we have 
    \begin{align*}
    \cW(\bbmu, \bp) 
    &= \int_0^1 \phi(x, \bp) \cdot \frac{d}{dx} g(x, \bp)^{1/\beta} \, dx \\
    &= \left[ \phi(x, \bp) \cdot g(x, \bp)^{1/\beta} \right]_0^1 - \int_0^1 g(x, \bp)^{1/\beta} \cdot \frac{d}{dx} \phi(x, \bp) \, dx.
    \end{align*} 
    Since we have
    \begin{align*}
        g(0,\bp) &= p_n - p_n = 0
        \\
        g(1,\bp) &= p_1 - p_n
        \\
        \phi(1,\bp) &= 0
        \\
        \phi(0,\bp) &= 1
    \end{align*}
    we obtain
    \begin{align*}
        \beta \cdot \cW(\bbmu, \bp) 
        &= \parans{\beta \phi(x,\bp)g(x,\bp)^{1/\beta}\Bigg|_0^1 - \beta \int_0^1\frac{d\phi(x,\bp)}{dx}g(x,\bp)^{1/\beta}dx}
        \\
        &=
        \beta\int_0^1 \frac{-d\phi(x,\bp)}{dx} g(x,\bp)^{1/\beta} dx
        \\
        &=
        \beta\int_0^1 \frac{-d\phi(x,\bp)}{dx} (h(x,\bp)-p_n)^{1/\beta} dx.\label{eq:w-middle}
    \end{align*}
    Finally by Lemma~\ref{lem:diff-phi}, we obtain
    \begin{align}
        \cW(\bbmu, \bp) = n\int_0^1 h(x,\bp)g(x,\bp)^{1/\beta}dx.
    \end{align}
    Now, consider the perturbation on $\bp$ that moves small $0 < \delta < p_n$ from $p_n$ to $p_1$.
    Writing the resulting vector $\bp^{(\delta)}$, for any $x \in [0,1]$ we have
    \begin{align*}
        \frac{d}{d\delta}h(x,\bp^{(\delta)}) = a_1(x) - a_n(x),
    \end{align*}
    and
    \begin{align*}
        \frac{d}{d\delta}(h(x,\bp^{(\delta)})-p_n^{(\delta)})= 1 + a_1(x) - a_n(x).
    \end{align*}
    Thus, we have
    \begin{align*}
        \frac{d}{d \delta}\parans{h(x,\bp)(h(x,\bp) -p_n)^{1/\beta}}
        &= (a_1(x) - a_n(x))(h(x,\bp) - p_n)^{1/\beta} +  \frac{h(x,\bp)}{\beta}(h(x,\bp) - p_n)^{1/\beta - 1}(1 + a_1(x) - a_n(x))
        \\
        &\ge (a_1(x) - a_n(x))(h(x,\bp) - p_n)^{1/\beta} + \frac{(h(x,\bp) -p_n)^{1/\beta}}{\beta} (1+a_1(x) - a_n(x)).
    \end{align*}
    The second term is apparently nonnegative as $a_n(x) \in [0,1]$, and is even strictly positive if $x \in (0,1)$.
    
    Thus, let us focus on the first term.
    In differentiating~\ref{eq:w-middle} with respect to $\delta$, the first term will appear as an integrand as follows
    \begin{align*}
        \int_0^1 (a_1(x) - a_n(x))(h(x,\bp) - p_n)^{1/\beta} dx.
    \end{align*}
    Notice that $a_1(x) - a_n(x)$ is increasing on $x$.
    Since $h(x,\bp)$ is also increasing on $x$ by Lemma~\ref{lm:h-monotone}, by FKG inequality (Lemma~\ref{lm:FKG}), we have
    \begin{align*}
        \int_0^1 (a_1(x) - a_n(x))(h(x,\bp) - p_n)^{1/\beta} dx 
        &\ge 
        \int_0^1 (a_1(x) - a_n(x))dx \int_0^1(h(x,\bp) - p_n)^{1/\beta} dx \\
        &\ge 0.
    \end{align*}
    Combining, $\cW(\bbmu, \bp)$ only increases by such operation, implying that optimal policy should have $p_n = 0$.
    Note further that this implies
    \begin{align}
        q^\beta = \sum_{i=1}^n p_i\binom{n-1}{i-1}F(q)^{n-i}(1-F(q))^{i-1} = h(F(q), \bp),\label{eq:q-to-h}
    \end{align}

    Similarly, for the quality, we have
    \begin{align*}
        \cQ(\bbmu, \bp) 
        &= \Exu{q \sim \mu_i}{q}
        \\
        &=  \int_{q \in \supp(F)} qdF(q)
        \\
        &=  \int_{q \in \supp(F)} \parans{\sum_{i=1}^n p_i\binom{n-1}{i-1}F(q)^{n-i}(1-F(q))^{i-1} - p_n}^{1/\beta} dF(q)
        \\
        &=  \int_0^1 g(x,\bp)^{1/\beta} dx.
    \end{align*}
    Recalling the definition of $g(x,\bp) =  \sum_{i=1}^n p_i\binom{n-1}{i-1}F(q)^{n-i}(1-F(q))^{i-1} - p_n$ from Equation~\eqref{eq:g}, observe that the coefficient of $p_n$ is negative as $p_n$'s coefficient is $(1-F(q))^{n-1} - 1$, whereas all the other indices have positive coefficient.

    Combining, both the user welfare and average quality increases by decreasing $p_n$, so an optimal policy would have $p_n = 0$.\footnote{There could be boundary regimes where it remains the same by decreasing $p_n$, \eg by taking $\alpha = 0$ and $\beta = 1$. In this case, one can easily check that any policy would trivially optimal by Lemma~\ref{lorentz2012bernstein}.}

    Therefore, we have $g(x,\bp) = h(x,\bp)$, and thus the average quality and user welfare can finally be written as
    \begin{align}
        \cW(\bbmu, \bp) &= n\int_0^1 h(x,\bp)^{1+1/\beta}dx,\label{eq:final-welfare}
        \\
        \cQ(\bbmu, \bp) &= \int_0^1 h(x,\bp)^{1/\beta}dx.\label{eq:final-quality}
    \end{align}

    This concludes that our optimization problem can be written as  
    \begin{align*}
    \begin{aligned}
        (\Opt_2)\quad & \underset{\bp \in \Delta_m}{\text{maximize}}
        & & \alpha n\int_0^1 h(x,\bp)^{1+ 1/\beta}dx + (1-\alpha) \int_0^1 h(x,\bp)^{1/\beta}dx  \\
        & \text{subject to}
        & & c(q) = \sum_{i=1}^n p_i\binom{n-1}{i-1}F(q)^{n-i}(1-F(q))^{i-1} \quad \text{ $\forall q \in \supp(F)$}
        \\
        & & & p_1 \ge p_2 \ge \ldots \ge p_n, \; \exists i,j: p_i \neq p_j
        \\
        & & & \sum_{i=1}^n p_i = 1
        \\
        & & & \text{$F$ is c.d.f. of symmetric MNE given $\bp$}
    \end{aligned}
    \end{align*}
    Since the objective function is purely a function over $\bp$, we can simply ignore the first and the last constraints involving $F$ since for any $\bp$, we know that there exists an induced symmetric MNE $\bmu$ and its c.d.f. $F$ that satisfies these constraints.
    Thus, our optimization becomes
    \begin{align*}
    \begin{aligned}
        (\Opt)\quad & \underset{\bp \in \Delta_m}{\text{maximize}}
        & & \alpha n\int_0^1 h(x,\bp)^{1+ 1/\beta} + (1-\alpha) \int_0^1 h(x,\bp)^{1/\beta}  \\
        & \text{subject to}
        & & p_1 \ge p_2 \ge \ldots \ge p_n, \; \exists i,j: p_i \neq p_j
        \\
        & & & \sum_{i=1}^n p_i = 1
    \end{aligned}
    \end{align*}
    and it finishes the proof.
\end{proof}

\subsection{Proof of Proposition~\ref{prop:price-hardmax}}\label{proof:price-hardmax}

\begin{proof}
    If $\alpha = 0$, recall that
    \begin{align*}
        G(\bp) =\int_0^1 h(x,\bp)^{1/\beta} dx.
    \end{align*}
    We will lower bound $G(\bp)$ of $\UNI$ policy and argue that it can be significantly larger than $G(\HM)$ for sufficiently large $n$.

    Recall from Proposition~\ref{thm:single-hm} that
    \begin{align*}
        G(\HM) = \alpha \frac{\beta n}{\beta n + n-1} +  (1-\alpha)\frac{\beta}{\beta + n-1} = \frac{\beta}{\beta + n-1}.
    \end{align*}

    For $\UNI$ policy, observe that 
    \begin{align*}
        h(x,\UNI) &= \sum_{i=1}^n a_i(x)p_i = \sum_{i=1}^{n-1}\frac{a_i(x)}{n-1}.
    \end{align*}
    Since $\sum_{i=1}^n a_i(x) = 1$ by Binomial theorem, we have
    \begin{align*}
        h(x,\UNI) = \frac{1}{n-1} (1-(1-x)^{n-1}).
    \end{align*}
    Thus, we have
    \begin{align*}
        \int_0^1 h(x, \UNI)^{1/\beta} dx
        &= \int_0^1 \left(\frac{1 - (1-x)^{n-1}}{n-1}\right)^{1/\beta} dx
        \\
        &= \int_0^1 \frac{\left(1 - (1-x)^{n-1}\right)^{1/\beta}}{(n-1)^{1/\beta}}dx
        \\
        &\ge \int_{1/2}^{1} \frac{\left(1 - (1-x)^{n-1}\right)^{1/\beta}}{(n-1)^{1/\beta}}dx
        \\
        &\ge \int_{1/2}^{1} \frac{\left(1 - 1/2^{n-1}\right)^{1/\beta}}{(n-1)^{1/\beta}}dx
        \\
        &\ge \frac{1}{2} \cdot \frac{1}{2} \cdot \frac{1}{(n-1)^{1/\beta}}.
    \end{align*}
    As $\beta > 1$, we conclude that there exists a sufficiently large $n$ such that 
    \begin{align*}
        G(\HM) / G(\UNI) < \eps
    \end{align*} for any $\eps > 0$.
\end{proof}

\section{Omitted Proofs in Section~\ref{sec:opt-structure}}
\subsection{Proof of Lemma~\ref{thm:schur}}\label{apd:thm:schur}
We require the following property on the Bernstein basis polynomials:

\begin{lemma}[\citet{doha2011derivatives}]\label{lem:bernstein-diff-deriv} 
    Define a Bernstein basis polynomial \[a_{i,n}(x) = \binom{n}{i}x^{n-i}(1-x)^{i}\,,\]and hence $a_i(x) = a_{i-1,n-1}(x)$.
    Then, for any nonnegative integer $n$ and $i=0,\ldots, n$, it follows that
    \begin{align*}
        (a_i(x) - a_{i+1}(x)) = -\frac{1}{n}\frac{d a_{i,n}(x)}{dx}.
    \end{align*}
\end{lemma}
{
Although the proof can be found in many textbooks~\citep{lorentz2012bernstein}, we provide the proof for completeness.
\begin{proof}
    We will compute the left-hand side (LHS) and right-hand side (RHS) individually to argue its equality.

    First, we evaluate the LHS:
    \begin{align*}
        a_i(x) - a_{i+1}(x) &= \binom{n-1}{i-1}x^{n-i}(1-x)^{i-1} - \binom{n-1}{i}x^{n-i-1}(1-x)^{i} \\
        &= x^{n-i-1}(1-x)^{i-1} \left[ \binom{n-1}{i-1}x - \binom{n-1}{i}(1-x) \right]
    \end{align*}
    Using the identities $\binom{n-1}{i-1} = \frac{i}{n}\binom{n}{i}$ and $\binom{n-1}{i} = \frac{n-i}{n}\binom{n}{i}$, we substitute the coefficients:
    \begin{align}
        a_i(x) - a_{i+1}(x) &= x^{n-i-1}(1-x)^{i-1} \left[ \frac{i}{n}\binom{n}{i}x - \frac{n-i}{n}\binom{n}{i}(1-x) \right] \nonumber \\
        &= \frac{1}{n}\binom{n}{i}x^{n-i-1}(1-x)^{i-1} \left[ ix - (n-i)(1-x) \right] \nonumber\\
        &= \frac{1}{n}\binom{n}{i}x^{n-i-1}(1-x)^{i-1} [ix - (n - nx - i + ix)] \nonumber\\
        &= \frac{1}{n}\binom{n}{i}x^{n-i-1}(1-x)^{i-1} [nx - (n-i)]\,.\label{eq:07011600}
    \end{align}
    
    Next, we evaluate the RHS by taking the derivative of $a_{i,n}(x)$:
    \begin{align*}
        -\frac{1}{n}\frac{d}{dx} a_{i,n}(x) &= -\frac{1}{n}\frac{d}{dx} \left[ \binom{n}{i}x^{n-i}(1-x)^i \right] \\
        &= -\frac{1}{n}\binom{n}{i} \left[ (n-i)x^{n-i-1}(1-x)^i - i x^{n-i}(1-x)^{i-1} \right] \\
        &= -\frac{1}{n}\binom{n}{i}x^{n-i-1}(1-x)^{i-1} \left[ (n-i)(1-x) - ix \right] \\
        &= -\frac{1}{n}\binom{n}{i}x^{n-i-1}(1-x)^{i-1} [n - nx - i + ix - ix] \\
        &= -\frac{1}{n}\binom{n}{i}x^{n-i-1}(1-x)^{i-1} [n-i-nx] \\
        &= \frac{1}{n}\binom{n}{i}x^{n-i-1}(1-x)^{i-1} [nx - (n-i)],
    \end{align*}
    which coincides with Equation~\eqref{eq:07011600}.
\end{proof}
}

Now we are ready to prove the main theorem.
\begin{proof}[Proof of Theorem~\ref{thm:schur}]
    Consider any probability vector $\bp$ in $n$-dimensional probability simplex, which is possibly not sorted.
     Recall a generalized notion of $a_i(x)$ as follows: $a_{i,n}(x) := \binom{n}{i}x^{n-i}(1-x)^{i}$ and 
  note that $a_i(x) = a_{i-1,n-1}(x)$.
    Then, we will use the following well-known fact from Bernstein basis polynomials~\citet{lorentz2012bernstein}, presented in Lemma \ref{lem:bernstein-diff-deriv}:
    \begin{align}
        n(a_i(x) - a_{i+1}(x)) = -\frac{d a_{i,n}(x)}{dx}.\label{eq:diff-to-deriv}
    \end{align}
    {Recall that $O(x, \bp) = \sum_{i=1}^n a_i(x)p_{(i)}$.}
    Let $\sigma(\cdot)$ be the permutation associated with $p_i$'s such that $p_i$ is the $\sigma(i)$-th largest one among $p_1,\ldots, p_n$.
    Consider two indices $i,j$ such that $p_i \ge p_j$, \ie $\sigma(i) < \sigma(j)$.
    Since $p_i \ge p_j$, by Theorem~\ref{thm:schur-ostrowski}, it suffices to prove that
    \begin{align*}
        \frac{\partial\int_0^1  O(x,\bp)^r dx}{\partial p_i} - \frac{\partial\int_0^1  O(x,\bp)^r dx}{\partial p_j} \ge \text{($\le$) } 0,
    \end{align*}
    depending on the parameter $r$.
    Notice that
    \begin{align*}
        \frac{\partial\int_0^1  O(x,\bp)^r dx}{\partial p_i} = \int_0^1 rO(x,\bp)^{r-1}a_{\sigma(i)}(x) dx.
    \end{align*}
    Thus, it suffices to prove that for any index $i,j$ with $\sigma(i) < \sigma(j)$:
    \begin{align*}
        \int_0^1 rO(x,\bp)^{r-1}(a_{\sigma(i)}(x)- a_{\sigma(j)}(x)) dx \ge \text{($\le$) } 0.
    \end{align*}
    We will prove a stronger statement such that for any index $i=1,2,\ldots, n-1$, it follows that
    \begin{align*}
        \int_0^1 rO(x,\bp)^{r-1}(a_i(x)- a_{i+1}(x)) dx \ge \text{($\le$) } 0.
    \end{align*}
    Note that this immediately finishes the proof. Recalling that  $n(a_i(x) - a_{i+1}(x)) = -\frac{d a_{i,n}(x)}{dx}$, the left hand side can be written as 
\[
    \int_0^1 r\, O(x,\bp)^{r-1} (a_i(x) - a_{i+1}(x)) \, dx 
    = -\frac{r}{n} \int_0^1 O(x,\bp)^{r-1} \frac{d}{dx} a_{i,n}(x) \, dx.
\]
Now, apply integration by parts:
\[
    \int_0^1 O(x,\bp)^{r-1} \frac{d}{dx} a_{i,n}(x) \, dx
    = \left[ O(x,\bp)^{r-1} a_{i,n}(x) \right]_0^1 
    - \int_0^1 a_{i+1,n}(x) \frac{d}{dx} O(x,\bp)^{r-1} \, dx.
\]  
    for $i \in [n-1]$.
    Note that $a_{i,n}(x)$ is always zero for $i \in [n-1]$ if $x = 0$ or $x = 1$.
    Further, we have
    \begin{align*}
        \frac{\partial  O(x,\bp)^{r-1}}{\partial x} = (r-1)O(x,\bp)^{r-2}\psi(x,\bp) \ge 0,
    \end{align*}
    by Lemma~\ref{lm:psi-to-h} and $\psi(x,\bp) \ge 0$ due to Lemma~\ref{lm:h-monotone}.
    Thus, $ \int_0^1 rO(x,\bp)^{r-1}(a_i(x)- a_{i+1}(x)) dx$ becomes equivalent to
    \begin{align*}
        \frac{r(r-1)}{n} \int_0^1 O(x,\bp)^{r-2}\psi(x,\bp) a_{i,n}(x) dx,
    \end{align*}
    which is $\ge 0$ if $r \ge 1$ or $r \le 0$, and otherwise $\le 0$ as all the terms $O(x,\bp), \psi(x,\bp)$ and $a_{i,n}(x)$ are nonnegative.
\end{proof}

\subsection{Proof of Theorem~\ref{thm:opt-single}}\label{apd:thm:opt-single}
\begin{proof}
    Note first that the $\Opt$ in Theorem~\ref{thm:prop-user} enforces $p_n = 0$.
    To prove the first condition, if $\alpha =1$, then the objective function becomes $\cW(\bbmu, \bp)$.
   
    From Equation~\eqref{eq:final-welfare}, recall that the user welfare is given by
\[
\cW(\bbmu, \bp) = n \int_0^1 h(x, \bp)^{1 + 1/\beta} dx.
\]
Since \( 1 + 1/\beta \ge 1 \) for any \( \beta > 0 \), the integrand is Schur-convex in \( \bp \) by Theorem~\ref{thm:schur}. Moreover, since \( \HM \) majorizes any \( \bp \in \Delta_n \) with \( p_n = 0 \), this implies that \( \HM \) maximizes welfare among such strategies.

Similarly, from Equation \eqref{eq:final-quality} user quality is given by
\[
\cQ(\bbmu, \bp) = \int_0^1 h(x, \bp)^{1/\beta} dx.
\]
When \( \beta \le 1 \), we again have \( 1/\beta \ge 1 \), and so the integrand is Schur-convex in \( \bp \).  Therefore, when \( \beta \le 1 \), the overall objective—being a sum of two Schur-convex functions—is also Schur-convex. It follows that \( \HM \) is the optimal allocation in this case.



   For the remaining conditions, observe that the objective function reduces to $\cQ(\bbmu, \bp)$ when $\alpha = 0$. If $\beta > 1$, then $1/\beta < 1$, and the integrand becomes Schur-concave in $\bp$. In this case, the uniform policy $\UNI$ is majorized by every $\bp \in \Delta_n$ with $p_n = 0$, and thus $\UNI$ is optimal by the definition of Schur-concave functions.

    
    Finally, if $\alpha = 0$ and $\beta = 1$, then the objective function simply reduces to (up to constant)
    \begin{align*}
        \int_0^1 h(x,\bp) dx =\int_0^1 a_i(x) p_i dx.
    \end{align*}
    Note that 
    \begin{align*}
        \int_0^1 a_i(x) dx = 1/n,
    \end{align*}
    by Lemma~\ref{lorentz2012bernstein} and hence the objective function is always a constant.
    Since the sufficient condition to write the optimization problem is to have {a nontrivial policy such that $p_i \neq p_j$ for some $i,j \in [n]$ and $p_n = 0$}, we finish the proof.
\end{proof}

\section{Omitted Proofs in Section~\ref{sec:opt-structure-general}}
\subsection{Proof of Lemma~\ref{lm:unimodal}}\label{apd:lm:unimodal}
Before proving the main lemma, we require the following property technical property on the monotonicity of $h(x,\bp)$ with respect to $x$.
\begin{lemma}\label{lm:h-monotone}
    $h(x,\bp)$ is strictly increasing in $x \in [0,1]$ for any nontrivial policy $\bp$.
\end{lemma}
\begin{proof}
    Observe that
    $\frac{dh(x,\bp)}{dx} = \sum_{i=1}^n \frac{d a_i(x)}{dx} p_i$.
    We will use the following Abel's transformation formula~\citep{chu2007abel}:
    \begin{align*}
        \sum_{k=1}^n f_k g_k = F_ng_n - \sum_{k=1}^{n-1} F_k(g_{k+1} - g_k),
    \end{align*}
    where $F_k = \sum_{i=1}^k f_i$ for two real sequences $f_k$'s and $g_k$'s.
    Plugging in $f_k = da_k(x)/dx$ and $g_k = p_k$, we obtain
    \begin{align*}
        \sum_{i=1}^n \frac{d a_i(x)}{dx} p_i &= 
        A_n'(x)p_n  - \sum_{k=1}^{n-1} A_k'(x)(p_{k+1} - p_k)
        \\
        &= A_n'(x)p_n + \sum_{k=1}^{n-1} A_k'(x)(p_{k} - p_{k+1}),
    \end{align*}
    for $A_k'(x) = \sum_{i=1}^k a_i'(x)$
    where the derivative is with respect to $x$.
    Noting that $\sum_{i=1}^n a_i(x) = 1$ by binomial theorem, we have $A_n'(x) = 0$.

    To compute $A_k'(x)$, recall the generalized notation $a_{i,n}(x)$ defined in Lemma~\ref{lem:bernstein-diff-deriv}:
    \begin{align}
        a_{i,n}(x) = \binom{n}{i} x^{n-i}(1-x)^{i},\label{eq:temp-equiv}
    \end{align}
    and the equivalence:
    \begin{align*}
        a_{i-1,n-1}(x) - a_{i,n-1}(x) = -\frac{1}{n} \frac{d a_{i,n}(x)}{dx}.
    \end{align*}
    Thus, we have
    \begin{align*}
        a_i'(x) = a_{i-1,n-1}'(x) = (n-1)(a_{i-1,n-2}(x) - a_{i-1,n-2}(x)).
    \end{align*}
    Let us define $a_{-1,k}(x) = 0$ for nonnegative integer $k$, for ease of exposition.
    Note, then, that equivalence~\eqref{eq:temp-equiv} holds for $i=1$ as well.
    By telescoping, this implies:
    \begin{align*}
        A_k'(x) 
        &= \sum_{i=1}^k a_i'(x) 
        \\
        &= \sum_{i=1}^k (n-1)(a_{i-1,n-2}(x) - a_{i-1,n-2}(x))
        \\
        &= (n-1)(a_{k-1,n-2}(x)  - a_{-1,n-2}(x) = a_{k-1,n-2}(x).
    \end{align*}
    Thus, we finally have
    \begin{align*}
        \frac{dh(x,\bp)}{dx} =  \sum_{k=1}^{n-1} a_{k-1,n-2}(x)(p_k - p_{k+1}).
    \end{align*}
    Note that this is strictly positive for $x \in (0,1)$ as we have at least one $k$ such that $p_k > p_{k+1}$ for any nontrivial policy $\bp$.
    This finishes the proof.
\end{proof}

Now the main lemma can be proven in what follows.
\begin{proof}[Proof of Lemma~\ref{lm:unimodal}]
    
    {If $\beta \le 1$, then the proof immediately follows from the fact that $\HM$ is optimal from Theorem~\ref{thm:opt-single}. Thus, we assume $\beta > 1$.}
    
    We will first prove that $K(x,i) = a_i(1-x)$ is $\STP_{n-1}$, where we recall that $ a_i(x) =  \binom{n-1}{i-1}x^{n-i}(1-x)^{i-1}$.\footnote{This is in fact proven by~\cite{schoenberg1959variation}, but we provide here the proof for completeness.}
    That is, for any $k \in [n-1]$, by Definition \ref{def:tp}, it suffices to prove that
    \begin{align*}
        \det 
        \left(\begin{matrix}
            a_{i_1}(1-x_1) & \ldots &a_{i_1}(1-x_k)\\
            \vdots & & \vdots \\
            a_{i_k}(1-x_1) & \ldots &a_{i_k}(1-x_k)
        \end{matrix}
        \right) > 0,
    \end{align*}
    for all $1 \le i_1 < \ldots < i_k \le n-1$ and $0 \le x_1 <\ldots < x_k \le 1$.
    Denote the matrix inside the determinant by $M$ for simplicity.
    
    Consider $r,s \in [k]$.
    Note first that
    \begin{align*}
        a_{i_r}(1-x_s) 
        &= \binom{n-1}{i_r-1} (1-x_s)^{n-1}\parans{\frac{x_s}{1-x_s}}^{i_r-1}.
    \end{align*}
    Define $y_s = x_s/(1-x_s)$, then
    \begin{align*}
        a_{i_r}(1-x_s)  = \binom{n-1}{i_r-1} (1-x_s)^{n-1} y_s^{i_r-1}.
    \end{align*}
    Factoring out the row constants, we have
    \begin{align*}
        \det (M) = \parans{\prod_{r=1}^k \binom{n-1}{i_r-1}}
        \det
        \left(\begin{matrix}
            (1-x_1)^{n-1}y_1^{i_1-1} & \ldots &(1-x_k)^{n-1}y_k^{i_1-1}\\
            \vdots & & \vdots \\
            (1-x_1)^{n-1}y_1^{i_k-1} & \ldots &(1-x_k)^{n-1}y_k^{i_k-1}
        \end{matrix}
        \right)
    \end{align*}
    Factoring out the column constants, we have
    \begin{align*}
        \det(M) = \parans{\prod_{r=1}^k \binom{n-1}{i_r-1}}
        \parans{\prod_{s=1}^k (1-x_s)^{n-1}} 
        \det
        \left(\begin{matrix}
            y_1^{i_1-1} & \ldots &y_k^{i_1-1}\\
            \vdots & & \vdots \\
            y_1^{i_k-1} & \ldots &y_k^{i_k-1}
        \end{matrix}
        \right)
    \end{align*}
    By taking transpose in the remaining matrix and applying Theorem~\ref{thm:vandermonde}, we conclude that $\det(M) > 0$.
    Since this holds for arbitrary choice of $x_1,\ldots, x_k$ and $i_1,\ldots, i_k$, this concludes that $K(x,i)$ is $\STP_{n-1}$.

    {Now we will prove that the diminishing property of totally positive matrix $K(x,i)$ allows us to show the quasi-convexity of the sequence of partial derivatives.}
    {Let us define}  
    $q(x) = f(x) + g(x)$, where
    \begin{align*}
        f(x) &= n\alpha(1+1/\beta) h(x,\bp)^{1/\beta}
        \\
        g(x) &= (1-\alpha)/\beta h(x,\bp)^{1/\beta - 1}.
    \end{align*}
    Then, by Equation \eqref{eq:d_i}, we have 
    \begin{align*}
        d_i &= \int_0^1 q(x)a_i(x) dx
        \\
        &= \int_0^1 q(1-y)a_i(1-y) dy.
    \end{align*}

    By Lemma~\ref{lm:q-quasiconvex}, we know that $q(x)$ is quasiconvex over $x \in [0,1]$.
    This implies that $q(1-y)$ is also quasiconvex on $y \in [0,1]$ because quasiconvexity on a single dimension is preserved under affine transformations of the input variable.

    Thus, by Lemma~\ref{lm:unimodal-to-signchange}, $q(1-y)- \lambda$ has at most two sign changes for every $\lambda \in \R$, and whenever it has exactly two sign changes, it must be $(+,-,+)$.
    
    Recalling the definition of $d_i$ and $K(x,i) = a_i(1-x)$:
    \begin{align*}
        d_i&= \int_0^1 q(1-y)a_i(1-y)dy
        = \int_0^1 K(y,i) q(1-y) dy.
    \end{align*}
    Subtracting $\lambda' \in \R$:
    \begin{align*}
        d_i - \lambda' 
    = - \lambda' + \int_0^1 K(y,i) q(1-y) dy
    = \int_0^1 K(y,i)(q(1-y) -n\lambda')dy,
    \end{align*}
    where the second equality follows from Lemma~\ref{lorentz2012bernstein}, as we show $\int_0^1 a_i(x) dx = 1/n$ for any $i \in [n]$ and $K(x,i) = a_i(1-x)$.
    
    Taking $\lambda = n\lambda'$, we know that $q(1-y) - \lambda$ has at most two sign changes, and the pattern is exactly $(+,-,+)$ if it has exactly two sign changes.
    Further, since we know that $K(y,i)$ is $\STP_{n-1}$, we can apply Theorem~\ref{thm:karlin} as $S^-(q(1-y)) \le 2$ if $n \ge 3$ and conclude that $S^+(d_i-\lambda') \le 2$ for any $\lambda' \in \R$.
    This concludes that $d_i - \lambda'$ always have at most two sign changes assigning zeroes arbitrary signs, and the pattern is $(+,-,+)$ if it has exactly two sign changes.
    
   We will first prove that $d_i = d_{i+1}$ can only happen at the slope-changing index, to be defined shortly, or for $i= n-2$ if $d_i$'s are monotone decreasing without slope-changing index.
   
   Suppose first that there exists an index $k$ such that the slope changes, \ie monotonicity pattern changes.
   Let \( k \) denote the index at which the monotonicity pattern changes, i.e.,
    \[
    d_1 \ge d_2 \ge \ldots \ge d_k \le d_{k+1} \le \ldots \le d_{n-1}.
    \]
    Suppose there exists an index \( i \le k - 2 \) or \( i \ge k + 1 \) such that \( d_i = d_{i+1} \). Then, by choosing \( \lambda' = d_i \), and assigning the signs appropriately,\footnote{Formally, one can assign $(+,-,+,-)$ for $i \le k-2$, and $(-,+,-,+)$ for $i \ge k+1$.} we obtain at least four sign changes in the resulting sequence, implying \( S^+(\lambda') \ge 3 \), which contradicts our assumptions.
    On the other hand, suppose $d_1 \ge d_2 \ge \ldots \ge d_{n-1}$, \ie slope does not change.
    Similarly, if $d_i = d_{i+1}$ for $1 < i < n-1$, then by choosing $\lambda' = d_i$ and similarly assigning signs for zeroes as above, we obtain $S^+(d_i - \lambda') \ge 3$.
    If $i = 1$, then by setting $\lambda' = d_1$, sign changes become $(\cdot, \cdot, -)$, which contradicts as the sign pattern should be $(+,-,+)$ in case of exact two changes.
    Similarly, if the entire sequence is monotone increasing, then $d_i = d_{i+1}$ can only occur if $i= 1$.
    Therefore, \( d_i = d_{i+1} \) can only occur for \( i = k - 1 \) or \( i = k \), or $i = n-1$ if it is monotone decreasing sequence, or $i = 1$ if it is monotone increasing sequence.

	\sscomment{CHECK: What if $n \le 4$? In this case, $n-2 \le 2$, so constant gradients are allowed. This leads to potentially $p_1 > p_2 > p_3$. So, $n=4$ needs to be separately handled.}
    Finally, we will prove that if $k$ is the slope-changing index, then exactly one of $d_{k-1} = d_k$ or $d_{k} = d_{k+1}$ happens.
    Suppose there exists a slope-changing index $k$, and assume \( d_{k-1} = d_k = d_{k+1} \). Setting \( \lambda' = d_k \), we can assign the signs of the zeros at \( d_{k-1} - \lambda' \), \( d_k - \lambda' \), and \( d_{k+1} - \lambda' \) as \( (-, +, -) \), resulting in a sign pattern of the form \( (+, -, +, -, +) \), which again leads to a contradiction.
    Thus, there can be at most two consecutive equal values among \( (d_{k-1}, d_k) \) or \( (d_k, d_{k+1}) \), and the rest of the sequence must be strictly monotonic on either side. This completes the proof.

    \sscomment{In the full version, include a figure for such case study.}

\end{proof}

\subsection{Proof of Lemma~\ref{lm:q-quasiconvex}}
\begin{proof}[Proof of Lemma~\ref{lm:q-quasiconvex}]
    {
    Recall that $\psi(x,\bp)$ is a derivative of $h(x,\bp)$ with respect to $x$ by Lemma~\ref{lm:psi-to-h}.
    }
    Thus, $\psi(x,\bp) = dh(x,\bp)/dx \ge 0$ from Lemma~\ref{lm:h-monotone}. 
    Differentiating $q(x)$ over $x$, we obtain 
    \begin{align*}
        \frac{d q(x)}{dx} 
        &= n\alpha(1+1/\beta)(1/\beta)h(x,\bp)^{1/\beta - 1} \psi(x,\bp) + (1-\alpha)(\frac{1}{\beta} - 1)\frac{1}{\beta}h(x,\bp)^{1/\beta - 2} \psi(x,\bp)
        \\
        &=
        \psi(x,\bp)h(x,\bp)^{1/\beta - 2} \parans{w_1 h(x,\bp) + w_2},
    \end{align*}
    where we write $w_1 =  n\alpha(1+1/\beta)(1/\beta)$ and $w_2 = (1-\alpha)(\frac{1}{\beta} - 1)\frac{1}{\beta}$.
    
    Note that $w_1  \ge 0$ and $w_2 \le 0$ as we have $\beta > 1$.
    Thus $dq(x)/dx$ is nonnegative if and only if $w_1h(x,\bp) + w_2 \ge 0$.
    Since $h(x,\bp)$ is increasing function over $x$ {by Lemma~\ref{lm:h-monotone}}, this concludes that $q(x)$ is either increasing/decreasing function, or decreasing then increasing function, \ie it is quasiconvex.
\end{proof}

\subsection{Proof of Theorem~\ref{thm:opt-general}}\label{apd:thm:opt-general}

\begin{proof}[Proof of Theorem~\ref{thm:opt-general}]
    For a policy $\bp$ to be optimal, the Karush-Kuhn-Tucker (KKT) condition should hold as the constraints are linear.
    We will ignore $p_n$ as it is always zero in the optimum.
    First, the Lagrangian can be written as:
    \begin{align*}
        L(\bp, \nu, \lambda) = G(\bp) + \lambda(1 - \sum_{i=1}^{n-1} p_i) + \sum_{i=1}^{n-2}\nu_i (p_i - p_{i+1}) + \nu_{n-1}p_{n-1},
    \end{align*}
    where $\nu_i \ge 0$ for $i \in [n-1]$ and $\lambda \in \R$.
    Writing down KKT condition, we obtain
    \begin{align*}
        d_1 = \frac{\partial G(\bp)}{\partial p_1} 
        &= \lambda - \nu_1
        \\
        d_i = \frac{\partial G(\bp)}{\partial p_i} 
        &= \lambda + \nu_{i-1} - \nu_i,\;  i=2,3,\ldots, n-1
        \\
        \nu_i(p_i - p_{i+1}) &= 0,\; i=1,2,\ldots, n-2
        \\
        \nu_{n-1}p_{n-1} &= 0
        \\
        \lambda(p_1+\ldots + p_n - 1) &= 0
        \\
        p_1 \ge p_2 \ge \ldots \ge p_{n-1} &\ge 0
        \\
        \nu_i &\ge 0,\; i=1,2,\ldots, n-1
        \\
    \end{align*}

    According to Lemma~\ref{lm:unimodal}, the sequence $(d_1,\ldots, d_{n-1})$  can only have either of the following structures for some $k \in [n-1]$:
    \begin{enumerate}
        \item $d_1 > d_2 > \ldots > d_k \le d_{k+1} < \ldots < d_{n-1}$
        \item $d_1 > d_2 > \ldots > d_{n-2} \ge d_{n-1}$
        \item $d_1 \le d_2 < d_3 < \ldots, d_{n-1}$.
    \end{enumerate}
    Let $k$ be the slope-changing index in the former two cases.

    {
    Overall, given the properties of $d_i$'s above, we will consider two scenarios.
    First we will show that if $p_1 \neq p_2$, then we cannot have $i$ such that $p_i \neq p_{i+1}$ for $i=2,\ldots, n-2$, concluding that $p_1 > p_2=\ldots = p_{n-1}$.
    Then, we will show that if $p_1 = p_2$, then we always have $p_2 = \ldots = p_{n-1}$.
    }

    \textbf{Scenario  1: $\mathbf{p_1 \neq p_2}$.}
    Assume first that $p_1 \neq p_2$.
    Then, we have $\nu_1 = 0$.
    Thus, we have
    \begin{align*}
        d_1 &= \lambda
        \\
        d_2 &= \lambda - \nu_2.
    \end{align*}
    {We will deal with two cases where $p_2 = p_3$ or $p_i = p_{i+1}$ for $i=3,\ldots, n-2$ separately.}
   
  \noindent\textbf{Scenario  1a: $\mathbf {p_2 \neq p_3}$.} Suppose, for the sake of contradiction, that $p_2 \neq p_3$. Then, we must have $\nu_2 = 0$, which implies
\begin{align*}
    d_1 &= \lambda, \\
    d_2 &= \lambda, \\
    d_3 &= \lambda - \nu_3.
\end{align*}
This implies that the slope-changing index is $k = 1$. However, this leads to a contradiction: we have {$d_3 = \lambda - \nu_3 \le \lambda = d_2$} since $\nu_3 \ge 0$, but since $k=1$ is the slope-changing index, it should have been $d_2 < d_3$. Therefore, we conclude that if $p_1 \neq p_2$, then it must be that $p_2 = p_3$. 

\noindent\textbf{Scenario  1b: $\mathbf {p_i \neq p_{i+1}}$ for $\mathbf{i\ge  3}$.} Now, again for the sake of contradiction, suppose that the first index $i$ such that $p_i\ne p_{i+1}$ is greater than or equal  to $3$. That is,
    \begin{align*}
        p_1 \neq p_2 = p_3 = \ldots = p_i \neq p_{i+1},
    \end{align*}
    for some $i \ge 3$.
    Then $\nu_i = 0$ and we have the following   series of equations:
    \begin{align*}
        d_1 &= \lambda
        \\
        d_2 &= \lambda - \nu_2
        \\
        d_3 &= \lambda + \nu_2 - \nu_3
        \\
        \vdots
        \\
        d_i &= \lambda + \nu_{i-1}
        \\
        d_{i+1} &= \lambda - \nu_{i+1}.
    \end{align*}
    If $\nu_{i+1} \neq 0$ and $\nu_{i-1} \neq 0$, then it is a contradiction as the $d$ sequence increases then decreases.
    Thus, either of $\nu_{i-1}$ or $\nu_{i+1}$ is zero.
    
    If $\nu_{i+1} = 0$, then we should again have $\nu_{i-1} = 0$ as otherwise the $d$ sequence increases from {$d_1 = \lambda$ to $d_i = \lambda + \nu_{i-1}$, then decreases to $d_{i+1} = \lambda - \nu_{i+1}$.} 
    Then, notice that by setting $\nu_{i+1} = \nu_{i-1}= 0$, we have
    \begin{align*}
        d_1 &= \lambda
        \\
        d_2 &= \lambda - \nu_2
        \\
        d_3 &= \lambda + \nu_2 - \nu_3
        \\
        \vdots
        \\
        d_{i-1} &= \lambda + \nu_{i-2}
        \\
        d_i &= \lambda
        \\
        d_{i+1} &= \lambda,
    \end{align*}
    which is a contradiction as $i$ should be the slope-changing index and thus it should be the case that $\lambda = d_1 > d_2 > \ldots > d_i = d_{i+1} = \lambda$, but we have $d_1 = \lambda \le d_{i-1}  = \lambda + \nu_{i-2} \ge d_i = d_{i+1}$. 
    
    Therefore, we have $\nu_{i-1} = 0$.
    In this case, again, we have
    \begin{align*}
        d_1 &= \lambda
        \\
        d_2 &= \lambda - \nu_2
        \\
        d_3 &= \lambda + \nu_2 - \nu_3
        \\
        \vdots
        \\
        d_i &= \lambda
        \\
        d_{i+1} &= \lambda - \nu_{i+1},
    \end{align*}
    and thus it should be the case that $\lambda = d_1 > d_2 > \ldots \ge d_k < \ldots < d_i = \lambda < d_{i+1}$ since slope-changing needs to occur strictly before the index $i$, which is a contradiction as $\lambda = d_i \ge d_{i+1} = \lambda - \nu_{i+1}$.
    Thus, if $p_1 \neq p_2$, we always have $p_1 \neq p_2 = \ldots = p_{n-1} \ge p_n = 0$.

    \textbf{Scenario  2: $\mathbf {p_1 = p_2}$.}
    {
    Now the remaining case is when $p_1 = p_2$.
    Similarly, we will separately deal with two cases: (i) when $p_i \neq p_{i+1}$ for some $i=3,\ldots, n-2$, and (ii) when $p_2 \neq p_3$.
    }
    Let $i \ge 2$ be the first index such that $p_i \neq p_{i+1}$, \ie $\nu_i = 0$, for $i \in [n-2]$.

    \noindent\textbf{Scenario  2a: $\mathbf {p_1 = p_2>\ldots > p_i = p_{i+1}}$ for $i =3,\ldots, n-2$.}
    In this case, we have
    \begin{align*}
        d_1 &= \lambda - \nu_1
        \\
        d_2 &= \lambda + \nu_1 - \nu_2
        \\
        d_3 &= \lambda + \nu_2 - \nu_3
        \\
        \vdots
        \\
        d_{i-1} &= \lambda + \nu_{i-2} - \nu_{i-1}
        \\
        d_i &= \lambda + \nu_{i-1}
        \\
        d_{i+1} &= \lambda - \nu_{i+1}.
    \end{align*}
    For these equations to be true under the condition on $d_i$'s, it should be the case that either of  $\nu_{i-1}$ and $\nu_{i+1}$ is zero as otherwise it increases from $d_1$ then decreases to $d_{i+1}$.
    If $\nu_{i-1} = 0$, we have $d_{i-1} = \lambda + \nu_{i-2}$, which is again a contradiction as we have $d_1 \le d_{i-1} \ge d_i \ge d_{i+1}$.
    If $\nu_{i+1} = 0$, then we should have $d_i = d_{i+1}$ as otherwise it increases from $d_1 = \lambda - \nu_1$ then decreases to $\lambda$.
    Thus, $d_i = d_{i+1} = \lambda$.
    Again, in this case, we have $d_{i-1} = \lambda + \nu_{i-2}$, which is a contradiction.

    \noindent\textbf{Scenario  2b: $\mathbf {p_1 = p_2 > p_3}$.}
    In this case, $\nu_2 = 0$ and we have
    \begin{align*}
        d_1 &= \lambda - \nu_1
        \\
        d_2 &= \lambda + \nu_1
        \\
        d_3 &= \lambda - \nu_3
        \\
        d_4 &= \lambda + \nu_3 - \nu_4
    \end{align*}
    Thus, it should be the case that $\nu_1 = 0$ as otherwise $d_1 < d_2 \ge d_3$.
    Then, the resulting equations are:
    \begin{align*}
        d_1 &= \lambda
        \\
        d_2 &= \lambda
        \\
        d_3 &= \lambda - \nu_3
        \\
        d_4 &= \lambda + \nu_3 - \nu_4,
    \end{align*}
    and thus the slope-changing index should be $k = 1$.
    Thus, after $k+1= 2$, it should always strictly increase, which is a contradiction to $d_3 = \lambda - \nu_3$.
    Therefore, if $p_1 = p_2$, we always have $p_1 = p_2 = \ldots = p_{n-1}$.

    Combining, the only possible cases are:
    \begin{align*}
        p_1 > p_2 =\ldots =p_{n-1} \ge p_n = 0
        \\
        p_1 = p_2 = \ldots = p_{n-1} \ge p_n = 0,
    \end{align*}
    and it finishes the proof.


\end{proof}

\section{Omitted Proofs in Section~\ref{sec:computation}}\label{apd:computation}

Before presenting the proofs, we need the following technical lemma, the proof of which is omitted as it is an elementary algebra.
\begin{lemma}\label{lem:power}
For $r\ge1$ and $a,b\in[0,1]$, $|a^r-b^r|\le r\,|a-b|$. 
For $r\in(0,1]$ and $a,b\ge0$, $|a^r-b^r|\le |a-b|^r$.
\end{lemma}

\subsection{Proof of Lemma~\ref{lem:computation-validity}}
\begin{proof}[Proof of Lemma~\ref{lem:computation-validity}]
For each fixed $x$, $p_1\mapsto h(x;p_1)=c_0(x)+c_1(x)p_1$ is affine and $t\mapsto t^{r}$ is increasing
for $r>0$; hence $\sup_{p_1\in[\ell,u]} h(x;p_1)^{r}=\max\{h(x;\ell)^r,h(x;u)^r\}$.
Integrating in $x$, and summing the two exponents $r\in\{1+1/\beta,\,1/\beta\}$ with weights
$\alpha n$ and $1-\alpha$, gives the upper bound. The lower bound follows since $L(I)$ is the
objective at boundary points.
\end{proof}

\subsection{Proof of Lemma~\ref{lem:gap}}
\begin{proof}[Proof of Lemma~\ref{lem:gap}]
Let $p^\diamond\in\{\ell,u\}$ be such that $L(I)=G(p^\diamond)$ and let $p^\circ$ be the other endpoint.
By the definition of $U(I)$ and $L(I)$, linearity of the integral, and the inequality
$\max\{A,B\}-A\le |B-A|$ (valid for all real $A,B$), we have
\begin{align*}
U(I)-L(I)
&= \alpha n\!\int_0^1\!\!\Big(\max\{h(x;p^\circ)^{1+1/\beta},h(x;p^\diamond)^{1+1/\beta}\}-h(x;p^\diamond)^{1+1/\beta}\Big)\,dx\\
&\quad+(1-\alpha)\!\int_0^1\!\!\Big(\max\{h(x;p^\circ)^{1/\beta},h(x;p^\diamond)^{1/\beta}\}-h(x;p^\diamond)^{1/\beta}\Big)\,dx\\
&\le \alpha n\!\int_0^1\!\big|h(x;p^\circ)^{1+1/\beta}-h(x;p^\diamond)^{1+1/\beta}\big|\,dx\\
&\qquad+(1-\alpha)\!\int_0^1\!\big|h(x;p^\circ)^{1/\beta}-h(x;p^\diamond)^{1/\beta}\big|\,dx.
\end{align*}
Define $\Delta_r(x):=\big|h(x;p^\circ)^r-h(x;p^\diamond)^r\big|$ for $r>0$. Then
\begin{align}
U(I)-L(I)\ \le\ \alpha n\!\int_0^1\!\Delta_{1+1/\beta}(x)\,dx+(1-\alpha)\!\int_0^1\!\Delta_{1/\beta}(x)\,dx.
\label{eq:ULtoDelta}
\end{align}

Because $h(x;p_1)\in[0,1]$ and $p_1\mapsto h(x;p_1)$ is affine,
the increment across the interval $I=[\ell,u]$ is
\[
h(x;u)-h(x;\ell)=c_1(x)\,(u-\ell)=c_1(x)\,|I|.
\]
For $r\ge1$ and $a,b\in[0,1]$, $|a^r-b^r|\le r\,|a-b|$ and for $r\in(0,1]$ and $a,b\ge0$, $|a^r-b^r|\le |a-b|^r$ by Lemma~\ref{lem:power}.
Applying these with $(a,b)=(h(x;u),h(x;\ell))$ yields
\[
\Delta_{1+1/\beta}(x)\le \Big(1+\tfrac1\beta\Big)\,|c_1(x)|\,|I|,
\qquad
\Delta_{1/\beta}(x)\le |c_1(x)|^{1/\beta}\,|I|^{1/\beta}.
\]

Substituting the bounds for $\Delta_{1+1/\beta}$ and $\Delta_{1/\beta}$ into \eqref{eq:ULtoDelta},
\[
U(I)-L(I)\ \le\ \alpha n\Big(1+\tfrac1\beta\Big)\!\int_0^1\!|c_1(x)|\,dx\;|I|
\ +\ (1-\alpha)\!\int_0^1\!|c_1(x)|^{1/\beta}\,dx\;|I|^{1/\beta},
\]
which is the claimed inequality with $C_1$ and $C_2$ as stated.
\end{proof}

\subsection{Proof of Lemma~\ref{prop:invariants}}
\begin{proof}[Proof of Lemma~\ref{prop:invariants}]
(i) $L^\star$ is the maximum of objective values at feasible endpoints, hence a valid global lower
bound. (ii) Let $I^\star$ be any interval (among the descendants of $I_0$) that contains a global
maximizer; Lemma~\ref{lem:computation-validity} gives $U(I^\star)\ge\sup_{p\in I^\star}G(p)\ge\mathrm{OPT}$, and
$I^\star$ remains in $\mathcal A$ until pruned or split, so $\max_{I\in\mathcal A}U(I)\ge U(I^\star)\ge\mathrm{OPT}$.
\end{proof}

\subsection{Proof of Theorem~\ref{thm:rate}}
\begin{proof}[Proof of Theorem~\ref{thm:rate}]
Let $I^\star$ be any active interval containing a maximizer. If $|I^\star|\le\delta(\varepsilon)$, then by
Lemma~\ref{lem:gap}, $U(I^\star)-L(I^\star)\le\varepsilon$. Since $\mathrm{OPT}\le U(I^\star)$ and
$L^\star\ge L(I^\star)$, we get $\mathrm{OPT}-L^\star\le \varepsilon$. The loop exits only when
$\max_{I\in\mathcal A}U(I)\le L^\star+\varepsilon$; combined with Proposition~\ref{prop:invariants},
\[
\mathrm{OPT}\le \max_{I}U(I)\le L^\star+\varepsilon\quad\Rightarrow\quad G(p^\star)=L^\star\ge \mathrm{OPT}-\varepsilon.
\]
If $|I^\star|>\delta(\varepsilon)$, it cannot be pruned because Proposition~\ref{prop:invariants}
gives $U(I^\star)\ge\mathrm{OPT}>L^\star+\varepsilon$. Hence it must be split. After $k$ bisections
of the interval containing the maximizer, the length is $\le D\,2^{-k}$. The stated bound on $k$
ensures $D\,2^{-k}\le\delta(\varepsilon)$, completing the proof.
To see the time complexity, note that a uniform partition of $\mathcal D$ with mesh $\delta(\varepsilon)$ has $O(D/\delta(\varepsilon))$
intervals, hence the total number of nodes explored by the algorithm is
\[
O\!\Big(\tfrac{C_1}{\varepsilon}\Big)\ +\ O\!\Big(\tfrac{C_2}{\varepsilon}\Big)^{\!\beta},
\]
dominated by the second term when $\beta>1$, which is polynomial in $1/\eps$ given constant $\beta$.
\end{proof}

\subsection{Proof of Lemma~\ref{lem:riemann}}
\begin{proof}[Proof of Lemma~\ref{lem:riemann}]
For a monotone increasing $f:[0,1]\to\mathbb{R}$, the upper–minus–lower Riemann sum on the uniform grid $\{j/m\}_{j=0}^m$ equals
$\frac{1}{m}\sum_{j=1}^{m}[f(j/m)-f((j-1)/m)]=\frac{f(1)-f(0)}{m}$.
The integral lies between the two sums, proving the bound.
Monotonicity of $f_r$ follows from Lemma~\ref{lm:h-monotone} and $r>0$.
\end{proof}

\section{Omitted Proofs in Section~\ref{sec:general}}\label{apd:general-obj}
\subsection{Proof of Theorem~\ref{thm:general-obj}}
\begin{proof}[Proof of Theorem~\ref{thm:general-obj}]
    We first prove the general statement for arbitrary polynomial with nonnegative cofficients.
    It suffices to prove that the resulting sequence of partial derivatives with respect to $p_i$'s has quasiconvexity by Lemma~\ref{lm:unimodal} under the new objective function $\cG(\bmu, \bp)$.
    
    Since we deal with symmetric MNE $\bmu$, we will simply drop the index $i$ in $q_i$ and $\mu_i$.
    Using~\eqref{eq:03092207},\footnote{Note that one can analogously prove $p_n = 0$, and we omit the detail due to its redundancy.} for any $j \in [n]$, we have
    \begin{align*}
        \Exu{q \sim \mu}{e_jq^{k_j}}
        &= \int_{q \in \supp(F)} e_jq^{k_j} dF(q)
        \\
        &= 
        \int_{q \in \supp(F)} e_j\parans{\parans{\sum_{i=1}^n p_i \binom{n-1}{i-1} F(q)^{n-i}(1-F(q))^{i-1}}^{{k_j}}}^{1/\beta}dF(q)
        \\
        &=
        \int_{q \in \supp(F)} e_j\parans{\sum_{i=1}^n p_i \binom{n-1}{i-1} F(q)^{n-i}(1-F(q))^{i-1}}^{{k_j}/\beta}dF(q)
        \\
        &=
        \int_0^1 e_jh(x,\bp)^{{k_j}/\beta}dx.
    \end{align*}
    
    Note that the partial derivative $d_i$ of the objective function with respect to $p_i$ can be written as:
    \begin{align*}
        d_i = 
        \frac{\partial \cG(x,\bp)}{\partial p_i} 
        &= \frac{\partial}{\partial p_i} \sum_{j=1}^m \int_0^1 e_jh(x,\bp)^{{k_j}/\beta}dx
        \\
        &=
        \sum_{j=1}^m \int_0^1 \frac{j}{\beta} e_jh(x,\bp)^{{k_j}/\beta-1}a_i(x) dx
        \\
        &=
        \int_0^1 \parans{\sum_{j=1}^m \frac{j}{\beta}e_jh(x,\bp)^{{k_j}/\beta - 1}}a_i(x) dx.
    \end{align*}
    Note that the same term $\sum_{j=1}^m \frac{k_j}{\beta}e_jh(x,\bp)^{k_j/\beta - 1}$ remains the same for every $i \in [n]$.

    Given $\bp$, define $q(x) = \sum_{j=1}^m \frac{k_j}{\beta}e_jh(x,\bp)^{k_j/\beta - 1}$.
    We will prove that $q(x)$ is quasiconvex, \ie it either decreases, increases, or decreases then increases.
    Note that this concludes the proof.

    Let us compute the derivative $q'(x)$:
    \begin{align*}
        q'(x) &= \frac{d}{dx} \left( \sum_{j=1}^{m} e_j \frac{k_j}{\beta} h(x,\bp)^{\frac{k_j}{\beta} - 1} \right) 
        \\
        &= \sum_{j=1}^{m} e_j \frac{k_j}{\beta} \left(\frac{k_j}{\beta} - 1\right) h(x,\bp)^{\frac{k_j}{\beta} - 2} \cdot \frac{dh(x,\bp)}{dx}.
    \end{align*}
    We can factor out the common terms and write:
    \begin{align*}
        q'(x) = \frac{dh(x,\bp)}{dx}h(x,\bp)^{-2} \parans{ \sum_{j=1}^{m} e_j \frac{k_j}{\beta} \left(\frac{k_j}{\beta} - 1\right) h(x,\bp)^{\frac{k_j}{\beta}} }.
    \end{align*}
    By Lemma~\ref{lm:h-monotone}, we know that $h'(x,\bp) \ge 0$ for any $x \in [0,1]$.
    Further, $h(x,\bp) \ge 0$ for any $x \in (0,1)$.
    Thus, it suffices to prove that the following quantity
    \begin{align}
        \sum_{j=1}^{m} e_j \frac{k_j}{\beta} \left(\frac{k_j}{\beta} - 1\right) h(x,\bp)^{\frac{k_j}{\beta}},\label{eq:07082306}
    \end{align}
    changes its sign at most once, from $-$ to $+$, as this implies that the original function $q(x)$ is decreases then increases, or just monotone.
    
    Define $c_j = e_j \frac{k_j}{\beta}(\frac{k_j}{\beta}-1)$.
    Noting that $h(x,\bp)$ is strictly increasing and nonnegative by Lemma~\ref{lm:h-monotone}, we will simply write $z = h(x,\bp)^{1/\beta}$.
    Then, Equation~\eqref{eq:07082306} simply becomes a posynomial over $z$:
    \begin{align*}
        P(z) = \sum_{j=1}^m c_jz^{k_j}.
    \end{align*}
    By Descartes' rule of signs, the number of positive real roots of $P(z)$ is less than or equal to the number of sign changes in the sequence of its coefficients $(c_1,\ldots, c_m)$.\footnote{We note that Descartes' rule is proven to be true with real exponents by~\cite{curtiss1918recent}.}
    
    Given that $e_j \ge 0$ for all $j \in [m]$, the sign of $c_j$ is entirely determined by the $e_j\frac{k_j}{\beta}(\frac{k_j}{\beta} - 1)$.
    Note that this changes from $-$ to $+$ based on our assumption.
    In particular, when $e_j$'s are all nonnegative, this is positive or negative depending on whether $k_j/\beta - 1 > 0$ or $<0$, \ie whether $k_j > \beta$ or $k_j < \beta$.
    Further, if $c_j \ge 0$, this implies that $j \ge \beta$ and thus $c_{k} \ge 0$ for any $k \ge j$.
    Combining, we conclude that the sequence of coefficients $(c_1,\ldots, c_m)$ changes its sign from $-$ to $+$.
    Thus, the posynomial $P(z)$ has at most one positive real root by Descartes' rule of signs.
    Further, if we take limit $x \to 0+$, we have $z = h(x,\bp)^{1/\beta}\to p_n = 0$.
    In this case, the term with the smallest power of $z$ will dominate in $P(z)$.
    Thus, the sign changes of $P(z)$ can only happen from $-$ to $+$.
    This concludes that $q(x)$ is quasiconvex.

\end{proof}

\subsection{Proof of Corollary~\ref{cor:social-welfare}}
\begin{proof}[Proof of Corollary~\ref{cor:social-welfare}]
    Recall that under the symmetric MNE $\bmu$, we have
    \begin{align*}
        \cW(\bmu, \bp) &= n\int_0^1 h(x,\bp)^{1/\beta + 1}dx
        \\
        \cV(\bmu, \bp) &= 
        \sum_{j=1}^m \int_0^1 h(x,\bp)^{j/\beta}dx
        \\
        \cU(\mu, \bmu_{-i},\bp) &= \frac{1}{n} - \Exu{q \in \mu}{q^\beta}
        \\
        &= \frac{1}{n} - n\int_0^1 h(x,\bp)dx,
    \end{align*}
    where the last equation follows from~\eqref{eq:03092207}.
    Notice, however, that $\int_0^1h(x,\bp)dx = \frac{\sum_{i=1}^n p_i}{n} = \frac{1}{n}$,
    by Lemma~\ref{lorentz2012bernstein}.
    Thus, it remains to show that
    \begin{align*}
        \int_0^1 h(x,\bp)^{1/\beta + 1}dx + \sum_{j=1}^m \int_0^1 h(x,\bp)h(x)^{j/\beta}dx
    \end{align*}
    satisfies the condition stated in Theorem~\ref{thm:general-obj}, which is straightforward since all the coefficients are nonnegative.
    This concludes the proof.
\end{proof}

\begin{proof}[Proof of Corollary~\ref{cor:order-statistics}]
    Let $F_{(1)}$ be the c.d.f. of the largest order statistics among $q_i$'s at symmetric MNE $\bmu$ with c.d.f. $F$.
    Then, we have
    \begin{align*}
        \Exu{q_i \sim \mu_i, i\in[n]}{q_{(1)}}
        &= \int_{q \in \supp(F)} (1-F_{(1)(q)})dF(q)
        \\
        &= \int_{q \in \supp(F)} (1-F(q)^n)dF(q)
        \\
        &=
        \int_0^1  (1-x^n)\psi(x,\bp) h(x,\bp)^{1/\beta - 1}dx
        \\
        &=
        n\int_0^1 x^{n-1}h(x,\bp)^{1/\beta}dx,
    \end{align*}
    where the derivations are analogous to the proof of Theorem~\ref{thm:prop-user}.

    Taking a partial derivative with respect to $p_i$, we obtain
    \begin{align*}
        n\int_0^1x^{n-1}h(x,\bp)^{1/\beta}a_i(x) dx.
    \end{align*}
    Thus, we take $q(x) = nx^{n-1}h(x,\bp)^{1/\beta}$ to apply Lemma~\ref{lm:unimodal}.
    Then, notice that
    \begin{align*}
        \frac{dq(x)}{dx} 
        &= n(n-1)x^{n-2}h(x,\bp)^{1/\beta} + \frac{n}{\beta}x^{n-1}h(x,\bp)^{1/\beta - 1}\psi(x,\bp)
        \\
        &=
        nx^{n-2}h(x,\bp)^{1/\beta - 1} \parans{(n-1) h(x,\bp) + x/\beta \psi(x,\bp)}.
    \end{align*}
    Since $\psi(x,\bp) \ge 0, h(x,\bp) \ge 0$, this is always nonnegative, and thus quasiconvex.
    Thus, we can apply Lemma~\ref{lm:unimodal} and finish the proof.
\end{proof}

\subsection{Proof of Corollary~\ref{cor:exp-utility}}
\begin{proof}[Proof of Corollary~\ref{cor:exp-utility}]
Let us define
\[
\Phi(\bp, \mu)\ :=\ \Exu{q \sim \mu}{e^{\lambda q}},
\]
for the symmetric MNE strategy $\mu$ given the policy $\bp$.
For $M\in \N$, we define the truncated posynomial
\[
\Phi_M(\bp, \mu)\ :=\ \Ex{\sum_{j=0}^{M}\frac{\lambda^j}{j!}q^{j}}
 = \sum_{j=0}^{M}\frac{\lambda^j}{j!}\Ex{q^{j}}.
\]
For each $M$, $\Phi_M$ is a finite linear combination of moments $\Ex{q^j}$ with nonnegative
coefficients $e_j:=\lambda^j/j!$. 
Ordering the exponents as $k_j=j$, the sign pattern
$e_j(k_j-\beta)$ exhibits at most one sign change as $j$ increases (it is negative for $j<\beta$,
nonnegative for $j\ge \beta$). Hence Theorem~\ref{thm:general-obj} applies and yields a maximizer $\bp^{(M)}$ which is piecewise constant with at most two value changes at $p_1$ or $p_n$.

Since $q\in[0,1]$, the Lagrange remainder of the exponential series in its Taylor series gives, uniformly
over $q\in[0,1]$,
\[
\Big|e^{\lambda q}-\sum_{j=0}^{M}\frac{\lambda^j}{j!}\,q^{j}\Big|
\ \le\ \frac{e^{\lambda}\,(\lambda q)^{M+1}}{(M+1)!}
\ \le\ \frac{e^{\lambda}\,\lambda^{M+1}}{(M+1)!}.
\]
Thus, we obtain
\[
\sup_{\bp \in \Delta} \big|\Phi(\bp, \mu)-\Phi_M(\bp, \mu)\big|
 \le \sup_{q\in[0,1]}\Big|e^{\lambda q}-\sum_{j=0}^{M}\frac{\lambda^j}{j!}\,q^{j}\Big|
\le \frac{e^{\lambda}\lambda^{M+1}}{(M+1)!}.
\]

The sequence
$\{\bp^{(M)}\}$ lies in $\Delta$, which is a compact feasible set of policies, so it has a convergent subsequence $\bp^{(M_k)}\to \bp^\star$.
Since we know that $\Phi_M\to\Phi$ uniformly on $\Delta$, so
\[
\Phi\big(\bp^{(M_k)}\big)\ \ge\ \sup_{\bp \in \Delta}\Phi_{M_k}(\bp)-\frac{e^{\lambda}\lambda^{M_k+1}}{(M_k+1)!}
\ \ge\ \Phi_{M_k}(\bp)-\frac{e^{\lambda}\lambda^{M_k+1}}{(M_k+1)!}\,,
\]
Hence, letting $k\to\infty$ yields $\Phi(\bp^\star)\ge \Phi(\bp)$ for all $\bp \in \Delta$, \ie $\bp^\star$ is optimal for $\Phi$.
Since $\bp^\star$ is a convergent sequence of $\bp^{(M)}$ with the structure $p_1 \ge p_2 = \ldots = p_{n-1} \ge p_n = 0$, the resulting sequence $\bp^\star$ also inherits the same structure, as desired.

Finally, it is immediate to extend the analysis to any convex combination of exponential utilities with different $\lambda$'s, since the uniform convergence analogously carries over and the structural result also follows as the sign of $e_j(k_j - \beta)$ only depends on whether $k_j \ge \beta$ or not.
\end{proof}

\end{document}